\documentclass[a4paper,11pt]{article}
\pdfoutput=1

\usepackage{xcolor}
\usepackage{jheppub}
\usepackage{mathtools,bm,amsthm}
\usepackage{booktabs,array,enumitem,needspace}
\hypersetup{colorlinks=true,linkcolor=red!45!black,citecolor=blue!55!black,urlcolor=blue!70!black}

\newcommand{\Tr}{\operatorname{Tr}}
\newcommand{\Pf}{\operatorname{Pf}}
\newcommand{\Hs}{\mathcal H_{\sigma}}
\newcommand{\Hu}{\mathcal H_{1}}
\newcommand{\Dsig}{D_{\sigma}}
\newcommand{\Dpsi}{D_{\psi}}
\newcommand{\Leff}{L_{\mathrm{eff}}}

\newtheorem{proposition}{Proposition}
\newtheorem{theorem}{Theorem}
\newtheorem{corollary}{Corollary}

\theoremstyle{definition}

\newtheorem{remark}{Remark}

\title{Physical reduced states and continuum characters\\
 of the lattice Kramers--Wannier defect}
\author{Yi Liang}
\affiliation{Independent Researcher, Beijing, China}
\emailAdd{iantrans2042@gmail.com}
\arxivnumber{2607.01137}
\abstract{
A non-invertible topological line fixes global defect data but does not select
a reduced density matrix (RDM) in a spatially twisted sector.  For the
Kramers--Wannier defect of the critical Ising chain, we choose the equal-weight
incoherent mixture of the two charge-sector ground states and reduce it to the
ordinary full spin algebra of a complete non-wrapping prefix anchored at the
defect endpoint.  The finite-size RDM is the convex average of two
charge-sector RDMs rather than a sign-zero Gaussian proxy, and an exact
descended antiunitary enforces many-body Kramers pairing without determining
the entropy.  At fixed prefix, the RDMs converge in trace norm to a common Gaussian
state as the circumference grows.  Convergence is uniform over all
ground-doublet preparations, including coherent pure states.  The subsequent
large-prefix entropy limit follows from the odd--even
Toeplitz/Fisher--Hartwig endpoint, yielding half the natural logarithm of
two above the homogeneous chain.  The matched spin-flip defect has zero excess.
The exact joint energy--modified-translation character of the same twisted
Hamiltonian records the finite-size roots and their multiplicities, retains
chirality in the marked scaling limit, and, using the standard Ising character
identities, resolves the four Virasoro towers of the Ising duality-twisted
sector.
}

\begin{document}
\setcounter{tocdepth}{2}
\maketitle

\section{Introduction}
\label{sec:introduction}

Symmetry constrains dynamics, organizes quantum sectors, and relates dual
descriptions of the same physics~\cite{BhardwajEtAl2024}.  Generalized
symmetry extends the ordinary group notion to operators of different
codimensions and to fusion structures that need not form a group.  In two
dimensions, non-invertible symmetries are implemented by topological lines
with non-group-like fusion; they constrain operator sectors, correlation
functions, and partition functions
\cite{Shao2023TASI,SchaferNameki2024,BhardwajTachikawa2018}.  Whether this
global algebraic structure leaves a signature in local reduced states, and
how such a signature appears in their entanglement, is not settled by the line
data alone.  The Kramers--Wannier defect of the critical Ising chain allows
both questions to be addressed by explicit construction.

The critical Ising model makes an exact treatment possible from the lattice to
the continuum.  Its Kramers--Wannier (KW)
duality exchanges the high- and low-temperature descriptions (equivalently,
the order and disorder variables) of the two-dimensional Ising model
\cite{KramersWannier1941I,KramersWannier1941II}.  Onsager's solution determines
its exact critical point and thermodynamics~\cite{Onsager1944}.  At
criticality, KW duality is carried by a topological line in the Ising minimal
model
\cite{BelavinPolyakovZamolodchikov1984,FrohlichFuchsRunkelSchweigert2004}.
Its fusion rule $\sigma\times\sigma=1+\psi$ makes the line non-invertible.
Composing the duality with itself produces the identity and fermion sectors
rather than a single inverse.  This branching of sectors makes the KW line a
sharp test of whether, and how, non-group-like fusion is reflected in local
entanglement structure.

Rational conformal field theory (RCFT) describes Ising boundaries and conformal
defects while organizing twisted partition functions and interface fusion
\cite{Ishibashi1989,Cardy1989,PetkovaZuber2001,BachasBrunner2008,FrohlichFuchsRunkelSchweigert2004,FrohlichFuchsRunkelSchweigert2007,OshikawaAffleck1996,OshikawaAffleck1997}.
On the lattice, the same fusion algebra has been realized together with defect
motion, modified translation, and the duality-twisted spectrum
\cite{AasenMongFendley2016,AasenFendleyMong2020,Grimm2002}.  Continuum and
lattice constructions have thus established the algebraic and spectral
identity of the KW defect.  These global data do not, however, specify the
physical reduced density matrix (RDM) of a selected state in the spatially
twisted sector or the local algebra on which it is defined.

Universal CFT results characterize interval entanglement
\cite{CalabreseCardy2004,CalabreseCardy2009}, but constructing a defect RDM
also requires choosing the orientation of the line.  That choice changes the
quantum problem.  When the line winds around the spatial circle, it acts as a defect operator
on the untwisted Hilbert space.  A line running along Euclidean time instead
defines a spatially twisted Hilbert space
\cite{AasenMongFendley2016,AasenFendleyMong2020}.  This Hilbert space fixes the
kinematics, but not the state or the subsystem.  The global charge selects the
physical sector, and a zero-mode kernel requires a prescription for the state
within it.  The subsystem is fixed separately by the local spin algebra
assigned to the entangling region.  These choices cannot be bypassed by
passing to free fermions.  A principal
submatrix of the fermionic covariance matrix need not represent the RDM of a
spin subsystem.  Likewise, a zero entanglement level at finite size does not
establish a universal entropy constant, which requires control of the full
reduced state.  The Ising KW regulator exposes both obstructions.  Its
free-fermion description contains a localized spectator together with a
delocalized zero mode.  Identifying the defect is only the starting point.  One
must still determine the physical state and subsystem that realize its
entanglement structure.

Categorical symmetry and gauging set the two-dimensional continuum framework
for defects and their fusion~\cite{BhardwajTachikawa2018}.  The passage from
this algebra to microscopic degrees of freedom is a separate problem.  In the
Majorana/Ising chain it takes the form of non-invertible translations
\cite{SeibergShao2024}.  Tensor-network formulations reach the lattice by a
different route, encoding duality in matrix-product-operator (MPO)
intertwiners whose composition reproduces fusion.  An MPO represents a spin-chain
operator by contracting one operator-valued tensor per site along shared
auxiliary indices
\cite{WilliamsonEtAl2016,BultinckEtAl2017,CiracEtAl2021,LootensDelcampOrtizVerstraete2023}.
Once the defect action has been fixed microscopically, its effect on a
subsystem becomes a well-posed question.  Entropy at conformal
interfaces~\cite{SakaiSatoh2008}, fusion of critical Ising defect
lines~\cite{BachasBrunnerRoggenkamp2013}, entropy for Ising defects
\cite{RoySaleur2022}, and both entropy and negativity for them
\cite{Rogerson2022} have been studied from complementary angles.
For defects, the full reduced-state problem has only recently come into view.
Rockwood works
directly with periodic free-fermion chains, where the KW zero mode generates
nonlocal couplings in the microscopic entanglement Hamiltonian
\cite{Rockwood2025}.  Northe and Rossi begin instead in CFT.  Their
defect-dressed RDMs expose entanglement-spectrum structure and reproduce
R\'enyi and von Neumann entropies~\cite{NortheRossi2025}.  The CFT construction
has since been extended to states prepared with two topological interfaces
\cite{NorthePolettiRossi2026Interfaces}.  Together, these
works provide the nearest lattice and continuum points of comparison for the
reduced-state problem addressed here.

Symmetry-resolved entanglement asks how a reduced state decomposes under
conserved charges
\cite{GoldsteinSela2018,XavierAlcarazSierra2018,BonsignoriRuggieroCalabrese2019,MurcianoDiGiulioCalabrese2020}.
Recent CFT work generalizes this decomposition to categorical and twisted
sectors~\cite{Saura2024,Das2024,ChoiRayhaunZheng2024,HeymannQuella2025}.
These constructions isolate symmetry-labelled components of reduced states or
entanglement observables.  Here we determine the complete physical RDM and its
entanglement spectrum for a selected state in the spatial KW sector, with the
subsystem defined by the local spin algebra.  We then resolve this same RDM
by local spin parity.  The construction of the physical reduced state thus
leads to its symmetry-resolved description, while the global preparation
remains distinct from the local parity decomposition.

The spatially wound defect operator encodes the global non-invertible fusion
algebra, while a line directed along Euclidean time defines the spatially
twisted Hamiltonian
\cite{AasenMongFendley2016,AasenFendleyMong2020}; neither construction by
itself specifies a local physical density matrix.  We therefore prescribe the
charge-unresolved preparation as the equal-weight incoherent mixture of the two
charge-sector ground states and assign the ordinary full spin matrix algebra to
a non-wrapping interval of complete spins starting at site~1 and shorter than
the full chain, which we call a \emph{proper prefix}.  Its physical RDM is the
ordinary convex average of the two charge-sector RDMs.  Each sector separately
admits a Gaussian representation
\cite{ChungPeschel2001,CheongHenley2004,Peschel2003,PeschelEisler2009}, within
the broader theory of critical-chain entanglement
\cite{VidalLatorreRicoKitaev2003}; by contrast, the sector-derived sign-zero
covariances are auxiliary comparators and do not define the finite-size
physical mixture.  Figures~\ref{fig:kw-quantization} and~\ref{fig:defect-rdm-dictionary} summarize the
passage from the global defect algebra to the selected spatial state and its
complete-prefix spin subsystem.

This physical RDM has an exact finite-size structure.  On every complete
proper prefix, a descended antiunitary acts within the physical reduced state
and enforces many-body Kramers pairing of its spectrum.  The pairing belongs
to the charge-unresolved spin RDM itself; it is distinct from an auxiliary
single-particle zero entanglement level.  The two local-parity sectors have
equal weights and identical normalized spectra, so the parity-inserted
charged moments vanish.  Conditioning on either sector lowers every
positive-index R\'enyi entropy, including the von Neumann limit, by $\log 2$
relative to this same unresolved RDM.  Separately, the finite-size
physical mixture cannot be identified with either of the two sector-derived
sign-zero Gaussian proxies.  These statements separate the physical
many-body degeneracy from a convenient covariance-level comparison, but they
do not by themselves fix the entropy constant.

The entropy follows from control of the complete reduced state.  At fixed
prefix size, the two charge-sector RDMs converge separately in trace norm to
the same limiting Gaussian state, and the cross-sector reduced operator
vanishes.  Every ground-doublet preparation, including coherent pure states,
therefore approaches this same local RDM.  The resulting comparison
with the homogeneous chain reduces to neighboring odd and even
critical-Majorana Toeplitz blocks; their common critical contribution cancels,
and the independent Fisher--Hartwig endpoint supplies the finite increment.
In the strict order $L\to\infty$ first and prefix size $m\to\infty$ second,
this gives
\begin{equation}
 \lim_{m\to\infty}\left\{
 \lim_{\substack{L\to\infty\\L>m}}
 \left[S\!\left(\rho_{{\rm can},A_m}^{(L)}\right)
      -S\!\left(\rho_{{\rm hom},A_m}^{(L)}\right)\right]\right\}
 =\frac12\log2,
 \label{eq:intro-main-theorem}
\end{equation}
The same ordered excess holds for every ground-doublet preparation, whereas
the matched invertible $\eta$ defect has zero shift under the same convention.  In the Affleck--Ludwig $g$-function language, this constant is the
KW defect entropy
\cite{AffleckLudwig1991,SakaiSatoh2008,BrehmBrunnerJaudSchmidtColinet2016,GutperleMiller2016}.
It is determined by the ordered full-RDM and Toeplitz analysis, not by Kramers
pairing alone.

A separate spatial diagnostic resolves the continuum content of the same
spatially twisted lattice realization.  Energy together with modified translation gives an
exact finite-size joint character that records which translation roots
actually occur at each excitation energy and with what multiplicities.  After
centering the root labels, the eventual no-wrap scaling limit preserves the
sign of momentum and hence chirality.  The resulting Virasoro decomposition
contains the four primary pairs $(\sigma,1)$, $(\sigma,\epsilon)$,
$(1,\sigma)$, and $(\epsilon,\sigma)$.  Thus the spatial spectrum reconstructs
the four chiral towers of the Ising duality-twisted sector rather than merely
matching energy degeneracies.  The character theorem uses the spectrum and
modified translation of the twisted Hamiltonian, while the entanglement
statements use the chosen preparation and its physical prefix RDM.  MPO fusion fixes the global defect species but is not used in either spatial
derivation.

Section~\ref{sec:microscopic} establishes the temporal--spatial dictionary and
the microscopic regulator.
Section~\ref{sec:entanglement} constructs the physical complete-prefix RDM,
distinguishes it from the auxiliary proxies, and proves its finite-size
Kramers structure.  Section~\ref{sec:entanglement-entropy} derives the ordered
entropy theorem, and Section~\ref{sec:cft} obtains the joint
energy--translation character and its four-tower scaling limit.
Section~\ref{sec:discussion} closes by separating the state--subsystem logic
that can transfer beyond Ising from the ingredients specific to the present
regulator and geometry.

\section{Microscopic realization of the defect sector}
\label{sec:microscopic}

The lattice calculation takes place in a single Euclidean geometry with
spatial coordinate $x$ on a circle, Euclidean time $\tau$, and quantization on
constant-$\tau$ slices.  A topological line placed in that geometry may run in either of two
directions, and the two choices pose different questions.  A line wound around
the spatial circle lies inside a single time slice, so it enters as an operator
$\Dsig:\Hu\to\Hu$ on the untwisted Hilbert space, and the non-invertible
fusion algebra is a statement about that operator.  A line running along
$\tau$ instead punctures every slice at one point, and the Hilbert space of a
slice becomes the twisted space $\Hs$, on which the spatial-chain Hamiltonian
$H^{\mathrm d}$ and the modified translation $T_\sigma$ act
\cite{FrohlichFuchsRunkelSchweigert2004,BachasBrunnerRoggenkamp2013,AasenMongFendley2016}.
The Hamiltonian furnishes the states we shall reduce; $T_\sigma$ labels their
spatial quantum numbers.  Both constructions descend from the same Ising line
$\sigma$, yet their operators are attached to different Hilbert spaces, and
nothing requires them to agree as finite matrices.

Our order of presentation follows that division.  The spatially wrapped line
operator comes first, in a lattice convention fixed once and for all; the
duality-twisted Hamiltonian produced by a time-directed line comes second.
Keeping them apart keeps the fusion data of the defect apart from the two
further inputs that a reduced density matrix needs, namely a choice of state
and a choice of retained subsystem -- a separation drawn equally sharply in
recent CFT treatments of defect-dressed reduced states
\cite{NortheRossi2025,NorthePolettiRossi2026Interfaces}.
Figure~\ref{fig:kw-quantization} shows the two quantizations and their constant-time slices.
Appendix~\ref{app:explicit} collects the conventions in full.

\begin{figure}[t]
 \centering
 \includegraphics[width=\textwidth]{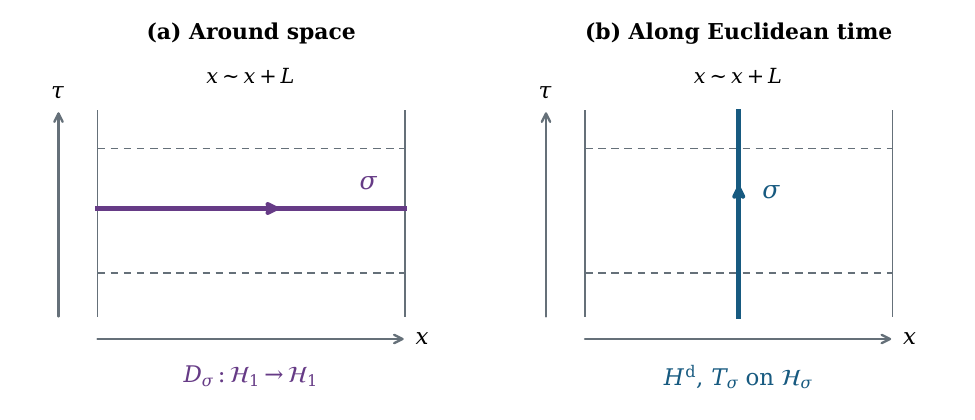}
 \caption{\textbf{Two quantizations of the KW line.}  The vertical edges
 of each strip are identified, $x\sim x+L$.  In (a), the defect wraps around
 the spatial circle at fixed Euclidean time and acts as
 $D_\sigma:\mathcal H_1\to\mathcal H_1$.  In (b), the line runs along
 Euclidean time and crosses every quantization slice, giving the spatially
 twisted space $\mathcal H_\sigma$, on which $H^{\rm d}$ and $T_\sigma$
 act.  Dashed lines mark constant-time slices, not defects.}
 \label{fig:kw-quantization}
\end{figure}

\subsection{Spatially wrapped line operator and non-invertible algebra}\label{sec:microscopic-mpo}

An MPO, or matrix product operator, writes an operator on the whole spin chain
as a contraction of one operator-valued tensor per site along shared auxiliary
indices; the range of those indices is the bond dimension, which for the KW
line is two.  Contracted around the spatial circle, the tensor below reproduces
the KW fusion rule already established on the lattice and in the continuum
\cite{AasenMongFendley2016,FrohlichFuchsRunkelSchweigert2004}.  We display it
because doing so pins down the normalization and the orientation of the virtual
legs before any twisted-sector state is discussed.  The tensor is $W_{\alpha\beta}^{st}$, where
$\alpha,\beta\in\{0,1\}$ are virtual indices and $s,t\in\{0,1\}$ label output
and input spin states.  Its nonzero entries are
\begin{align}
 W_{00}^{s0}&=\frac1{\sqrt2},
 &W_{01}^{s1}&=\frac{(-1)^s}{\sqrt2},\label{eq:MPO-W1}\\
 W_{10}^{s0}&=\frac{(-1)^s}{\sqrt2},
 &W_{11}^{s1}&=\frac1{\sqrt2}.
 \label{eq:MPO-W2}
\end{align}
On a periodic chain, closing the virtual indices in the orientation fixed in
Eq.~\eqref{eq:MPO-contraction} gives the spatially wrapped defect operator,
\begin{equation}
 \langle s_1\cdots s_L|\Dsig|t_1\cdots t_L\rangle
 =\Tr_{\mathrm v}\!
 \left(W^{s_Lt_L}\cdots W^{s_2t_2}W^{s_1t_1}\right),
 \label{eq:MPO-contraction}
\end{equation}
where $\Tr_{\mathrm v}$ traces over the virtual indices.
This reverse-written product establishes the virtual-leg orientation used
throughout.  Using the same
component array in the order $W_1W_2\cdots W_L$ would transpose that
orientation and reverse the neighboring index in the closed kernel.  Ring
reflection relates the two forms, but they cannot be mixed when deciding
whether the fusion factor is $T$ or $T^{-1}$.
The invertible defect is
\begin{equation}
 \Dpsi=\prod_{j=1}^{L}X_j,
 \label{eq:Dpsi}
\end{equation}
and $T$ translates the spin configuration one site to the right.

Contracting the virtual indices in Eq.~\eqref{eq:MPO-contraction} gives a
closed expression in the computational basis.  Write the spin configurations
as binary vectors $s=(s_1,\ldots,s_L)$ and $t=(t_1,\ldots,t_L)$ in
$\mathbb F_2^L$, where $\mathbb F_2=\{0,1\}$ is the two-element field with
addition modulo~2 and $\mathbb F_2^L$ is the set of $L$-component binary
vectors.  Take all site indices periodically.  Then
\begin{equation}
 \langle s|\Dsig|t\rangle
 =2^{-L/2}(-1)^{B(s,t)},
 \qquad
 B(s,t)=\sum_{j=1}^{L}s_j(t_j+t_{j+1})
 \quad(\mathrm{mod}\ 2).
 \label{eq:MPO-closed-kernel}
\end{equation}
The closed form follows directly from the stated reverse virtual contraction.
For fixed physical indices, a nonzero tensor entry requires its second virtual
index to equal $t_j$ and contributes
$2^{-1/2}(-1)^{s_j(\alpha_j+t_j)}$.  In
Eq.~\eqref{eq:MPO-contraction}, the periodic contraction identifies
$\alpha_j=t_{j+1}$, and multiplying the local phases gives
Eq.~\eqref{eq:MPO-closed-kernel}.  If the matrix product is instead written in
the opposite order while retaining the same component array, the exponent is
$\sum_js_j(t_{j-1}+t_j)$ and the translation orientation is reversed.  Direct
local-tensor contractions for $2\le L\le6$ agree entry by entry with
Eq.~\eqref{eq:MPO-closed-kernel}, with maximum difference below $10^{-14}$.

\begin{proposition}[KW fusion for arbitrary circumference]\label{prop:MPO-fusion}
For every $L\ge2$, the operator defined by
Eq.~\eqref{eq:MPO-contraction} obeys
\begin{align}
 \Dpsi\Dsig&=\Dsig\Dpsi=\Dsig,
 \label{eq:MPO-absorption}\\
 [T,\Dsig]&=0,
 \label{eq:MPO-translation}\\
 \Dsig^2&=T^{-1}(1+\Dpsi),
 \label{eq:MPO-fusion}\\
 \operatorname{rank}\Dsig&=2^{L-1}.
 \label{eq:MPO-rank}
\end{align}
\end{proposition}

\begin{proof}
Adding the all-one vector to either argument of $B(s,t)$ changes its exponent by
an even sum, which proves the two absorption identities.  Simultaneously
translating $s$ and $t$ only relabels the periodic sum, proving
Eq.~\eqref{eq:MPO-translation}.

For the square, insert Eq.~\eqref{eq:MPO-closed-kernel} and sum over the
intermediate binary vector $t$ to obtain
\begin{equation}
 \langle s|\Dsig^2|u\rangle
 =2^{-L}\sum_{t\in\mathbb F_2^L}
 (-1)^{B(s,t)+B(t,u)}.
 \label{eq:fusion-fourier-sum}
\end{equation}
The coefficient of $t_j$ in the exponent is
\begin{equation}
 s_j+s_{j-1}+u_j+u_{j+1}.
 \label{eq:fusion-coefficient}
\end{equation}
Orthogonality of characters of $\mathbb F_2^L$ makes the sum vanish unless all
these coefficients are zero.  Defining $d_j=s_j+u_{j+1}$ reduces the condition
to $d_j=d_{j-1}$, so there are precisely two branches,
\begin{equation}
 s_j=u_{j+1}+c,
 \qquad c\in\{0,1\}.
 \label{eq:fusion-branches}
\end{equation}
For either branch the sum in Eq.~\eqref{eq:fusion-fourier-sum} is $2^L$ and
the normalized matrix element is one.  The branch $c=0$ is $T^{-1}$ in our
translation convention; $c=1$ is $T^{-1}\Dpsi$.  This proves the fusion
identity.

The same character-orthogonality argument, now summing the product of a kernel
and its transpose, gives
\begin{equation}
 \Dsig^\dagger\Dsig
 =\Dsig\Dsig^\dagger
 =1+\Dpsi.
 \label{eq:MPO-normality}
\end{equation}
Indeed, the sum over the output configuration is nonzero precisely when the two
input configurations differ by a constant bit.  Thus $\Dsig$ is normal for
every $L$, has singular value $\sqrt2$ on the even spin-flip sector, and
vanishes on the odd sector.

Finally, absorption shows that $\Dsig$ annihilates the odd $D_\psi$ sector.  On
the even sector, Eq.~\eqref{eq:MPO-fusion} gives
$\Dsig^2=2T^{-1}$, which is invertible.  Hence $\Dsig$ is injective on the
even sector and zero on the odd sector, proving
$\operatorname{rank}\Dsig=2^{L-1}$.
\end{proof}

\begin{corollary}[Exact MPO spectral support]\label{cor:MPO-spectrum}
Choose a common eigenvector of $T$ and $D_\psi$ with
$T|p,q\rangle=e^{ip}|p,q\rangle$ and
$D_\psi|p,q\rangle=q|p,q\rangle$.  If $q=-1$, then
$D_\sigma|p,q\rangle=0$.  If $q=+1$, every $D_\sigma$ eigenvalue on that
translation sector obeys
\begin{equation}
 \lambda_\sigma^2=2e^{-ip},
 \qquad |\lambda_\sigma|=\sqrt2.
 \label{eq:MPO-eigenvalues}
\end{equation}
\end{corollary}

\begin{proof}
Equations~\eqref{eq:MPO-absorption} and \eqref{eq:MPO-fusion} give the zero
branch and the squared eigenvalue, while Eq.~\eqref{eq:MPO-normality} fixes
its modulus.  Since $D_\sigma$ is normal and commutes with $T$, the simultaneous
spectral decomposition exists.
\end{proof}

The fusion algebra established above identifies the spatially wrapped operator as the lattice KW line.  Its defining duality
action is seen by applying the same kernel to the Ising family
\begin{equation}
 H(J,h)=-J\sum_jZ_jZ_{j+1}-h\sum_jX_j,
 \label{eq:H-Jh}
\end{equation}
where $J$ is the nearest-neighbor Ising coupling and $h$ is the transverse
field.
The translation factor and its orientation convention are the lattice
half-translation accompanying KW fusion~\cite{AasenMongFendley2016}.  The same
kernel gives local operator relations.  Flipping the input bit $t_j$
changes $B(s,t)$ by $s_{j-1}+s_j$, whereas flipping the output bit $s_j$
changes it by $t_j+t_{j+1}$.  Therefore
\begin{equation}
 \Dsig X_j=Z_{j-1}Z_j\Dsig,
 \qquad
 X_j\Dsig=\Dsig Z_jZ_{j+1}.
 \label{eq:MPO-local-duality}
\end{equation}
Using periodicity to relabel the first bond sum, Eq.~\eqref{eq:MPO-local-duality}
implies, for every $L\ge2$,
\begin{equation}
 \Dsig H(J,h)=H(h,J)\Dsig,
 \qquad
 [\Dsig,H(1,1)]=0.
 \label{eq:MPO-duality}
\end{equation}
Thus both the fusion identity and the Hamiltonian duality intertwiner are
all-circumference analytic statements for the tensor in
Eqs.~\eqref{eq:MPO-W1}--\eqref{eq:MPO-W2}.

The all-size identities in Proposition~\ref{prop:MPO-fusion} and
Eq.~\eqref{eq:MPO-duality} were also assembled as full spin matrices for
$L=3,\ldots,8$.  The maximum relative residual was
$6.67\times10^{-16}$, every rank equaled $2^{L-1}$, and
Eq.~\eqref{eq:MPO-normality} agreed to the same precision.  These calculations
check the implementation and translation convention without entering the
analytic proof.

The line-operator computation has fixed the global fusion algebra.
Constructing a reduced density matrix instead starts with a defect running
along Euclidean time and its Hamiltonian $H^{\mathrm d}$ on $\Hs$.  One must
also choose a state of that Hamiltonian and a local algebra retained on a
region.  The Hamiltonian is constructed in the present section;
the state and the retained algebra are chosen in
Section~\ref{sec:entanglement}.

\subsection{Spatially twisted Hamiltonian and fixed-charge Majoranas}\label{sec:microscopic-spatial}

A line pointing along Euclidean time meets each spatial slice at a single
point, so in the twisted sector the defect sits at one fixed location on the
circle of $L$ sites.  With Pauli matrices $X_j,Y_j,Z_j$ on those sites, the
homogeneous critical Ising Hamiltonian is
\begin{equation}
 H^{\mathrm{hom}}
 =-\sum_{j=1}^{L}X_j-\sum_{j=1}^{L}Z_jZ_{j+1},
 \qquad Z_{L+1}\equiv Z_1.
 \label{eq:Hhom}
\end{equation}
In the dressed convention used throughout, the duality-twisted Hamiltonian reads
\begin{equation}
 H^{\mathrm d}
 =-\sum_{j=1}^{L-1}Z_jZ_{j+1}
  -\sum_{j=2}^{L}X_j
  -Z_LY_1.
 \label{eq:Hdef}
\end{equation}
The seam closes a twisted spatial circle.  The omission of the transverse
field at site~1 is part of the KW defect construction, shared by the
duality-twisted chains of
Refs.~\cite{Grimm2002,AasenMongFendley2016,RoySaleur2022}.
One Majorana, $a_1$, decouples as a result, as the fixed-charge rewriting below
exhibits directly.

The comparison with Ref.~\cite{RoySaleur2022} can be drawn term by term.  Its
positive-defect Hamiltonian and charge read
\begin{equation}
 H_{\rm R}=-\frac12\left(\sum_{j<L}X_jX_{j+1}
              +\sum_{j<L}Z_j+Y_LX_1\right),\qquad
 Q_{\rm R}=\prod_{j=1}^L Z_j.
 \label{eq:roy-H-convention}
\end{equation}
Let $h$ denote the one-site Hadamard gate, $\mathsf R$ the reversal of the site
order, and $\mathcal U=Z_1\mathsf R h^{\otimes L}$.  Then
\begin{equation}
 H^{\mathrm d}=2\mathcal U H_{\rm R}\mathcal U^\dagger,
 \qquad
 \mathcal U Q_{\rm R}\mathcal U^\dagger=-\Omega.
 \label{eq:roy-spin-dictionary}
\end{equation}
The Hadamard rotation interchanges $X$ and $Z$ and reverses the sign of $Y$.
Reversing the sites transports the missing field from site~$L$ to site~1,
and conjugation by $Z_1$ supplies the dressed seam sign.  To express our
selected state in the reference convention, we use the inverse transformation
$\mathcal U^\dagger$, which sends our interval $\{1,\ldots,m\}$ to
$\{L-m+1,\ldots,L\}$ and reverses the global charge label.  An entropy
comparison must therefore carry both the state and the retained region
through this dictionary, as in Section~\ref{sec:entanglement-entropy}.

The fermionic description used below is the standard Jordan--Wigner
transformation \cite{LiebSchultzMattis1961,Pfeuty1970}, with every sign
rederived in our Pauli and right-translation conventions.

The charge operator is
\begin{equation}
 \Omega=\prod_{j=1}^{L}X_j,
 \qquad \Omega=\pm1
 \label{eq:Omega}
\end{equation}
within a fixed sector.  We introduce Majoranas
\begin{equation}
 a_j=\left(\prod_{k<j}X_k\right)Z_j,
 \qquad
 b_j=i\left(\prod_{k<j}X_k\right)Z_jX_j,
 \label{eq:JW-majoranas}
\end{equation}
ordered as
\begin{equation}
 \gamma=(a_1,b_1,a_2,b_2,\ldots,a_L,b_L)^{\mathsf T}.
 \label{eq:gamma-order}
\end{equation}
Our quadratic convention is
\begin{equation}
 H=\frac{i}{4}\gamma^{\mathsf T}A\gamma,
 \qquad A^{\mathsf T}=-A\in\mathbb R^{2L\times2L}.
 \label{eq:quadratic-convention}
\end{equation}

The remaining ordered Majoranas
\begin{equation}
 (b_1,a_2,b_2,\ldots,a_L,b_L)
 \label{eq:odd-cycle-order}
\end{equation}
are governed, in the sign convention of Eq.~\eqref{eq:Hdef}, by
\begin{equation}
 H_\Omega^{\mathrm d}
 =i\sum_{j=1}^{L-1}b_ja_{j+1}
 +i\sum_{j=2}^{L}a_jb_j
 +i\Omega b_Lb_1.
 \label{eq:odd-cycle-H}
\end{equation}

\begin{theorem}[Fixed-charge graph]\label{prop:odd-cycle}
In a fixed $\Omega$ sector, Eq.~\eqref{eq:Hdef} is exactly equal to
Eq.~\eqref{eq:odd-cycle-H}.  The Majorana $a_1$ is a spectator and the sequence
in Eq.~\eqref{eq:odd-cycle-order} forms an odd cycle of length $2L-1$.
\end{theorem}

\begin{proof}
The definitions in Eq.~\eqref{eq:JW-majoranas} give
\begin{equation}
 -X_j=i a_jb_j,
 \qquad
 -Z_jZ_{j+1}=i b_ja_{j+1}.
 \label{eq:JW-local-identities}
\end{equation}
They also give $b_1=-Y_1$.  Multiplication by the fixed charge yields
\begin{equation}
 \Omega b_L=-iZ_L,
 \qquad
 Z_L=i\Omega b_L,
 \label{eq:JW-boundary-identity}
\end{equation}
and hence
\begin{equation}
 -Z_LY_1=i\Omega b_Lb_1.
 \label{eq:JW-defect-closure}
\end{equation}
No term contains $a_1$, since the transverse-field sum begins at site~2, the first
bond contains $b_1a_2$, and the closure contains $b_Lb_1$.  The remaining
terms join the sequence in Eq.~\eqref{eq:odd-cycle-order} consecutively and
Eq.~\eqref{eq:JW-defect-closure} closes the cycle.  Reversing the dressed
defect sign or the charge reverses the closing orientation but not the
single-particle spectrum.
\end{proof}

The full antisymmetric matrix is $2L$ dimensional and has nullity two.  One null
direction is the spectator $a_1$, and the other is the zero mode of the odd
cycle.  The distinction between the full $2L$-dimensional representation and its
$2L-1$ active cycle is important.  We do not remove the spectator from a
physical complete-site subsystem.

The spatial defect is mobile under a modified translation $T_\sigma$.  In the
dressed convention used here it is the Clifford circuit
\begin{equation}
 T_\sigma=V\,T\,\bigl(H_L\,\mathrm{CZ}_{L,1}\bigr)V^\dagger,
 \qquad
 V=e^{-i\pi Z_1/4},
 \label{eq:Tsigma-circuit}
\end{equation}
where $H_L$ is the Hadamard gate on site~$L$,
$\mathrm{CZ}_{L,1}$ is the controlled-$Z$ gate on sites $L$ and~1, $T$
translates one site to the right, and matrix products act from right to left.  Conjugation cycles the $2L-1$ positive local terms of
$-H^{\mathrm d}$ as
\begin{align}
 Z_jZ_{j+1}&\longmapsto Z_{j+1}Z_{j+2}
 &&(1\le j\le L-2),\nonumber\\
 Z_{L-1}Z_L&\longmapsto Z_LY_1
 \longmapsto X_2,
 \label{eq:Tsigma-term-cycle}\\
 X_j&\longmapsto X_{j+1}
 &&(2\le j\le L-1),\nonumber\\
 X_L&\longmapsto Z_1Z_2.
 \nonumber
\end{align}
Thus, for every $L\ge2$,
\begin{equation}
 T_\sigma^\dagger T_\sigma=1,
 \qquad
 [T_\sigma,H^{\mathrm d}]=0.
 \label{eq:modified-translation}
\end{equation}
This is the spin-space form, in our dressing and right-translation convention,
of the lattice modified translation/Dehn twist discussed in
Refs.~\cite{AasenMongFendley2016,HauruEtAl2016}.  Related lattice
non-invertible translations and their Majorana/Ising interpretation are
developed in Ref.~\cite{SeibergShao2024}; no identification of their operator
convention with Eq.~\eqref{eq:Tsigma-circuit} is assumed here.  The cycle has
$2L-1$ steps, fixing the effective circumference
\begin{equation}
 \Leff=L-\frac12.
 \label{eq:Leff}
\end{equation}
This half-integer shift is structural rather than fitted.  It is a property of
the duality-twisted spectrum obtained already in Ref.~\cite{Grimm2002}, and it
also enters the numerical entropy comparisons of Ref.~\cite{RoySaleur2022}.
In the present derivation it follows from the $2L-1$ active Majoranas alone,
before any ground state has been selected.  The pure charge-sector states and
the mixture built from them consequently share one effective circumference,
even though their reduced states, and the entropies computed from those, need
not agree.  Section~\ref{sec:cft} uses $\Leff$ in the energy--momentum
spectrum; Section~\ref{sec:entanglement-entropy} keeps that use apart from the
different question of which reference state an entropy is subtracted against.

The fixed-charge odd-cycle Fock spectrum agrees with the spin-sector spectrum
of Eq.~\eqref{eq:Hdef} for $L=2,4,6,8$, with maximum discrepancy below
$4\times10^{-14}$.  Equations~\eqref{eq:Tsigma-circuit}--\eqref{eq:modified-translation}
are all-size Clifford identities, and explicit matrices through $L=8$ agree at
machine precision.  The spatial realization has therefore supplied a
fixed-charge Hamiltonian, its spectator-plus-odd-cycle graph, and the modified
translation fixing its effective circumference.  These global objects now
provide the state and subsystem data needed to construct the physical RDM in
the next section.

\section{Defect reduced density matrix and entanglement spectrum}\label{sec:entanglement}

A defect Hamiltonian and a global charge, of the kind fixed in
Section~\ref{sec:microscopic}, still leave a reduced density matrix
undetermined.  One must also specify which state of the twisted chain is
prepared and which algebra of observables an interval retains.
Continuum treatments of defect-dressed reduced states make the same two
choices explicitly~\cite{NortheRossi2025,NorthePolettiRossi2026Interfaces}.
For the finite-size results of this section we use the equal-weight incoherent mixture
of the two physical charge-sector ground states, reduced to the full spin
algebra of a complete non-wrapping interval anchored at site~1 -- and that
pair of choices already has an exact finite-size consequence.  An antiunitary
sector exchange descends through the ordinary partial trace and pairs the
whole spectrum of the reduced state.  Read in the language of symmetry
resolution, the same mechanism makes the parity-inserted moments vanish and
renders the two normalized parity blocks isospectral.  The Gaussian machinery
needed for the later entropy limit is collected last, since it applies
sectorwise rather than to the mixture itself.

\subsection{The prepared state and the algebra retained by the interval}
\label{sec:entanglement-states}

Fix the global preparation first.  With the global charge
$
 \Omega=\prod_{j=1}^{L}X_j
$
and projectors $P_\omega=(I+\omega\Omega)/2$, let
$\rho_{\omega,L}$ be the unique ground-state projector in the sector
$\Omega=\omega$, $\omega=\pm1$.  Uniqueness is here the fixed-charge statement
that the zero-mode occupation is selected rather than doubled;
Eq.~\eqref{appE:zero-selector} derives the exact fermionic selector.  The
charge-unresolved state studied in this section is
\begin{equation}
 \rho_{{\rm can},L}=\frac12\bigl(\rho_{+,L}+\rho_{-,L}\bigr).
 \label{eq:physical-global-ensemble}
\end{equation}
This prescribed equal-charge preparation is invariant under the exact
sector-exchange antiunitary recalled below.

For the entropy comparison we also allow general preparations within this
same ground-state doublet.  Writing
$\rho_{\omega,L}=|\psi_{\omega,L}\rangle\langle\psi_{\omega,L}|$, define
\begin{equation}
 \rho_L(p,c)=p\rho_{+,L}+(1-p)\rho_{-,L}
 +c|\psi_{+,L}\rangle\langle\psi_{-,L}|
 +\bar c|\psi_{-,L}\rangle\langle\psi_{+,L}|,
 \quad 0\leq p\leq1,\quad |c|^2\leq p(1-p).
 \label{eq:general-ground-preparation}
\end{equation}
The parameters may depend on $L$.  Pure states saturate the positivity
bound; the prescribed incoherent mixture is $p=1/2,c=0$.  The finite-size
pairing and local-parity claims below concern that prescribed point, not
the entire family.

The retained region is $A_n=\{1,\ldots,n\}$ with $1\leq n<L$, and every spin
operator on each of its sites is kept; such a non-wrapping interval, anchored
at the first site of the fixed site order, will be called a complete prefix.
Its observable algebra is then the full matrix algebra
$B((\mathbb C^2)^{\otimes n})$, on which no superselection centre is imposed,
so the physical reduced state is the ordinary partial trace
\begin{equation}
 \rho_{{\rm can},A_n}^{(L)}
 =\operatorname{Tr}_{A_n^c}\rho_{{\rm can},L}
 =\frac12\left(\rho_{+,A_n}^{(L)}+\rho_{-,A_n}^{(L)}\right).
 \label{eq:physical-prefix-rdm}
\end{equation}

Figure~\ref{fig:defect-rdm-dictionary} relates the fixed-charge regulator to
this complete-prefix algebra.
\begin{figure}[t]
 \centering
 \includegraphics[width=\textwidth]{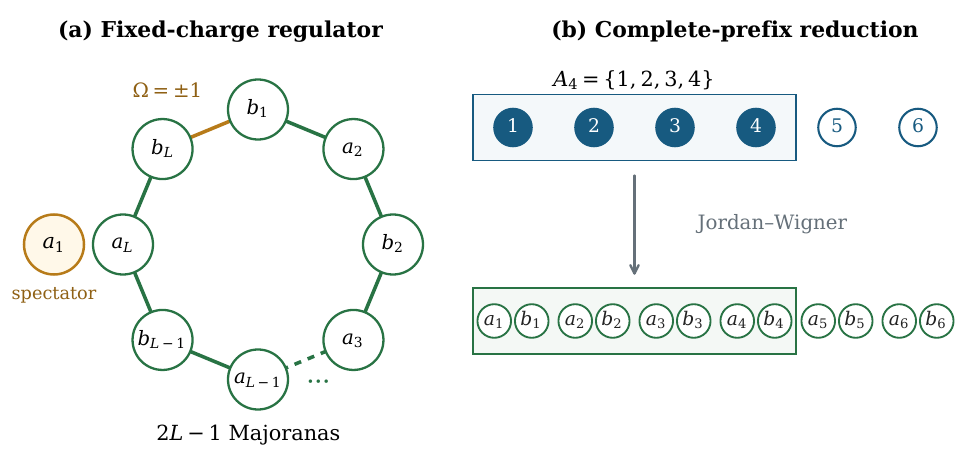}
 \caption{\textbf{Microscopic regulator and complete-prefix reduction.}
 (a) At fixed charge, the KW Hamiltonian separates into the spectator $a_1$
 and an active cycle of $2L-1$ Majoranas.  The closing edge carries
 $\Omega=\pm1$.  Omitted vertices belong to the same cycle.
 (b) The shaded prefix $A_4$ contains four complete spin sites and retains
 both Majoranas at every site.  The unshaded spin sites belong to the traced-out
 complement.  The retained Majoranas' Jordan--Wigner strings remain within
 the non-wrapping prefix, so the spin and CAR descriptions generate the same
 local matrix algebra and, for the same state, have identical RDM spectra.
 Each selected charge-sector ground state has a Gaussian representation.
 The physical mixture is formed only after the separate reductions.}
 \label{fig:defect-rdm-dictionary}
\end{figure}

Two features of this formula are used repeatedly below.  The global state is
generally mixed, and so is its reduction.  Moreover, the labels $\pm$ in
Eq.~\eqref{eq:physical-prefix-rdm} record which full-system preparation was
reduced; they do not split $\mathcal H_{A_n}$, and the two reductions need not
have orthogonal supports.  Local symmetry resolution, taken up in
Section~\ref{sec:charged-parity}, instead decomposes the single state
$\rho_{{\rm can},A_n}^{(L)}$ into parity blocks of its own.

That state inherits an exact local symmetry from its preparation.  Complex
conjugation in the computational basis, followed by $Z_1$, exchanges the two
global charge-sector ground states; the global exchange was exhibited in
Ref.~\cite{AasenMongFendley2016}, Eqs.~(3.38)--(3.39).  Multiplying that
antiunitary by $\Omega$ and descending the product through the ordinary
partial trace yields, for every complete proper prefix,
\begin{equation}
 \Theta_{A_n}\rho_{{\rm can},A_n}^{(L)}\Theta_{A_n}^{-1}
 =\rho_{{\rm can},A_n}^{(L)},
 \qquad
 \Theta_{A_n}^{2}=-I.
 \label{eq:body-prefix-kramers}
\end{equation}
Kramers' argument then pairs every eigenvalue of the physical prefix RDM, zero
included.  The mechanism is the antiunitary one behind symmetry-enforced
entanglement-spectrum degeneracies in one-dimensional symmetry-protected
phases~\cite{PollmannEtAl2010}, acting here on a charge-unresolved
critical-defect RDM rather than on a gapped SPT edge spectrum.  Because
$\Theta_{A_n}$ interchanges the two local-parity blocks, the characteristic
polynomial is a square, and one may write
\begin{equation}
 \rho_{{\rm can},A_n}^{(L)}
 \cong\frac{I_2}{2}\otimes\sigma_{L,n},
 \qquad \operatorname{Tr}\sigma_{L,n}=1.
 \label{eq:body-kramers-factor}
\end{equation}
This is an abstract many-body factorization, not a Gaussian zero mode or a
canonical localized qubit.  Finite-size pairing is exact, yet the
ordered defect-entropy coefficient is left undetermined by it and follows only
from the separate full-RDM and Toeplitz/Fisher--Hartwig analysis of
Section~\ref{sec:entanglement-entropy}.  The descent, the parity-block
equivalence and the scope proof are carried out in
Appendix~\ref{app:physical-rdm}.

\subsection{Charged moments and the exact local-parity blocks}
\label{sec:charged-parity}

With a reduced state and a symmetry action on it in hand, charged moments are
the standard device for extracting symmetry-resolved information.  For a
unitary representation $U_A(g)$ on the retained Hilbert space and a real moment
order $\alpha>0$ -- written $\alpha$ because $n$ already counts spin sites --
define
\begin{equation}
 Z_\alpha(g)=\Tr_A\!\left(\rho_A^\alpha U_A(g)\right).
 \label{eq:charged-moment-definition}
\end{equation}
In a CFT replica construction the symmetry insertion dresses the replicated
geometry; this is how defect-dressed reduced-state information enters
symmetry-resolution methods
\cite{GoldsteinSela2018,XavierAlcarazSierra2018,ZhaoNortheMeyer2020SymmetryResolved,WeisenbergerEtAl2021SymmetryResolved}.
When the symmetry commutes with the state, $[\rho_A,U_A(g)]=0$, character
projection turns these moments into block moments, and positive block
probabilities then permit normalized sector states.  Insertion, projection and
normalization remain three separate steps; the trace in
Eq.~\eqref{eq:charged-moment-definition} does not by itself supply the
conditions needed for the later two.

For the physical KW preparation, set $A=A_n$ and
$\rho_A=\rho_{{\rm can},A_n}^{(L)}$, and resolve its onsite spin-flip parity,
\begin{equation}
 \Omega_A=\prod_{j=1}^{n}X_j,
 \qquad P_{A,\pm}=\frac{I\pm\Omega_A}{2},
 \qquad \mathcal H_{A,\pm}=P_{A,\pm}\mathcal H_A.
 \label{eq:local-parity-projectors}
\end{equation}
The descended symmetry established in Appendix~\ref{app:physical-rdm}
satisfies
\begin{equation}
 [\rho_A,\Omega_A]=0,
 \qquad \Theta_A\Omega_A\Theta_A^{-1}=-\Omega_A.
 \label{eq:body-parity-exchange}
\end{equation}
Taken together with Eq.~\eqref{eq:body-prefix-kramers}, these relations have a
sharp finite-size consequence in the charged-moment language.

\begin{corollary}[Local-parity moments and normalized blocks]
\label{cor:parity-resolved-moments}
For $L\geq2$, every complete proper prefix $A_n$, and every real
$\alpha>0$, the state in Eq.~\eqref{eq:physical-prefix-rdm} obeys
\begin{align}
 Z_\alpha(\Omega_A)&=0,
 &Z_\alpha(\pm)&:=\Tr_A(P_{A,\pm}\rho_A^\alpha)
                         =\frac12 Z_\alpha(e),
 \label{eq:parity-moments}\\
 p_+&=p_-=\frac12,
 &p_\pm&:=\Tr_A(P_{A,\pm}\rho_A).
 \label{eq:parity-probabilities}
\end{align}
Let $R_\pm$ be $P_{A,\pm}\rho_AP_{A,\pm}$ restricted to
$\mathcal H_{A,\pm}$, and normalize on that space by
$\widehat\rho_{A,\pm}=R_\pm/p_\pm=2R_\pm$.  Their spectra coincide, and
\begin{equation}
 S_\alpha(\widehat\rho_{A,\pm})=S_\alpha(\rho_A)-\log2,
 \label{eq:normalized-parity-entropy}
\end{equation}
where $S_\alpha(\rho)=(1-\alpha)^{-1}\log\Tr\rho^\alpha$ for
$\alpha\ne1$ and the $\alpha=1$ value is the von Neumann entropy.
\end{corollary}

\begin{proof}
By the commutation relation in Eq.~\eqref{eq:body-parity-exchange} the state
is block diagonal, $\rho_A=R_+\oplus R_-$, while the state-preserving
antiunitary carries $\mathcal H_{A,+}$ onto $\mathcal H_{A,-}$.  If
$R_+v=\lambda v$, then $R_-\Theta_Av=\lambda\Theta_Av$, because $\lambda$ is
real; since $\Theta_A$ maps an orthonormal basis of one parity space
bijectively onto a basis of the other, the two blocks share all nonnegative
eigenvalues with their multiplicities, zeros included.  Hence
$\Tr R_+^\alpha=\Tr R_-^\alpha$ for every $\alpha>0$.  The sum of those power
traces is $Z_\alpha(e)=\Tr\rho_A^\alpha$ and their difference is the
parity-inserted moment, which proves Eq.~\eqref{eq:parity-moments}; at
$\alpha=1$, unit trace forces $p_\pm=1/2$.  Normalizing on each block then
gives
\begin{equation}
 \Tr\widehat\rho_{A,\pm}^{\,\alpha}
       =2^{\alpha-1}\Tr\rho_A^\alpha,
 \label{eq:normalized-parity-moments}
\end{equation}
whose logarithm yields Eq.~\eqref{eq:normalized-parity-entropy}, the exponent
$\alpha-1$ cancelling against the prefactor $(1-\alpha)^{-1}$.  For
$\alpha=1$ the same identity for the von Neumann entropy follows from the
paired eigenvalues, with $0\log0=0$.
\end{proof}

By symmetry alone, Eq.~\eqref{eq:normalized-parity-entropy} gives complete
equipartition of the normalized local-parity blocks of the state fixed in
Eq.~\eqref{eq:physical-prefix-rdm}.  No Gaussian
representation of $\rho_A$ enters the corollary or its proof.
The distinction drawn after Eq.~\eqref{eq:physical-prefix-rdm} now becomes
quantitative.  Reductions of the two globally charged preparations can
overlap -- at $(L,n)=(2,1)$ their Hilbert--Schmidt overlap equals $1/3$
(Appendix~\ref{app:physical-rdm}) -- whereas the two local-parity blocks have
orthogonal supports.  The $\log2$ in Eq.~\eqref{eq:normalized-parity-entropy}
therefore belongs to those blocks, and is not a Shannon term attached to a
mixture of preparations.

Table~\ref{tab:state-resolution-objects} places these two uses of sector
labels beside a third, categorical resolution, whose projectors have to be
specified by a defect or boundary algebra
\cite{ChoiRayhaunZheng2024,HeymannQuella2025}.

\begin{table}[ht]
\centering
\begin{tabular}{@{}>{\raggedright\arraybackslash}p{.27\textwidth}>{\raggedright\arraybackslash}p{.64\textwidth}@{}}
\toprule
Object & Construction and information\\
\midrule
Global-charge preparation & Mix ground projectors in full-system charge
sectors, then reduce; the resulting local supports can overlap.\\
Local-parity resolution & Project the specified $\rho_A$ with $P_{A,\pm}$
and normalize; the two orthogonal blocks here have equal spectra.\\
Categorical resolution & Specify the appropriate defect or boundary-algebra
projectors; their labels and retained data are different from onsite parity.\\
\bottomrule
\end{tabular}
\caption{Three different roles of sector data.  The finite result in
Corollary~\ref{cor:parity-resolved-moments} concerns the second row.}
\label{tab:state-resolution-objects}
\end{table}

How much a symmetry insertion settles by itself is visible from the existing
constructions.  For a $U(1)$ charge, the endpoint-insertion construction and
its later extensions to excited global states and to two intervals place both
the preparation and the region geometry inside the replicated observable
\cite{ZhaoNortheMeyer2020SymmetryResolved,WeisenbergerEtAl2021SymmetryResolved}.
A modular spin-three charge supplies a higher-spin insertion of quite another
type, and its large-central-charge analysis need not display equipartition at
leading order \cite{ZhaoEtAl2022ChargedMoments}.  Inserting a symmetry is
therefore not enough to force the cancellation in
Eq.~\eqref{eq:parity-moments}; what forces it here is the additional
state-preserving parity exchange.

For a vacuum interval with specified conformal entangling boundaries, a
complementary route avoids explicit charged-moment computation.  The interval
state space decomposes into representations of the preserved boundary algebra;
the corresponding characters give the sector traces
\cite{DiGiulioEtAl2022BCFTSymmetryResolved,Northe2023ConformalResolution}.
Resolving an interval state into conformal families is a different operation
from resolving its individual energy levels, and different again from
resolving its onsite parity.  For specified Majorana BCFT factorizations, a
vanishing parity-inserted trace likewise gives exact parity
equipartition~\cite{Northe2025FermionParity}; the algebraic criterion is the
one shared with Corollary~\ref{cor:parity-resolved-moments}, while the spin
structures, entangling boundaries and normalized trace of that construction
make it a different realization, and the two reduced states are not identified
here.  Interface constructions prepare reduced states by yet another route, in
which sectors are filtered
physically~\cite{NorthePolettiRossi2026Interfaces};
Section~\ref{sec:entanglement-entropy} compares that operation, and the
entropy reference accompanying it, with the finite spin construction of this
paper.

Since the global state used here is mixed, the quantities in
Eq.~\eqref{eq:normalized-parity-entropy} are symmetry-resolved subsystem
entropies and spectral statements rather than measures of quantum
entanglement alone.  What they resolve is an invertible onsite symmetry of a
state prepared with a non-invertible defect, on the ordinary full local
algebra.  The projectors of Eq.~\eqref{eq:local-parity-projectors} are built
from $\Omega_A$, and no projector is formed from the non-invertible line
itself.  The KW-minus-homogeneous entropy excess is a separate matter, fixed
by the ordered-limit analysis of Section~\ref{sec:entanglement-entropy} and
not by local parity equipartition.

\subsection{Gaussian sector states and complete-site reduction}
\label{sec:gaussian-sector-reduction}

The entropy comparison of Section~\ref{sec:entanglement-entropy} calls for two
further ground states, the homogeneous chain and a matched invertible control,
both in the spin and Jordan--Wigner conventions already fixed.  On the full
spin Hilbert space put
\begin{equation}
 H_\eta=-\sum_{j=1}^{L}X_j
 -\sum_{j=1}^{L-1}Z_jZ_{j+1}-\eta Z_LZ_1,
 \qquad \eta\in\{+1,-1\}.
 \label{eq:eta-control-Hamiltonian}
\end{equation}
Thus $\eta=+1$ is homogeneous, while $\eta=-1$ reverses the spatial seam.  The
matched $\eta$ control is the unique global ground state of the $\eta=-1$
Hamiltonian.  It lies in $\Omega=+1$ and, for $L\geq3$, is the
Pfaffian-$+1$ completion of the resulting periodic zero complex mode.  At
$L=2$, where bulk and seam are one and the same bond, the prescription is
applied directly to the spin Hamiltonian.  Seam and $D_\psi$ are the two
windings of a single invertible Ising line.  Along Euclidean time, the line
produces this spatial seam, whereas wrapped around space it acts as
$D_\psi=\Omega$ on the untwisted Hilbert space
\cite{AasenMongFendley2016}.  This control preparation is pure on the full
chain even though its subregion RDM is generally mixed.  Its state and prefix
RDM are denoted $\rho_{\eta,L}$ and $\rho_{\eta,A}^{(L)}$, and its physical
covariance $\Gamma_\eta$.

For the fixed-charge sector states, for this control, and for the auxiliary
comparison states introduced below, we use the Majorana-covariance form of the
standard free-fermion correlation-matrix construction
\cite{ChungPeschel2001,CheongHenley2004,Peschel2003,PeschelEisler2009}.  Once
a state is represented in a fixed Jordan--Wigner convention -- the CAR chart,
CAR standing for the canonical anticommutation relations -- its real
antisymmetric covariance reads
\begin{equation}
 \Gamma_{mn}=-\langle i\gamma_m\gamma_n\rangle,
 \qquad m\ne n,
 \qquad \Gamma_{mm}=0.
 \label{eq:covariance-definition}
\end{equation}
When $iA_\omega$ has no zero eigenvalue, the pure ground-state covariance
follows from the spectral sign.  In either KW charge block, however,
$iA_\omega$ has a two-dimensional kernel, and we take
\begin{equation}
 \Gamma_{\omega,0}=-i\,\operatorname{sign}(iA_\omega).
 \label{eq:sector-sign-zero-covariance}
\end{equation}
with $\operatorname{sign}(0)=0$ as the auxiliary sign-zero covariance.  Its
vanishing Pfaffian excludes it as the covariance of the physical Gaussian
ground state in the sector $\Omega=\omega$, since
Theorem~\ref{thm:parity-pfaffian} equates that state's covariance Pfaffian with its charge expectation $\omega$.

Let $Z_\omega\in\mathbb R^{2L\times2}$ have orthonormal columns spanning the
kernel of $A_\omega$ and let
\begin{equation}
 J=\begin{pmatrix}0&1\\-1&0\end{pmatrix}.
 \label{eq:Jzero}
\end{equation}
Within that same charge-block CAR chart, the two pure Gaussian completions are
\begin{equation}
 \Gamma_{\omega,\pm}=\Gamma_{\omega,0}
 \pm Z_\omega JZ_\omega^{\mathsf T},
 \qquad
 \Gamma_{\omega,\pm}^2=-1.
 \label{eq:pure-completions}
\end{equation}
In the ordered Majorana convention of Eq.~\eqref{eq:gamma-order}, the physical
ground state of the block is represented by the unique choice
$\Gamma_{\omega,{\rm phys}}$ satisfying
\begin{equation}
 \operatorname{Pf}(\Gamma_{\omega,{\rm phys}})=\omega.
 \label{eq:physical-sector-pfaffian}
\end{equation}
Each selected sector state is thus Gaussian in its own chart, but averaging
the two completions of one fixed matrix $A_\omega$ returns the Gaussian state
associated with $\Gamma_{\omega,0}$, whereas the physical ensemble of
Eq.~\eqref{eq:physical-global-ensemble} averages selected states drawn from
two different charge blocks.  The two operations are not interchangeable, and
the finite mixture need not be Gaussian.

Even the smallest nontrivial complete prefix shows the difference.  Write
$\rho_{G,\omega,0;A_n}^{(L)}$ for the auxiliary Gaussian restriction generated
by $\Gamma_{\omega,0}$.  At $(L,n)=(3,2)$, either charge block obeys the exact
purity mismatch
\begin{equation}
 \left|\operatorname{Tr}\!\left[
   \bigl(\rho_{{\rm can},A_2}^{(3)}\bigr)^2\right]
 -\operatorname{Tr}\!\left[
   \bigl(\rho_{G,\omega,0;A_2}^{(3)}\bigr)^2\right]\right|
 =\frac{\sqrt5}{50}.
 \label{eq:physical-proxy-purity-mismatch}
\end{equation}
No sector-derived sign-zero proxy, therefore, reproduces the finite-size
physical RDM universally.  In particular, a zero single-particle level of
$\Gamma_{\omega,0}$ is not a physical finite-size $\xi=0$ theorem for
$\rho_{{\rm can},A_n}^{(L)}$; the full spectra and the proof are given in
Appendix~\ref{app:physical-rdm}.

\begin{theorem}[Parity--Pfaffian identity]\label{thm:parity-pfaffian}
For any pure Gaussian parity eigenstate in the convention of
Eq.~\eqref{eq:covariance-definition},
\begin{equation}
 \left\langle\prod_{j=1}^{L}X_j\right\rangle=\Pf(\Gamma).
 \label{eq:parity-pfaffian}
\end{equation}
\end{theorem}

\begin{proof}
Equation~\eqref{eq:JW-local-identities} gives $X_j=-ia_jb_j$, and Wick's
theorem gives
\begin{equation}
 \left\langle\gamma_1\gamma_2\cdots\gamma_{2L}\right\rangle
 =\Pf\!\left(\langle\gamma_m\gamma_n\rangle\right)
 =i^L\Pf(\Gamma),
\end{equation}
because $\langle\gamma_m\gamma_n\rangle=i\Gamma_{mn}$ for $m\ne n$.  The
factor $(-i)^L$ carried by $\prod_jX_j$ then proves the identity.
\end{proof}

For the homogeneous and invertible-$\eta$ ground states used below, the
corresponding physical selector is
\begin{equation}
 \operatorname{Pf}(\Gamma)=+1.
 \label{eq:physical-parity}
\end{equation}
Together with Eq.~\eqref{eq:physical-sector-pfaffian}, this condition is
invariant under any change of basis inside the numerical zero-mode subspace,
and it replaces conventions based on the index returned by a QR decomposition
or an eigensolver.

Consider now the reduction itself.  For a prefix region containing $n$
complete spin sites,
\begin{equation}
 A_n=\{1,2,\ldots,n\},
 \label{eq:prefix-region}
\end{equation}
the fermionic restriction keeps both Majoranas $(a_j,b_j)$ of every
$j\in A_n$; if site~1 lies in the region, in particular, the spectator $a_1$
is kept.  The physical prefix RDM is accordingly carried by the complete set
of $2n$ Majoranas associated with those sites.

Let $P_{A_n}$ denote the corresponding $2n\times2L$ coordinate projection,
which restricts the covariance to
\begin{equation}
 \Gamma_{A_n}=P_{A_n}\Gamma P_{A_n}^{\mathsf T}.
 \label{eq:restricted-covariance}
\end{equation}
If the eigenvalues of $i\Gamma_{A_n}$ are $\pm\nu_k$, with
$0\le\nu_k\le1$, the single-particle entanglement levels are
\begin{equation}
 \xi_k=\log\frac{1+\nu_k}{1-\nu_k},
 \label{eq:single-particle-ES}
\end{equation}
and the eigenvalues of the reduced state are products of
$(1\pm\nu_k)/2$.

\begin{theorem}[Prefix spin--Gaussian RDM equivalence]\label{thm:prefix-rdm}
Let $|\Psi\rangle$ be a spin state whose Jordan--Wigner image is a Gaussian
state with covariance $\Gamma$.  For the complete non-wrapping prefix $A_n$,
the spin reduced density matrix and the Gaussian reduced state determined by
$\Gamma_{A_n}$ have the same spectrum, for every $L$ and $n\le L$.
\end{theorem}

\begin{proof}
For $j\le n$, both Majoranas $a_j,b_j$ of
Eq.~\eqref{eq:JW-majoranas} involve spin operators on sites $1,\ldots,j$
only.  Conversely,
\begin{equation}
 X_j=-ia_jb_j,
 \qquad
 Z_j=\left(\prod_{k<j}X_k\right)a_j,
 \label{eq:inverse-prefix-JW}
\end{equation}
so the matrix algebra generated by the first $n$ spin sites coincides with the
matrix algebra generated by the first $2n$ Majoranas.  The two restrictions of
$|\Psi\rangle$ are therefore related by the finite-dimensional
Jordan--Wigner algebra isomorphism.  A Gaussian restriction is in turn
determined uniquely by $\Gamma_{A_n}$ through Wick's theorem and the standard
Gaussian RDM construction~\cite{Peschel2003}, and spectral equivalence
follows.
\end{proof}

The premise of Theorem~\ref{thm:prefix-rdm} is met by the homogeneous physical
ground state, by the invertible-$\eta$ state selected by
Eq.~\eqref{eq:physical-parity}, and by the fixed-charge KW states selected by
Eq.~\eqref{eq:physical-sector-pfaffian}, in each case through the quadratic
diagonalization and the physical-sector dictionary.  As an independent
finite-size check, complete spin and Gaussian RDM spectra were compared for
$L=4,6,8,10$ in 60 state--region comparisons, including both KW charge
sectors.  The corresponding RDM eigenvalues differ by less than
$2.34\times10^{-15}$, and the trace errors are below $1.78\times10^{-15}$.

\begin{remark}[Scope]
Theorem~\ref{thm:prefix-rdm} is invoked only for complete, non-wrapping prefix
regions.  Wrapping intervals call for a separate
Jordan--Wigner/subsystem-algebra analysis, and are not inferred from a
coincidence of entropies.
\end{remark}

\subsection{Zero-sector locality and the ordered limit}\label{sec:entanglement-locality}

The prefix spin--Gaussian theorem applies separately to each physical
fixed-charge ground state.  To control the ordered limit, we compare that
state within its own charge-block CAR chart with the auxiliary sign-zero
Gaussian state.  Fix $\omega\in\{+1,-1\}$ and suppress the charge label in
this subsection, writing $\Gamma_0=\Gamma_{\omega,0}$ and
$\Gamma_{\rm phys}=\Gamma_{\omega,{\rm phys}}$.  This comparison does not
identify either state with the charge-unresolved RDM in
Eq.~\eqref{eq:physical-prefix-rdm}.

We work in the ordered Majorana basis of Eq.~\eqref{eq:gamma-order}.  At the
critical point, $a_1$ is absent from the quadratic form and the remaining
$2L-1$ vertices form an odd cycle.  A convenient orthonormal basis of the
kernel is
\begin{equation}
 z_{\mathrm{sp}}=e_{a_1},
 \qquad
 z_{\mathrm{cyc}}=\frac{1}{\sqrt{2L-1}}
 \sum_{m\ne a_1}s_m e_m,
 \label{eq:kernel-basis}
\end{equation}
where $e_m$ is the coordinate unit vector at Majorana vertex $m$,
``sp'' denotes the spectator mode, ``cyc'' denotes the odd-cycle mode, and
$s_m=\pm1$ is fixed by the oriented-cycle recurrence.  The cycle zero vector
has uniform magnitude $1/\sqrt{2L-1}$; its charge-dependent sign pattern is
immaterial for the restriction-norm estimates below.
The two vectors exhaust the kernel, as the spectator contributes one null
direction and an odd real antisymmetric cycle contributes one.

The auxiliary covariance sets the spectral sign to zero on this kernel.  The
two mathematical pure completions of the same fixed-charge matrix are
\begin{equation}
 \Gamma_\pm=\Gamma_0\pm
 \left(z_{\mathrm{sp}}z_{\mathrm{cyc}}^{\mathsf T}
 -z_{\mathrm{cyc}}z_{\mathrm{sp}}^{\mathsf T}\right).
 \label{eq:completion-explicit}
\end{equation}
Exactly one of them is $\Gamma_{\rm phys}$, selected by
Eq.~\eqref{eq:physical-sector-pfaffian}; the other has the opposite Pfaffian.
The expression fixes their rank-two displacement from the auxiliary comparator
without choosing a localized linear combination of numerical zero
eigenvectors.

Let $A_\ell$ be the prefix containing $2\ell$ complete spin sites.  It contains
$4\ell$ Majoranas, comprising the spectator and $4\ell-1$ vertices of the odd cycle.
Consequently,
\begin{equation}
 \left\|P_{A_\ell}z_{\mathrm{cyc}}\right\|^2
 =\frac{4\ell-1}{2L-1}.
 \label{eq:restricted-zero-mass}
\end{equation}
Moreover, $P_{A_\ell}z_{\mathrm{sp}}$ is a unit vector orthogonal to
$P_{A_\ell}z_{\mathrm{cyc}}$.  Restricting Eq.~\eqref{eq:completion-explicit}
therefore gives the following exact finite-size statement.

\begin{theorem}[Zero-sector locality]\label{thm:zero-locality}
For the complete-prefix-site region $A_\ell$,
\begin{equation}
 \left\|P_{A_\ell}(\Gamma_\pm-\Gamma_0)
 P_{A_\ell}^{\mathsf T}\right\|_{\mathrm{op}}
 =\sqrt{\frac{4\ell-1}{2L-1}}.
 \label{eq:locality-opnorm}
\end{equation}
\end{theorem}

\begin{proof}
Write $u=P_{A_\ell}z_{\mathrm{sp}}$ and
$v=P_{A_\ell}z_{\mathrm{cyc}}$.  The restricted difference is the
antisymmetric rank-two operator $\pm(uv^{\mathsf T}-vu^{\mathsf T})$.
Because $\|u\|=1$ and $u^{\mathsf T}v=0$, its two nonzero singular values are
both $\|v\|$.  Equation~\eqref{eq:restricted-zero-mass} then gives
Eq.~\eqref{eq:locality-opnorm}.
\end{proof}

The identity was also evaluated on 32 $(L,\ell)$ pairs through $L=512$; the
maximum discrepancy in both Eqs.~\eqref{eq:restricted-zero-mass} and
\eqref{eq:locality-opnorm} was $2.22\times10^{-16}$.

\begin{corollary}[Fixed-block equivalence]\label{cor:ordered-equivalence}
At fixed $\ell$,
\begin{equation}
 \lim_{L\to\infty}
 \left\|P_{A_\ell}(\Gamma_\pm-\Gamma_0)
 P_{A_\ell}^{\mathsf T}\right\|_{\mathrm{op}}=0,
 \label{eq:fixed-block-limit}
\end{equation}
with leading scale $L^{-1/2}$.
\end{corollary}

\begin{corollary}[Entropy modulus at fixed block]\label{cor:entropy-modulus}
Let $S_A(\Gamma)$ be the Gaussian entropy of the $4\ell$-Majorana restriction,
and set
\begin{equation}
 \delta_{L,\ell}=\sqrt{\frac{4\ell-1}{2L-1}}.
\end{equation}
Whenever $\delta_{L,\ell}\le1$,
\begin{equation}
 \left|S_A(\Gamma_\pm)-S_A(\Gamma_0)\right|
 \le 2\ell\,h_2\!\left(\frac{\delta_{L,\ell}}2\right),
 \label{eq:entropy-modulus}
\end{equation}
where $h_2(x)=-x\log x-(1-x)\log(1-x)$.
At fixed $\ell$ this is $O(L^{-1/2}\log L)$.
\end{corollary}

\begin{proof}
Weyl's inequality pairs the $2\ell$ nonnegative singular values of the two
restricted covariances so that $|\nu_k^\pm-\nu_k^0|\le\delta_{L,\ell}$.
Each mode contributes the binary entropy
$h_2((1+\nu_k)/2)$.  The sharp binary entropy continuity inequality
$|h_2(p)-h_2(q)|\le h_2(|p-q|)$ for $|p-q|\le1/2$, together with symmetry of
$h_2$, bounds each contribution by $h_2(\delta_{L,\ell}/2)$.  Summing the
$2\ell$ modes proves Eq.~\eqref{eq:entropy-modulus}.
\end{proof}

The order of limits matters.  At fixed $\ell$, the bound proves that the
physical sector covariance and its charge-matched auxiliary comparator have
the same local Gaussian limit.  If $\ell/L$ approaches a nonzero constant,
the right-hand side of Eq.~\eqref{eq:locality-opnorm} need not vanish.  The
bound applies only to complete prefixes at fixed $\ell$; it neither extends to
fixed-ratio, arbitrary, or wrapping regions nor equates the finite-$L$
physical RDM with an auxiliary Gaussian state.

\subsection{Analytic lowest-level distinction}\label{sec:entanglement-es}

The charge-matched locality bound gives a finite-size softening ceiling for
the physical fixed-charge states relative to their auxiliary sign-zero comparators.
For a complete prefix, write $u=P_Az_{\mathrm{sp}}$ for the restricted
spectator zero vector.  Since $u\in\ker\Gamma_{\omega,0,A}$, the two
mathematical completions in one charge-block chart are related by the
orthogonal reflection
\begin{equation}
 R_A=1-2uu^{\mathsf T},
 \qquad
 R_A\Gamma_{\omega,+,A}R_A^{\mathsf T}=\Gamma_{\omega,-,A}.
 \label{eq:pure-ES-equivalence}
\end{equation}
They therefore have identical complete-prefix Gaussian single-particle
spectra.  Only the Pfaffian-selected completion is the physical state in that
charge sector.

By contrast,
Theorem~\ref{thm:zero-locality} gives a rigorous finite-size bound
on the pure lowest level.  Since $i\Gamma_{0,A}$ has a zero eigenvalue, Weyl's
inequality and Eq.~\eqref{eq:locality-opnorm} imply
\begin{equation}
 \nu_{\min}(\Gamma_{\pm,A})
 \le \delta_{L,\ell},
 \qquad
 \delta_{L,\ell}=\sqrt{\frac{4\ell-1}{2L-1}}.
 \label{eq:pure-nu-bound}
\end{equation}
Because $\xi(\nu)=2\operatorname{artanh}\nu$ is increasing,
\begin{equation}
 \xi_{\min}(\Gamma_{\pm,A})
 \le 2\operatorname{artanh}\delta_{L,\ell}.
 \label{eq:pure-xi-bound}
\end{equation}
At fixed $\ell$, the right-hand side is $O(L^{-1/2})$, giving an analytic
softening bound for the pure states.

The auxiliary covariance also has an exact restricted zero for every complete
prefix.  Its spectator row and column vanish, while its active restriction has
odd dimension $4\ell-1$ and hence at least one further null vector.  Therefore
$i\Gamma_{0,A}$ has nullity at least two, i.e. at least one exact Gaussian
single-particle level $\nu=\xi=0$.  This statement remains confined to the
auxiliary comparator.

Thus the auxiliary comparator has an exact restricted zero, whereas the
physical fixed-charge level obeys only the softening upper bound in
Eq.~\eqref{eq:pure-xi-bound}; its positivity is observed only on the finite
grid shown in Appendix~\ref{app:lowest-level-illustration}.  That appendix
retains the corresponding illustration, with all objects typed as
sectorwise physical states or auxiliary Gaussian comparators.  Neither result
supplies the many-body spectrum of the charge-unresolved physical RDM.  Its
exact finite-size pairing instead follows from the reduced antiunitary
symmetry proved in Appendix~\ref{app:physical-rdm}, while its ordered-limit
entropy depends on the full-RDM convergence established separately in
Appendix~\ref{app:entropy-bridge}.
The next section states the physical entropy theorem and reduces its limiting
constant to the difference between consecutive odd and even
critical-Majorana blocks.

\Needspace{9\baselineskip}
\section{Analytic ordered-limit defect entropy}\label{sec:entanglement-entropy}

Two different subtractions in this paper produce a logarithm of two, and what
separates them is the reference.  In Section~\ref{sec:charged-parity}, normalizing a local-parity block of
$\rho_{{\rm can},A_n}^{(L)}$ lowers its entropy by exactly $\log2$ relative to
the corresponding unresolved state, at every size.  The question opened here concerns
the cost of the defect itself, namely how much entropy the defect and the
preparation carrying it add at fixed retained region, measured against the
homogeneous critical chain.

That cost is accessible even though the physical reduced state resists the
usual shortcut.  Equation~\eqref{eq:physical-prefix-rdm} is an equal-weight
convex sum of two charge-sector reductions, and a convex sum of Gaussian
states need not be Gaussian, so averaging covariances is not available here.
The route taken below controls the density matrix itself.  At fixed complete
prefix, the diagonal reductions converge to a common Gaussian state and the
cross-sector reduction vanishes.  The resulting limit therefore covers
arbitrary ground-doublet preparations.  Entropy continuity at fixed
dimension then reduces the excess to the consecutive odd--even
critical-Majorana Toeplitz comparison.  The spectator and the odd-cycle zero
mode supply an explicit $\log2$; the endpoint increment removes half of it
when the prefix subsequently grows.  This identifies the lattice mechanism
before comparison with numerics and continuum constructions.

\subsection{Physical entropy functional and order of limits}

For a restricted Majorana covariance $\Gamma_A$, write the eigenvalues of
$i\Gamma_A$ as pairs $\pm\nu_k$, with $0\leq\nu_k\leq1$.  The Gaussian entropy is
\begin{equation}
 S(\Gamma_A)=\sum_k \eta(\nu_k),\qquad
 \eta(x)=-\frac{1+x}{2}\log\frac{1+x}{2}
         -\frac{1-x}{2}\log\frac{1-x}{2},
 \label{eq:gaussian-entropy}
\end{equation}
with $0\log0=0$.  This is the Majorana form of the standard
correlation-matrix construction~\cite{Peschel2003,PeschelEisler2009}.  The
constant considered here is a defect/interface entropy, extending the
Affleck--Ludwig $g$-function language to a spatial defect rather than a
physical boundary~\cite{AffleckLudwig1991,SakaiSatoh2008,BrehmBrunnerJaudSchmidtColinet2016,GutperleMiller2016}.

For the complete prefix $A_\ell$ containing $2\ell$ spin sites, define the
physical entropy differences
\begin{align}
 \Delta S_\eta(L,\ell)
 &=S(\rho_{\eta,A_\ell}^{(L)})
   -S(\rho_{{\rm hom},A_\ell}^{(L)}),\label{eq:dSeta}\\
 \Delta S_{\rm can}(L,\ell)
 &=S(\rho_{{\rm can},A_\ell}^{(L)})
   -S(\rho_{{\rm hom},A_\ell}^{(L)}).
 \label{eq:dSmixed}
\end{align}
Here $\rho_{{\rm can},A_\ell}^{(L)}$ is the many-body convex sum in
Eq.~\eqref{eq:physical-prefix-rdm}.  For comparison with the finite-size
covariance diagnostics of Section~\ref{sec:entanglement-es}, we also write
\begin{equation}
 \Delta S_{\omega,{\rm phys}}^{G}(L,\ell)
 =S[\rho_G(\Gamma_{\omega,{\rm phys}},A_\ell)]
  -S(\rho_{{\rm hom},A_\ell}^{(L)}),
 \label{eq:dSpure}
\end{equation}
where the superscript $G$ records the sectorwise Gaussian representation.
The auxiliary entropy built from $\Gamma_{\omega,0}$ will be denoted
$\Delta S_{\omega,0}^{G}$ and is never identified with
$\Delta S_{\rm can}$ at finite $L$.

The order of limits is part of the physical theorem.  We first send the
circumference $L$ to infinity at fixed complete prefix $A_\ell$, and only then
let the prefix grow.  The quantities to be evaluated are
\begin{equation}
 \lim_{\ell\to\infty}\left\{
 \lim_{\substack{L\to\infty\\L>2\ell}}
 \Delta S_{\rm can}(L,\ell)\right\},
 \qquad
 \lim_{\ell\to\infty}\left\{
 \lim_{\substack{L\to\infty\\L>2\ell}}
 \Delta S_\eta(L,\ell)\right\}.
 \label{eq:ordered-limit}
\end{equation}
No fixed-ratio, simultaneous, exchanged-limit, wrapping-region, or
arbitrary-region conclusion is inferred.

\subsection{Exact fixed-prefix reduction}

The finite-cycle covariances admit a direct Fourier representation.  Let $N$
be the number of vertices in the active cycle and let
$q_r=(2\pi r+\vartheta)/N$, $r=0,\ldots,N-1$, with twist
$\vartheta=0$ or $\pi$.  At fixed separation $d=k-j$, their entries are the
discrete Fourier coefficients of
$-\operatorname{sgn}(\sin q)$,
\begin{equation}
 C_{N,\vartheta}(d)
 =\operatorname{Im}\!\left[
   \frac1N\sum_{r=0}^{N-1}
   \bigl[-\operatorname{sgn}(\sin q_r)\bigr]e^{-iq_rd}
  \right].
 \label{eq:finite-cycle-kernel}
\end{equation}
For each fixed integer $d$, the Riemann sum converges to
\begin{equation}
 \frac{1}{2\pi}\operatorname{Im}
 \int_{-\pi}^{\pi}-\operatorname{sgn}(q)e^{-iqd}\,dq
 =\begin{cases}
 \dfrac{2}{\pi d},&d\ \text{odd},\\[3pt]
 0,&d\ \text{even}.
 \end{cases}
 \label{eq:finite-cycle-kernel-limit}
\end{equation}
The periodic and antiperiodic grids are both Riemann-sum partitions with mesh
$2\pi/N$, and changing the assigned value at either jump changes the sum by
at most $O(N^{-1})$.  Because this bound is uniform for the two twists, the
limit is independent of both the twist and the zero-value convention.  A
complete prefix has fixed dimension
$4\ell$; hence entrywise convergence implies operator-norm convergence.
The physical $\eta$ completion differs from its sign-zero periodic comparator
by a rank-two zero-mode term whose entries are $O(L^{-1})$, and therefore has
the same fixed-prefix limit.  In either KW charge block, the spectator row and
column of the auxiliary covariance vanish identically, while the remaining
$4\ell-1$ selected Majoranas are consecutive vertices of the odd cycle.  Thus
\begin{align}
 \Gamma_{{\rm hom},A_\ell},\ \Gamma_{\eta,A_\ell}
   &\longrightarrow T_{4\ell},\label{eq:fixed-prefix-even}\\
 \Gamma_{\omega,0,A_\ell}
   &\longrightarrow 0\oplus T_{4\ell-1},\qquad \omega=\pm1.
 \label{eq:fixed-prefix-odd}
\end{align}
where $T_N$ is the $N\times N$ real antisymmetric Toeplitz matrix
\begin{equation}
 (T_N)_{jk}=\begin{cases}
 \dfrac{2}{\pi(k-j)},&k-j\ {\rm odd},\\[3pt]
 0,&k-j\ {\rm even},
 \end{cases}
 \qquad 0\leq j,k<N.
 \label{eq:critical-Toeplitz}
\end{equation}
Toeplitz and Fisher--Hartwig methods have a long history in spin-chain
entanglement \cite{JinKorepin2004,ItsJinKorepin2005}; here they are used for
the distinct consecutive odd--even increment forced by the KW spectator.
The leading zero in Eq.~\eqref{eq:fixed-prefix-odd} is the retained spectator.
The odd active block has exactly one zero eigenvalue, proved below.  These two
real null directions form one maximally mixed complex mode and contribute one
$\log2$, not two.

For either charge sector, Theorem~\ref{thm:zero-locality} gives
\begin{equation}
 \left\|P_{A_\ell}(\Gamma_{\omega,{\rm phys}}-\Gamma_{\omega,0})
 P_{A_\ell}^{\mathsf T}\right\|_{\rm op}
 =\sqrt{\frac{4\ell-1}{2L-1}}\longrightarrow0
 \qquad(L\to\infty\ {\rm at\ fixed}\ \ell).
 \label{eq:entropy-pure-mixed-locality}
\end{equation}
This covariance estimate applies sectorwise.  To reach the physical
charge-unresolved RDM, one must also control every correlator and reconstruct
the complete density matrix.  Appendix~\ref{app:entropy-bridge} proves, for
each fixed $\ell$ and either $\omega$,
\begin{equation}
 \left\|\rho_{\omega,A_\ell}^{(L)}
 -\rho_G(0\oplus T_{4\ell-1})\right\|_1\longrightarrow0.
 \label{eq:body-sector-full-rdm-limit}
\end{equation}
The same comparator occurs in both charge sectors, so their ordinary convex
sum has the same limit,
\begin{equation}
 \left\|\rho_{{\rm can},A_\ell}^{(L)}
 -\rho_G(0\oplus T_{4\ell-1})\right\|_1\longrightarrow0.
 \label{eq:body-physical-full-rdm-limit}
\end{equation}
Similarly,
\begin{align}
 \left\|\rho_{{\rm hom},A_\ell}^{(L)}
 -\rho_G(T_{4\ell})\right\|_1&\longrightarrow0,
 \label{eq:body-hom-full-rdm-limit}\\
 \left\|\rho_{\eta,A_\ell}^{(L)}
 -\rho_G(T_{4\ell})\right\|_1&\longrightarrow0.
 \label{eq:body-eta-full-rdm-limit}
\end{align}
For the physical $\eta$ completion, the rank-two periodic zero-mode term has
prefix entries $O(L^{-1})$.  The $\eta$ state is parity-even and quasifree, so
the same Pfaffian telescoping and complete-Majorana-monomial reconstruction as
for the homogeneous state promote Eq.~\eqref{eq:fixed-prefix-even} to the
second trace-norm limit above.  These are fixed-dimensional trace-norm limits,
not covariance averaging.  Audenaert continuity~\cite{Audenaert2007} then
promotes them to the physical entropy limit.

To include coherent states, let $A_m=\{1,\ldots,m\}$, where $m$ counts
spins; the preceding $A_\ell$ is recovered by $m=2\ell$.  The reduction of
Eq.~\eqref{eq:general-ground-preparation} is
\begin{equation}
 X_{m,L}:=\operatorname{Tr}_{A_m^c}|\psi_{+,L}\rangle\langle\psi_{-,L}|,
 \qquad
 \rho_{A_m}^{(L)}(p,c)=p\rho_{+,A_m}^{(L)}+(1-p)\rho_{-,A_m}^{(L)}
                          +cX_{m,L}+\bar cX_{m,L}^\dagger.
 \label{eq:general-ground-reduction}
\end{equation}

\begin{proposition}[Common local limit of ground-state preparations]
\label{prop:ground-preparation-limit}
For every fixed complete proper prefix $A_m$, there are finite constants
$B_m,C_m$, independent of $L$, such that, with $N=2L-1$,
\begin{align}
 \|X_{m,L}\|_1&\leq B_mN^{-1/8}(1+\log N)^m,\notag\\
 \sup_{p,c}\|\rho_{A_m}^{(L)}(p,c)-\rho_G(0\oplus T_{2m-1})\|_1
 &\leq C_mf_m(L)+B_mN^{-1/8}(1+\log N)^m\longrightarrow0,
 \label{eq:general-ground-rdm-limit}
\end{align}
where the supremum is over the positive density matrices in
Eq.~\eqref{eq:general-ground-preparation} and $f_m$ is given in
Eq.~\eqref{appD:entrywise-rate}.
\end{proposition}

Appendix~\ref{appD:coherent-preparations} proves the cross-term bound by
controlling the complete set of transition moments.  Combining it with the
two diagonal-sector bounds and $2|c|\leq1$ gives the uniform statement.
Thus neither discarding cross terms alone nor averaging covariances is
sufficient for the common local limit.  The conclusion allows $L$-dependent amplitudes and phases.

In operational terms, this trace-norm limit means that all bounded prefix
observables become unable to distinguish these preparations, while the
global pure states may remain distinct.  Subsystem trace distance provides
the corresponding diagnostic in CFT~\cite{ZhangRuggieroCalabrese2019SubsystemDistances};
equality of entropies alone would not imply it.

Let $S_N$ denote the sum of $\eta(\nu)$ over all positive eigenvalues of
$iT_N$, namely the entropy functional evaluated on the $N\times N$ Toeplitz
block.  Let $S_{2n-1}^{\circ}$ denote the same sum with the unique zero
eigenvalue of the $(2n-1)\times(2n-1)$ odd block omitted.
Equations~\eqref{eq:body-eta-full-rdm-limit},
\eqref{eq:body-physical-full-rdm-limit}, and
\eqref{eq:body-hom-full-rdm-limit} give
\begin{align}
 \lim_{L\to\infty}\Delta S_\eta(L,\ell)&=0,
 \label{eq:eta-inner-exact}\\
 \lim_{L\to\infty}\Delta S_{\rm can}(L,\ell)
 &=\log2+S_{4\ell-1}^{\circ}-S_{4\ell}.
 \label{eq:kw-inner-exact}
\end{align}
It remains to determine whether the odd block removes all, half, or none of
the explicit $\log2$ contribution as $\ell$ grows.  The answer is the
consecutive-size theorem below.

\subsection{Odd--even critical-Majorana entropy theorem}

The required limit is not the entropy of either block separately, but their
consecutive-size difference.  This distinction removes the common critical
logarithm and isolates the finite increment produced by the KW spectator.

\begin{theorem}[Consecutive Toeplitz entropy increment]
\label{thm:toeplitz-entropy-increment}
For the matrices in Eq.~\eqref{eq:critical-Toeplitz},
\begin{equation}
 \lim_{n\to\infty}\left(S_{2n}-S_{2n-1}^{\circ}\right)
 =\frac12\log2.
 \label{eq:toeplitz-entropy-increment}
\end{equation}
\end{theorem}

\begin{proof}
Put $H_N=iT_N$ and $D_N(z)=\det(zI-H_N)$.  The Toeplitz symbol of $H_N$ is
\begin{equation}
 h(e^{i\theta})=-\operatorname{sgn}(\theta),\qquad -\pi<\theta<\pi,
 \label{eq:sign-symbol}
\end{equation}
for the Fourier convention in Eq.~\eqref{eq:critical-Toeplitz}.  The opposite
sign merely reflects the symmetric spectrum.  For
$z\in\mathbb C\setminus[-1,1]$, the characteristic symbol
$f_z=z-h$ is nonzero and has two inverse pure jumps.  Its geometric mean is
\begin{equation}
 G(f_z)=\exp\left[\frac{1}{2\pi}\int_{-\pi}^{\pi}
 \log f_z(e^{i\theta})\,d\theta\right]
 =\sqrt{z^2-1},
 \label{eq:fh-geometric-mean}
\end{equation}
where the branch satisfies $\sqrt{z^2-1}\sim z$ at infinity.

For each fixed $z$ off the cut, the scalar Fisher--Hartwig theorem for two pure
jumps~\cite{DeiftItsKrasovsky2011} gives
\begin{equation}
 D_N(z)=G(f_z)^N N^{-2\beta(z)^2}E(z)[1+o(1)],
 \qquad
 \beta(z)=\frac{1}{2\pi i}\log\frac{z+1}{z-1},
 \label{eq:fh-asymptotic}
\end{equation}
where $\log$ denotes the principal logarithm and $E(z)$ is a nonzero analytic
prefactor.  The hypotheses and local-uniformity step are recorded in
Appendix~\ref{app:technical}.  Dividing consecutive sizes gives the pointwise
limit
\begin{equation}
 \frac{D_N(z)}{D_{N-1}(z)}\longrightarrow\sqrt{z^2-1}.
 \label{eq:determinant-ratio}
\end{equation}
To upgrade it, write the principal extension and its Schur complement as
\begin{equation}
 H_N=\begin{pmatrix}H_{N-1}&u_N\\u_N^*&0\end{pmatrix},
 \qquad
 \frac{D_N(z)}{D_{N-1}(z)}
 =z-u_N^*(zI-H_{N-1})^{-1}u_N.
 \label{eq:body-schur-ratio}
\end{equation}
The spectral support $\sigma(H_N)\subset[-1,1]$ and $\|u_N\|\leq1$ make
these ratios locally bounded on $\mathbb C\setminus[-1,1]$.
Vitali--Porter therefore upgrades the pointwise limit to local-uniform
convergence.  Both the ratios and $\sqrt{z^2-1}$ are zero-free off the cut,
so Cauchy's estimates imply
\begin{equation}
 \partial_z\log\frac{D_N(z)}{D_{N-1}(z)}
 \longrightarrow\frac{z}{z^2-1}.
 \label{eq:log-derivative-ratio}
\end{equation}

Let $\lambda_j^{(N)}$ be the eigenvalues of $H_N$ and define the one-step
signed spectral measure
\begin{equation}
 \mu_N=\sum_{j=1}^{N}\delta_{\lambda_j^{(N)}}
       -\sum_{j=1}^{N-1}\delta_{\lambda_j^{(N-1)}}.
 \label{eq:one-step-measure}
\end{equation}
Its Stieltjes transform is the left-hand side of
Eq.~\eqref{eq:log-derivative-ratio}; the limiting transform is
\begin{equation}
 \frac{z}{z^2-1}=\frac{1}{2(z-1)}+\frac{1}{2(z+1)}.
 \label{eq:endpoint-transform}
\end{equation}
An elementary bound on the total-variation norm is
$\|\mu_N\|_{\rm TV}\leq2N-1$.  Uniform control of the entropy functional
comes instead from Cauchy interlacing.  With $N_A(x)$ the number of eigenvalues
of $A$ not exceeding
$x$, set
\begin{equation}
 F_N(x)=N_{H_N}(x)-N_{H_{N-1}}(x).
 \label{eq:spectral-shift-counting}
\end{equation}
Then $F_N(x)\in\{0,1\}$.  Since $H_N$ is a finite compression of the
self-adjoint Toeplitz operator with symbol $\pm\operatorname{sgn}\theta$,
$\sigma(H_N)\subset[-1,1]$.  For every
$\varphi\in W^{1,1}(-1,1)$, Stieltjes integration by parts gives
\begin{equation}
 \begin{split}
 L_N(\varphi)&:={\rm Tr}\,\varphi(H_N)-{\rm Tr}\,\varphi(H_{N-1})\\
 &=\varphi(1)-\int_{-1}^{1}\varphi'(x)F_N(x)\,dx,
 \end{split}
 \quad
 |L_N(\varphi)|\leq |\varphi(1)|+\|\varphi'\|_1.
 \label{eq:w11-spectral-shift}
\end{equation}
Expanding Eq.~\eqref{eq:log-derivative-ratio} at infinity on a circle
$|z|=R>1$ gives convergence of $L_N(p)$ for every polynomial $p$ to
\begin{equation}
 L_\infty(p)=\frac12p(-1)+\frac12p(1).
 \label{eq:polynomial-endpoint-functional}
\end{equation}
Polynomials are dense in the norm
$\|\varphi\|_*=|\varphi(1)|+\|\varphi'\|_1$.  To see this, approximate
$\varphi'$ in $L^1$ by a polynomial $q$ and set
$p(x)=\varphi(1)+\int_1^xq(t)dt$.  The uniform bound in
Eq.~\eqref{eq:w11-spectral-shift} therefore extends
Eq.~\eqref{eq:polynomial-endpoint-functional} to all $W^{1,1}$ functions.
In particular,
\begin{equation}
 \eta'(x)=\frac12\log\frac{1-x}{1+x}\in L^1(-1,1),
 \qquad \eta(\pm1)=0,
 \label{eq:entropy-w11}
\end{equation}
so
\begin{equation}
 {\rm Tr}\,\eta(H_N)-{\rm Tr}\,\eta(H_{N-1})\longrightarrow0.
 \label{eq:consecutive-trace-entropy}
\end{equation}

Finally, reordering the even- and odd-indexed rows and columns separately
writes the off-diagonal part of $H_{2n}$ in terms of the square $n\times n$
Cauchy block
\begin{equation}
 B_{ab}=\frac{2}{\pi[2(b-a)+1]},\qquad 0\leq a,b<n,
 \label{eq:square-cauchy-block}
\end{equation}
which is nonsingular.  Thus $H_{2n}$ has no zero eigenvalue and
${\rm Tr}\,\eta(H_{2n})=2S_{2n}$.  The odd matrix contains the rectangular
$n\times(n-1)$ block with the same entries.  Its maximal minors are
nonsingular Cauchy matrices, so it has rank $n-1$ and $H_{2n-1}$ has nullity
exactly one.  Hence
\begin{equation}
 {\rm Tr}\,\eta(H_{2n-1})=2S_{2n-1}^{\circ}+\eta(0)
 =2S_{2n-1}^{\circ}+\log2.
 \label{eq:odd-zero-bookkeeping}
\end{equation}
Applying Eq.~\eqref{eq:consecutive-trace-entropy} along $N=2n$ proves
Eq.~\eqref{eq:toeplitz-entropy-increment}.
\end{proof}

\subsection{Physical ordered-limit theorem}

The explicit $\log2$ from the two real KW null directions is reduced by the asymptotic
odd--even entropy difference, leaving precisely half of that value, while the
matched $\eta$ and homogeneous blocks leave no constant difference.

Combining Eqs.~\eqref{eq:eta-inner-exact}, \eqref{eq:kw-inner-exact}, and
Theorem~\ref{thm:toeplitz-entropy-increment} yields the physical ordered-limit
theorem
\begin{equation}
 \lim_{\ell\to\infty}\left\{
 \lim_{\substack{L\to\infty\\L>2\ell}}
 \Delta S_\eta(L,\ell)\right\}=0,\qquad
 \lim_{\ell\to\infty}\left\{
 \lim_{\substack{L\to\infty\\L>2\ell}}
 \Delta S_{\rm can}(L,\ell)\right\}=\frac12\log2.
 \label{eq:boundary-entropy-conclusion}
\end{equation}
The sectorwise convergence in Eq.~\eqref{eq:body-sector-full-rdm-limit} also
gives the same KW constant in either fixed physical charge sector,
\begin{equation}
 \lim_{\ell\to\infty}\left\{
 \lim_{\substack{L\to\infty\\L>2\ell}}
 \left[S(\rho_{\omega,A_\ell}^{(L)})
 -S(\rho_{{\rm hom},A_\ell}^{(L)})\right]\right\}
 =\frac12\log2,
 \qquad \omega=\pm1.
 \label{eq:fixed-charge-entropy-corollary}
\end{equation}
Proposition~\ref{prop:ground-preparation-limit} extends this result to the
whole ground-state family.  Applying fixed-dimensional entropy continuity
and the same Toeplitz endpoint gives
\begin{equation}
 \lim_{\ell\to\infty}\left\{
 \lim_{\substack{L\to\infty\\L>2\ell}}
 \left[S(\rho_{A_{2\ell}}^{(L)}(p_L,c_L))
       -S(\rho_{{\rm hom},A_{2\ell}}^{(L)})\right]\right\}
 =\frac12\log2,
 \label{eq:general-ground-entropy}
\end{equation}
where $A_{2\ell}$ uses the spin-count notation of the proposition.
The second limit in Eq.~\eqref{eq:boundary-entropy-conclusion} concerns the finite-size physical convex sum
$\rho_{{\rm can},A_\ell}^{(L)}$.  It does not identify that RDM with
$\Gamma_{\omega,0}$ at any finite $L$.  Nor does exact Kramers doubling alone
fix the coefficient, because the abstract $I_2/2$ factor is partly compensated by the
active state, and the remaining $\tfrac12\log2$ follows from the independent
odd--even Toeplitz endpoint.  The theorem is restricted to complete proper
prefixes and the displayed nested order.  The next subsection tests this
analytic conclusion numerically; a later section identifies the infrared
defect sector from its spatial energy and momentum, independently of the
present entropy argument.

\subsection{Finite-size convergence}

The Fourier cycle kernel implemented in the reproducibility package illustrates
finite-size convergence independently of the proof.  Denote by $g_\eta$, $g_{\omega,0}^G$,
and $g_{\omega,{\rm phys}}^G$ the fitted outer-limit intercepts for the
$\eta$ entropy difference and the auxiliary and physical-sector Gaussian
comparisons.  The scale-separated fits use
$L\in\{12288,16384,24576,32768,49152,65536\}$ and
$\ell\in\{32,48,64,96,128\}$.  The resulting intercepts are
\begin{align}
 |g_\eta|&<10^{-11},\label{eq:geta-result}\\
 g_{\omega,0}^G&=0.3465735705443,\label{eq:gkwmixed-result}\\
 g_{\omega,{\rm phys}}^G&=0.3465736918532,
 \label{eq:gkwpure-result}
\end{align}
The reference value is
$\frac12\log2=0.3465735902800$.  The $\eta$ entry is reported as a
numerical-null tolerance because its fitted signed value is at the
numerical-cancellation scale; analytically its limit is exactly zero.
These values quantify finite-size
convergence; the logical basis of
Eq.~\eqref{eq:boundary-entropy-conclusion} is the analytic argument above.
Figure~\ref{fig:ordered-defect-entropy} displays the two numerical scales
separately.  Panel~(a) resolves the inner thermodynamic extrapolation of the
physical-sector Gaussian representation.  Rather than plotting the same
limiting constant again, panel~(b) shows the remaining finite-subsystem
corrections
\begin{align}
 \varepsilon_0(\ell)
 &=\frac12\log2-
   \lim_{L\to\infty}\Delta S_{\omega,0}^{G}(L,\ell),\nonumber\\
 \varepsilon_{\rm phys}(\ell)
 &=\frac12\log2-
   \lim_{L\to\infty}\Delta S_{\omega,{\rm phys}}^{G}(L,\ell).
 \label{eq:finite-subsystem-correction}
\end{align}
These auxiliary and physical-sector Gaussian corrections are consistent with
the same approximate $1/\ell$ decay over the fitted range and with their
common fixed-prefix limit.  They are convergence-rate checks, not finite-size
measurements of $S(\rho_{{\rm can},A_\ell}^{(L)})$.

\begin{figure}[t]
 \centering
 \includegraphics[width=\textwidth]{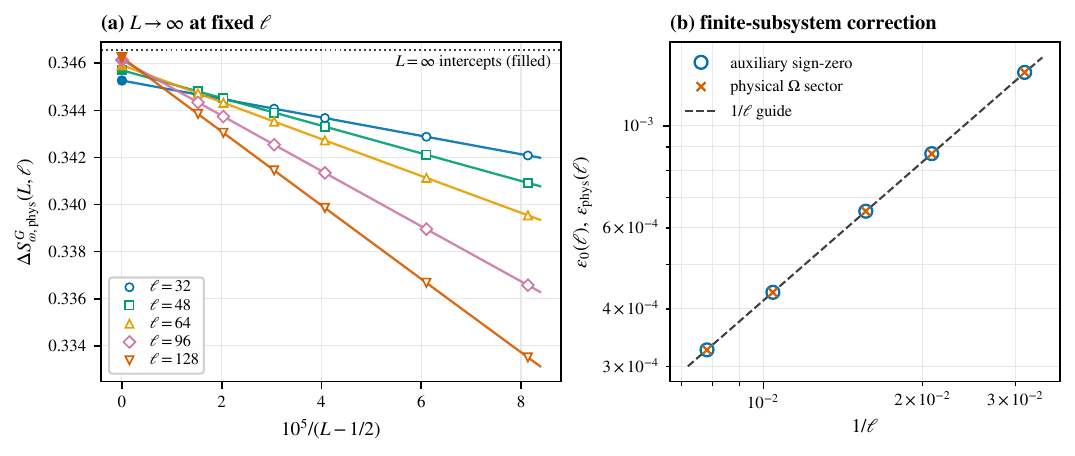}
 \caption{\textbf{Two-stage ordered defect-entropy convergence.}
 Panel~(a) takes the inner thermodynamic limit at fixed $\ell$ by fitting
 $\Delta S_{\omega,{\rm phys}}^{G}(L,\ell)$ against $1/(L-1/2)$ for the displayed
 $\ell=32,48,64,96,128$ curves.  Filled symbols are the resulting
 $L\to\infty$ intercepts.  Panel~(b) plots the residual
 $\varepsilon_0(\ell)$ and $\varepsilon_{\rm phys}(\ell)$ of
 Eq.~\eqref{eq:finite-subsystem-correction} for the auxiliary sign-zero and
 physical-sector Gaussian comparisons on logarithmic axes.  The blue open
 circles (auxiliary sign-zero) and orange crosses (physical $\Omega$ sector)
 are nested because the two residuals are indistinguishable on the plotted
 scale.  The dashed line is a $1/\ell$ reference guide anchored to the
 auxiliary residual at $\ell=64$.  The fitted grids are stated in the text.
 These extrapolations provide an independent
 finite-size check and do not enter the analytic proof.}
 \label{fig:ordered-defect-entropy}
\end{figure}

\subsection{Comparison with lattice and CFT reduced states}
\label{sec:entropy-comparison}

A defect entropy depends on how the global state was prepared, which region
was kept, and what was subtracted from it.  The
published Ising results use several combinations of these inputs, so the
comparisons below track those choices together with the constants.
Our region holds
$m=2\ell$ complete spin sites beginning at site~1.  Inverting the dictionary
of Eq.~\eqref{eq:roy-spin-dictionary} sends it to the suffix
$\{L-m+1,\ldots,L\}$ of Ref.~\cite{RoySaleur2022}, whose last site is the one
carrying no transverse field, and sends the selected charge to the reflected
label.  The interval swept in that paper's interface-entropy plot instead
keeps one end pinned at the defect, so the suffix has to be quoted as it
stands, rather than renamed after either the symmetric or the interface family
of that reference.

The zero-mode bookkeeping of Ref.~\cite{RoySaleur2022} is what makes its
symmetric data comparable at all.  Write $R=r/L$ and let
$\Delta S_{\rm R}(R)$ denote that reference's zero-mode correction function.
Its Figure~2 reports the mixed- and pure-state corrections
\begin{equation}
 \frac12\Delta S_{\rm R}(R)+\frac12\log2,
 \qquad
 \frac12\Delta S_{\rm R}(1-R),
 \label{eq:roy-symmetric-corrections}
\end{equation}
respectively, the first measured against the defect-free chain and the second
against the energy-defect curve, and both approach $\tfrac12\log2$ as
$R\to0$.  For the pure preparation, that reference finds the energy-defect
and defect-free entropies indistinguishable.  That common constant is attributed there to a single Majorana zero
mode localized inside the interval, and the factor $\tfrac12$ multiplying
$\Delta S_{\rm R}$ to the halved count of nonlocal zero modes of the duality
defect.  The left panel covers $10<r<100$ at $L=3000$, where the mixed and pure curves still lie on
top of one another, while the right panels extend to $r<L/2$ and there
separate the two preparations.  A shared small-$R$ limit
produced by one localized mode is thus compatible with unequal reduced states
at finite $R$, one of them a sector mixture and the other a filled or empty
zero-energy level; and since the complement identity that would trade $R$ for
$1-R$ requires a pure global state, the two entries of
Eq.~\eqref{eq:roy-symmetric-corrections} are not one function seen twice.

\begin{table}[ht]
\centering
\begin{tabular}{@{}>{\raggedright\arraybackslash}p{.30\textwidth}>{\raggedright\arraybackslash}p{.61\textwidth}@{}}
\toprule
Benchmark & Preparation, regime and comparison\\
\midrule
Roy--Saleur, Figure~2 & Symmetric intervals and their separate pure/mixed
zero-mode preparations; the corrections in
Eq.~\eqref{eq:roy-symmetric-corrections} share the small-$R$ limit.\\
Present Figure~\ref{fig:ordered-defect-entropy} & Auxiliary sign-zero and
physical fixed-charge Gaussian calculations, with the intercepts in
Eqs.~\eqref{eq:gkwmixed-result}--\eqref{eq:gkwpure-result}; they are not
finite-size measurements of the physical mixture.\\
Present physical RDM & Eqs.~\eqref{eq:boundary-entropy-conclusion},
\eqref{eq:fixed-charge-entropy-corollary} and
\eqref{eq:general-ground-entropy} establish the same constant for the
mixture, either charge ground state, and general ground-doublet preparations
in the stated nested limit.\\
Interface/half-system comparisons & Roy--Saleur Figure~3 and
Rogerson--Pollmann--Roy Eq.~(3.9) obtain the offset
$\tfrac12\log2-\tfrac14$ in their stated preparations, geometries and
reference conventions.\\
\bottomrule
\end{tabular}
\caption{The concrete lattice comparisons.  The periodic preparations of
Ref.~\cite{RoySaleur2022} and the pure-state, open/infinite-chain setups of
Ref.~\cite{Rogerson2022} are kept distinct from our finite periodic spin
mixture.  The effective length $L-\tfrac12$ enters the spectrum and may also
enter the chosen entropy reference; it does not by itself fix the excess.}
\label{tab:lattice-entropy-benchmarks}
\end{table}

Table~\ref{tab:lattice-entropy-benchmarks} keeps those inputs beside ours.
Its interface and half-system entries concern Roy--Saleur's periodic
$r=L/2$ comparison and Rogerson--Pollmann--Roy's left/right bipartitions
of open or infinite chains.  In the periodic case the entropy reference is a
defect-free chain of effective length $L-\tfrac12$.  For the open and infinite
chains, the scaling variable is the system size or the finite correlation
length of the numerical state, respectively.  At the short-fraction end, our result controls the complete RDM before
taking the outer entropy limit.  The computations in
Eqs.~\eqref{eq:gkwmixed-result}--\eqref{eq:gkwpure-result} remain estimates
from the auxiliary and physical-sector Gaussian representations on their
stated grids; they do not measure the finite physical RDMs of general coherent or
incoherent preparations.  Their common constant is governed by the theorem, not
established by the fits.

Seen from the continuum, the three choices reappear as the twisted state, the
defect network and the entangling boundary conditions, and they separate our
two lattice controls cleanly.  Ref.~\cite{NortheRossi2025} prepares the
invertible line from a pure $\mu$ twist completed by grouplike defects, and
obtains interval entropies that agree with the global vacuum at leading order;
this leading behavior agrees with the vanishing excess found here for the
matched spin-flip control, Eq.~\eqref{eq:eta-inner-exact}.  For the
non-invertible line the same reference prepares the pure state $\nu$.
In the symmetric geometry its short-interval operator-product expansion,
Eq.~(11), leaves the identity field on a duality line.  Fusing that line onto
an entangling boundary raises the entropy by a half-logarithm of two above
the untwisted value with the original boundaries, for any replica index.
The constant proved here is a von Neumann statement.  The agreement of
controlled leading regimes neither identifies the finite spin RDM with a
boundary CFT density matrix nor fixes the boundary condition corresponding
to our microscopic cut.

The fixed-charge pure branches were already covered by
Eq.~\eqref{eq:fixed-charge-entropy-corollary}; the new ingredient is control
of coherent preparations.  In the twisted-state construction of
Ref.~\cite{NorthePolettiRossi2026Interfaces}, Section~3.1, the state, network
and factorization are inputs.  Proposition~\ref{prop:ground-preparation-limit}
provides a whole lattice ground-state family against which such constructions
can be compared, once the cut, boundary conditions and relative phases are
matched.  Its scope does not include finite interval fractions.

A related distinction arises in the RCFT local-quench setting of
Ref.~\cite{HeEtAl2024SuperpositionTransitionRDM}, where cross terms disappear in the regulator
limit, while the late-time orthogonal diagonal branches retain the
preparation weights through a Shannon term.
Here the diagonal branches also approach the same RDM, eliminating that
remaining dependence.  The difference is in both steps of the reduction,
not simply in calling a global state pure or mixed.

The thermal construction in Section~4.3 of
Ref.~\cite{NorthePolettiRossi2026Interfaces} uses matching interface labels
in its displayed zero-temperature reduction.
This differs from directly choosing $p,c$ inside one KW ground doublet;
the general formal construction does not itself evaluate its coherent
cross terms.

The two-interface example changes the state as well as the network.
Ref.~\cite{NorthePolettiRossi2026Interfaces} places the identity field on a
network with \emph{two} interfaces, where the interfaces act as a filter on
the entanglement spectrum.  The reduced state of the duality-interface vacuum
field retains a single representation, whereas the bulk vacuum with the same
entangling boundaries retains all of them.  Filtering removes spectral weight,
and with the subtraction written in the order interface minus matched vacuum
the leading ultraviolet result of that reference is
\begin{equation}
 S_\alpha(\rho_{\rm interface})-S_\alpha(\rho_{\rm vacuum})
       =-\log2+o(1),\qquad \epsilon\to0.
 \label{eq:two-interface-entropy-comparison}
\end{equation}
Its Eq.~(5.14) displays the same relation with the subtraction in the opposite
order.  The character spaces being compared are not an isospectral pair at
finite modulus, so this $\log2$ is an asymptotic loss of spectrum under a
change of preparation, with the $o(1)$ recording what the cutoff leaves
behind.

The three operations compared here carry three references, and that is what
fixes their three signs.  Normalizing a local-parity block in
Section~\ref{sec:charged-parity} costs $\log2$ exactly and at every size,
because the state already carries two orthogonal isospectral blocks and the
deficit is mixing bookkeeping.  Filtering by interfaces costs $\log2$
asymptotically, because spectral weight is genuinely absent from the retained
state.  Replacing the homogeneous chain by the defect chain gives an excess
$\tfrac12\log2$ in the stated nested limit, because the odd--even endpoint
increment removes half of the explicit entropy contribution from the two
real null directions.  With the operation, the state and the
reference named, the signs and the magnitudes agree.

The resulting distinction is between finite-size preparation-sensitive
structure and a common local limit throughout this ground-state doublet.
The entropy theorem does not identify the continuum sector carrying the
lattice state; its spatial energy and momentum provide that independent
diagnostic in the next section.

\section{Defect-CFT data from the spatial sector}\label{sec:cft}

The global states behind the reduced density matrix can be identified on their
own terms, from the spectrum of the Hamiltonian that supports them.  We count
states in the full duality-twisted spin Hilbert space $\Hs$, recording their
energy levels together with their modified-translation eigenvalues.  The object computed is therefore a spectral
character of the full system, not the spectrum of a subregion density matrix.

A Casimir coefficient determines a scaling dimension while leaving more than
one momentum assignment open.  Resolving the conformal spin requires the
modified translation, which also keeps the two chiralities apart.  Data of this
kind is what a CFT reduced-state construction receives as its global-state
input, and no such comparison was assumed in the entropy proof.  We work
throughout in the conventions already fixed for the physical state and postpone
the continuum reading of these finite results to the end of the section.

\subsection{Fourier derivation of the exact modes}

The exact duality-twisted Ising spectrum and its continuum organization are
known~\cite{Grimm2002,AasenMongFendley2016}.  To identify the sector realized
by the present Hamiltonian, we express its finite-size spectrum in the same
effective-length, charge, and translation conventions used to construct the
RDM.

The active KW graph in Theorem~\ref{prop:odd-cycle} is an oriented cycle of odd
length
\begin{equation}
 N=2L-1.
 \label{eq:cycle-N}
\end{equation}
The fixed-charge closing sign is a genuine cycle flux and cannot be removed by
a real diagonal Majorana sign gauge.  In the active basis used in
Appendix~\ref{app:character-theorem}, the boundary condition is
$q_{r+N}=\omega q_r$.  The $\omega=+1$ block therefore has the periodic grid
$e^{ikN}=+1$ and its active zero mode at $k=0$, whereas the $\omega=-1$ block
has the antiperiodic grid $e^{ikN}=-1$ and its zero mode at $k=\pi$.  A common
labeling of the positive modes is
\begin{equation}
 \kappa_{\omega n}=
 \begin{cases}
  -\pi n/N,&(-1)^n=\omega,\\
  \pi+\pi n/N,&(-1)^n=-\omega,
 \end{cases}
 \qquad n=1,\ldots,L-1,
 \label{eq:charge-dependent-momenta}
\end{equation}
which obeys $e^{i\kappa_{\omega n}N}=\omega$.  The oriented-cycle eigenvalue
at momentum $k$ is $-4\sin k$; hence the labels in
Eq.~\eqref{eq:charge-dependent-momenta} give, in both charge blocks, the same
unordered positive-energy multiset
\begin{equation}
 \epsilon_n(L)=4\sin\frac{\pi n}{2L-1},
 \qquad n=1,\ldots,L-1.
 \label{eq:exact-positive-modes}
\end{equation}
Thus equality of the positive spectra follows from a relabeling of two
different Fourier grids, not from gauging both sectors to one periodic grid.
Together with the spectator, the single active zero mode gives the nullity two
used in Sections~\ref{sec:entanglement-states} and
\ref{sec:entanglement-locality}.

For comparison, the homogeneous chain in fixed physical parity gives an
antiperiodic even cycle.  Its $L$ positive eigenvalues are
\begin{equation}
 \epsilon_n^{\mathrm{hom}}(L)
 =4\sin\frac{(2n-1)\pi}{2L},
 \qquad n=1,\ldots,L.
 \label{eq:hom-positive-modes}
\end{equation}
These two mode sequences determine the finite-size vacuum energies and hence
the continuum scaling dimension carried by the KW sector.

\subsection{Exact finite-size ground energies}

For a real antisymmetric quadratic matrix $A$, our many-body ground-energy
convention is
\begin{equation}
 E_0(A)=-\frac12\sum_{\epsilon_k>0}\epsilon_k,
 \qquad \{\epsilon_k\}=\operatorname{spec}(iA).
 \label{eq:ground-energy-convention}
\end{equation}
Here $\operatorname{spec}(iA)$ denotes the spectrum of $iA$, and the sum runs
over its positive eigenvalues.  The trigonometric sums needed to evaluate the
homogeneous and KW ground energies are
\begin{align}
 \sum_{n=1}^{(N-1)/2}\sin\frac{\pi n}{N}
 &=\frac12\cot\frac{\pi}{2N},\label{eq:odd-sine-sum}\\
 \sum_{n=1}^{L}\sin\frac{(2n-1)\pi}{2L}
 &=\csc\frac{\pi}{2L},\label{eq:anti-sine-sum}
\end{align}
Together with Eqs.~\eqref{eq:exact-positive-modes} and
\eqref{eq:hom-positive-modes}, these identities give
\begin{align}
 E_{\mathrm{hom}}(L)&=-2\csc\frac{\pi}{2L},
 \label{eq:Ehom-exact}\\
 E_{\mathrm{KW}}(L)&=-\cot\frac{\pi}{4L-2}
 =-\cot\frac{\pi}{4\Leff}.
 \label{eq:Ekw-exact}
\end{align}
The formulas hold at every finite size.  Direct
diagonalization at ten even system sizes with $4\leq L\leq512$ reproduces
Eqs.~\eqref{eq:Ehom-exact} and \eqref{eq:Ekw-exact} to
$1.14\times10^{-13}$.

Expanding at large circumference gives
\begin{align}
 E_{\mathrm{hom}}(L)
 &=-\frac4\pi L-\frac{\pi}{6L}
 -\frac{7\pi^3}{1440L^3}
 -\frac{31\pi^5}{241920L^5}+O(L^{-7}),
 \label{eq:Ehom-series}\\
 E_{\mathrm{KW}}(\Leff)
 &=-\frac4\pi\Leff+\frac{\pi}{12\Leff}
 +\frac{\pi^3}{2880\Leff^3}
 +\frac{\pi^5}{483840\Leff^5}+O(\Leff^{-7}).
 \label{eq:Ekw-series}
\end{align}
For a critical spatial circle in sector $a$, let $L_a$ denote its effective
circumference, $e_\infty$ the bulk energy density, and $\Delta_a$ its lowest
scaling dimension.  The finite-size vacuum energy has the form
\begin{equation}
 E_0^{(a)}(L_a)=e_\infty L_a+
 \frac{2\pi v}{L_a}\left(\Delta_a-\frac{c}{12}\right)+\cdots.
 \label{eq:circle-Casimir}
\end{equation}
where $c$ and $v$ are the central charge and velocity.  With $c=1/2$, $v=2$,
and the $\pi/(12\Leff)$ coefficient of
Eq.~\eqref{eq:Ekw-series},
\begin{equation}
 \Delta_\sigma=\frac1{16}.
 \label{eq:Delta-sigma}
\end{equation}
This is the known Ising duality-sector dimension from conformal defect
theory~\cite{FrohlichFuchsRunkelSchweigert2004,PetkovaZuber2001} and the exact twisted-chain
spectrum~\cite{Grimm2002}, now recovered in the present lattice conventions.

If Eq.~\eqref{eq:Ekw-series} is reexpanded in the bare variable $L$, the bulk
term produces a constant $2/\pi$.  The effective-length form above is the
natural parametrization for the spatial defect circle.

\subsection{Joint energy--translation character}

The Casimir energy fixes the lowest scaling dimension but not its conformal
spin.  The missing finite-size observable is the modified translation
$T_\sigma$ of Eq.~\eqref{eq:Tsigma-circuit}.  Its exact Clifford power law is
\begin{equation}
 T_\sigma^{N}=e^{i\pi\Omega/4},
 \qquad N=2L-1.
 \label{eq:Tsigma-power}
\end{equation}
This identity constrains the possible roots of a translation eigenvalue, but
it does not by itself show which roots occur, with what multiplicity, or in
which energy eigenspaces.  Those questions require the action of $T_\sigma$
on a complete physical basis.

Restrict to the charge block $\Omega=\omega$, $\omega=\pm1$.  The bare odd
Fourier creator changes charge and therefore is not an operator within this
block.  Pairing it with the spectator Majorana gives the charge-preserving CAR
mode $\psi_{\omega n}^\dagger$ of
Appendix~\ref{app:character-theorem}, whose exact translation action is
\begin{equation}
 U_\omega\psi_{\omega n}^\dagger U_\omega^\dagger
 =e^{2\pi i\omega(-1)^n n/N}\psi_{\omega n}^\dagger,
 \qquad U_\omega=P_\omega T_\sigma P_\omega.
 \label{eq:body-mode-translation}
\end{equation}
For an occupied positive-mode set $A\subseteq\{1,\ldots,L-1\}$, the physical
zero-mode occupation is fixed rather than doubled, with
\begin{equation}
 n_0(A)=|A|\pmod2.
 \label{eq:body-zero-selector}
\end{equation}
The resulting $2^{L-1}$ states form a basis of the charge block.  Comparing
the exact circuit trace with the second-quantized Fock trace fixes the vacuum
eigenvalue
\begin{equation}
 t_{0,\omega}=e^{i\pi\omega/(4N)}.
 \label{eq:body-vacuum-root}
\end{equation}
The power law in Eq.~\eqref{eq:Tsigma-power} is consistent with this scalar but
would leave an $N$th-root ambiguity if used alone.

Define the energy and cyclic translation labels
\begin{equation}
 E_{\omega,L}(A)=E_{0,L}+\sum_{n\in A}\epsilon_n(L),
 \qquad
 r_\omega(A)=\sum_{n\in A}\omega(-1)^n n.
 \label{eq:body-energy-root-labels}
\end{equation}
Then every physical Fock state has the exact relative translation root
\begin{equation}
 \frac{t_{\omega,L}(A)}{t_{0,\omega}}
 =e^{2\pi i r_\omega(A)/N}.
 \label{eq:body-root-ratio}
\end{equation}
Equivalently, the finite joint occupation--root character is
\begin{equation}
 \mathcal C_{\omega,L}(\mathbf x,z)
 =\prod_{n=1}^{L-1}\left(1+x_nz^{\omega(-1)^nn}\right),
 \qquad z^N=1.
 \label{eq:body-finite-joint-character}
\end{equation}
Its monomials record actual states, and after equal energies and roots are
grouped its coefficients give the exact joint multiplicities, including
accidental energy collisions.  The character proves occurrence of the roots
it contains; it does not claim saturation of every algebraically allowed
$N$th root.

To choose a representative of a cyclic root without assigning parity to an
equivalence class on which it is not defined,
use the centered representative
\begin{equation}
 k_{\omega,L}(A)=\operatorname{cent}_N r_\omega(A)
 \in\{-(L-1),\ldots,L-1\}.
 \label{eq:body-centered-root}
\end{equation}
With $L_{\rm eff}=N/2$, the convention-fixed spin class is
\begin{equation}
 s(\omega,k)=\frac{\omega}{16}+\frac{k}{2}\pmod1.
 \label{eq:spin-set}
\end{equation}
The centered section is not globally additive across a wrap.  Even so, the explicit finite-size states with $A=\varnothing$ and $A=\{1\}$ are
present at every $L\geq2$.  Such states are what we call \emph{finite
witnesses}, and their function is to certify that the stated spin classes are
genuinely occupied in the physical spectrum rather than merely admitted by the
algebra.  Between them they realize all four classes
\begin{equation}
 \left\{\frac1{16},-\frac1{16},\frac7{16},-\frac7{16}\right\}\pmod1.
 \label{eq:body-four-spin-witnesses}
\end{equation}
Thus occurrence and multiplicity follow from the physical joint character,
not from the power relation alone.  Appendix~\ref{app:character-theorem}
gives the fixed-charge CAR construction, circuit-trace normalization, and
finite witnesses in detail.

\subsection{Marked scaling limit and Virasoro towers}

The same Fourier spectrum controls the excitation energies.  At fixed mode
index,
\begin{equation}
 d_{n,L}=\frac{N}{2\pi}\sin\frac{\pi n}{N}\longrightarrow\frac n2.
 \label{eq:scaled-modes}
\end{equation}
Energy counting alone therefore gives the diagonal product
\begin{equation}
 2\prod_{n\geq1}(1+x^{n/2}),
 \label{eq:Fock-product}
\end{equation}
but this specialization has erased the translation marks.  It cannot identify
chirality, the signs of conformal spin, or the placement of the primary
pairs.

Retaining the centered root before taking the limit supplies the missing
information.  For a finite occupation set $A$, define the marked excitation
coordinates
\begin{equation}
 p_{\omega,L}(A)=\frac{k_{\omega,L}(A)}2,
 \qquad
 u_{\omega,L}(A)=\frac{D_L(A)+p_{\omega,L}(A)}2,
 \qquad
 v_{\omega,L}(A)=\frac{D_L(A)-p_{\omega,L}(A)}2,
 \label{eq:body-marked-coordinates}
\end{equation}
where $D_L(A)=\sum_{n\in A}d_{n,L}$.  These are finite-size spectral markers,
not eigenvalues of separately defined finite-size chiral Hamiltonians.  Below
every fixed excitation cutoff, only finitely many mode sets occur and their
uncentered labels cannot wrap for all sufficiently large $L$.  The associated
locally finite marked measures therefore converge vaguely to coordinates
\begin{equation}
 u_\omega(A)=\frac14\sum_{n\in A}[n+\omega(-1)^nn],
 \qquad
 v_\omega(A)=\frac14\sum_{n\in A}[n-\omega(-1)^nn].
 \label{eq:body-limiting-coordinates}
\end{equation}
For $\omega=+1$, even modes populate the left chirality and odd modes the
right; for $\omega=-1$ the assignment is reversed.  Hence
\begin{align}
 G_+(q,\bar q)&=\prod_{m\geq1}(1+q^m)
                 \prod_{m\geq1}(1+\bar q^{m-1/2}),\label{eq:body-Gplus}\\
 G_-(q,\bar q)&=\prod_{m\geq1}(1+q^{m-1/2})
                 \prod_{m\geq1}(1+\bar q^m).
 \label{eq:body-Gminus}
\end{align}
Conditional on the standard Ising identification $c=1/2$ and the
effective-length, velocity, and right-translation conventions fixed above, the vacuum witnesses
combine the Casimir value $\Delta_\sigma=1/16$ with the translation orientation
to give
$(h_{+,0},\bar h_{+,0})=(1/16,0)$ and
$(h_{-,0},\bar h_{-,0})=(0,1/16)$.  Using the standard Ising/Jacobi product identities
\begin{equation}
 \chi_\sigma(q)=q^{1/24}\prod_{m\geq1}(1+q^m),
 \qquad
 \chi_1(q)+\chi_\epsilon(q)
 =q^{-1/48}\prod_{m\geq1}(1+q^{m-1/2}),
 \label{eq:body-Ising-products}
\end{equation}
we resolve the Neveu--Schwarz factor by the parity of occupied odd mode
indices in both charge blocks and obtain the continuum character of the
full twisted system,
\begin{equation}
 Z_{\rm KW}(q,\bar q)
 =\chi_\sigma(\bar\chi_1+\bar\chi_\epsilon)
  +(\chi_1+\chi_\epsilon)\bar\chi_\sigma.
 \label{eq:body-Virasoro-decomposition}
\end{equation}
The four primary pairs are
\begin{equation}
 (\sigma,1),\quad(\sigma,\epsilon),\quad
 (1,\sigma),\quad(\epsilon,\sigma),
 \label{eq:body-four-primary-pairs}
\end{equation}
with finite witnesses $A=\varnothing$ and $A=\{1\}$ in the two charge
blocks.  The one-mode witness is the character-theoretic bottom of the
$\epsilon$ module in the relevant chirality, not a descendant of the identity
primary.

Diagonal specialization of $G_++G_-$ recovers Eq.~\eqref{eq:Fock-product} and erases the
chiral assignment, so energy multiplicities alone cannot produce
Eq.~\eqref{eq:body-Virasoro-decomposition}.  The result is a scaling
spectral-counting and Virasoro-character theorem with the precise finite-size
and scaling scope stated in Appendix~\ref{app:character-theorem}.  That
appendix also proves eventual no-wrap, vague convergence, the NS parity
projection, and the finite witnesses.

\subsection{Identification of the duality-twisted sector}

What the spatial Hamiltonian yields is continuum data of several complementary
kinds.  Its Casimir energy locates the lowest scaling dimension; the exact
joint energy--translation character exhibits which spin classes occur and with
what finite multiplicities; the marked limit preserves their chiral
assignments.  Between them these recover the known Ising duality-twisted sector
\cite{Grimm2002,AasenMongFendley2016} within the dressing and right-translation
convention adopted here.

Under the Ising identification set out above, the two ground weights come out
as $(1/16,0)$ and $(0,1/16)$.  They correspond to the $\nu$ and $\bar\nu$
twist-field labels used in the CFT reduced-state construction of
Ref.~\cite{NortheRossi2025}, and the correspondence is one of global-state data
supplied to that construction.  The prepared state and the
factorization imposed at the interval endpoints remain separate
inputs~\cite{OhmoriTachikawa2015EntanglingSurface}, just as the global
preparation and the retained spin algebra do in
Section~\ref{sec:entanglement}.  Comparing coherent superpositions also
requires matching the relative phase of the two twist components under
that cutting map; the two ground weights alone do not specify it.

Consider a conformal single-interval factorization.  There the boundary
conditions tie the left and right stress tensors to one another, and the result
is a chiral BCFT state space with a preserved algebra
\cite{DiGiulioEtAl2022BCFTSymmetryResolved,Northe2023ConformalResolution}.  For
the vacuum, and for the specified sector-filter constructions, the associated
thermal trace can be written as a character sum or as a projected character.
Primary and twisted preparations generally retain operator or defect-network
insertions in their RDMs and moments \cite{NortheRossi2025,NorthePolettiRossi2026Interfaces}; an
undressed boundary character on its own does not fix such spectra.  The trace
in Eq.~\eqref{eq:body-Virasoro-decomposition} runs over the full circle and
retains both chiralities.  What it identifies is the global sector out of which
a reduced state is prepared, rather than the eigenvalues of the interval RDM
obtained from it.

Three operations are therefore in play, distinct from one another and yet
linked.  The spatially wrapped line operator settles the fusion identity; the
twisted Hamiltonian together with modified translation identifies the global
states; ordinary restriction to the chosen local algebra produces the physical
reduced states.  It is the last of these that the entropy and local-parity
results rest on, while the character obtained here supplies their independent
global-sector context.

\section{Discussion}
\label{sec:discussion}

An algebraic defect label and a continuum-sector assignment identify what
propagates around the circle, but neither specifies the reduced state seen by
a chosen microscopic subsystem.  For the KW line, the MPO and spatial
Hamiltonian fix the fusion algebra and twisted Hilbert space.  Turning those
global data into a physical RDM still requires a charge-sector prescription
and the local spin algebra of the entangling region.  Only after these choices
are made does the complete non-wrapping prefix admit the physical RDM analyzed
here.  Its ordered-limit entropy is consequently derived from a microscopic
state rather than assigned from its continuum identity, giving
$\tfrac12\log2$ for KW and no shift for the matched invertible $\eta$ defect.
Energy and modified translation instead provide an exact joint character that
reconstructs, using the standard Ising character identities, the four chiral
towers of the spatial KW Hamiltonian resolved by charge block, independently
of the equal-weight preparation and without using the entropy to identify the
sector.

The microscopic origin of the entropy constant lies in the full limiting
reduced state.  The two physical charge-sector RDMs converge separately
to the same fixed-prefix Gaussian comparator, and the cross-sector reduced
operator vanishes.  Every density matrix supported on the ground-state
doublet, including coherent pure states, therefore shares this local limit.  In the resulting Toeplitz problem, the two real null
directions---the spectator and the odd-cycle zero mode---supply an explicit
$\log2$, while the consecutive odd--even critical-block increment removes
$\tfrac12\log2$ and leaves the physical KW constant.  The matched $\eta$ and
homogeneous covariances instead approach the same even block and leave no
constant difference.  At finite size, the reduced antiunitary symmetry enforces exact many-body
Kramers pairing in the RDM of the equal-weight mixture.  This pairing
organizes the spectrum without fixing the entropy coefficient.

This construction connects lattice and continuum descriptions of the Ising
duality defect.  Its lattice algebra, modified translation, and twisted
spectrum were known~\cite{AasenMongFendley2016,Grimm2002}.  Lattice studies
had characterized entropies and negativity for Ising chains with
defects~\cite{RoySaleur2022,Rogerson2022}.  At the reduced-state level, Rockwood
exposed the nonlocal entanglement Hamiltonian induced by the KW zero mode in
periodic free-fermion chains~\cite{Rockwood2025}, while CFT work has
constructed defect-dressed RDMs and evaluated R\'enyi and von Neumann
entropies for selected twisted-state and interface
preparations~\cite{NortheRossi2025,NorthePolettiRossi2026Interfaces}.
The physical spin RDM derived here joins these descriptions for specified preparations
in the spatial KW sector, and the invertible $\eta$ line tests the entropy mechanism against a defect
with no constant shift.  The joint energy--translation theorem sharpens the
lattice identification by retaining the translation roots that actually occur
and their multiplicities.  Its marked scaling limit reconstructs chirality
and the four Virasoro towers rather than only the energy degeneracies.
For coherent preparations, matching the cutting map and relative phase
remains a further requirement for the lattice--CFT comparison.

The distinction between global identification and local construction extends
beyond the free-fermion solution.
The line algebra and twisted kinematics constrain which sectors are available;
a state prescription selects a density operator, and the subsystem algebra
determines what is traced out.  Infrared data can identify the resulting
sector, but they cannot supply any of those microscopic choices.  An exact
solution of the present problem, however, depends on Gaussian Ising structure.
That structure produces the localized spectator and odd Majorana cycle,
enables the sectorwise Pfaffian selection and fixed-prefix full-RDM control,
and permits the Toeplitz reduction.
Consequently, the spin--Gaussian theorem is restricted to complete
non-wrapping prefixes, while the physical state in each charge sector converges on
each fixed prefix to its auxiliary sign-zero Gaussian comparator only in
the ordered fixed-block limit, not at fixed $\ell/L$.  The spatial result
has a different scope.  Its exact finite-size joint character and marked
no-wrap scaling limit establish a spectral-counting and Virasoro-character
theorem.

Beyond Gaussian solvability, the next test is to construct a physical RDM
directly from microscopic states and determine its universal entropy or level
organization while using spatial symmetry and spectral data to recover the
continuum sector independently.  Such a construction would separate the
structural entanglement content of non-invertible symmetry from the solvable
Ising mechanism and provide a microscopic counterpart to continuum
defect-dressed RDMs and categorical symmetry resolution
\cite{NortheRossi2025,Saura2024,Das2024,ChoiRayhaunZheng2024,HeymannQuella2025}.

\appendix
\section{Conventions, explicit matrices, and special sectors}\label{app:explicit}

\subsection{Fixed-charge antisymmetric matrix in component form}

In the ordered Majorana basis $\gamma=(a_1,b_1,\dots,a_L,b_L)^{\mathsf T}$
of Eq.~(\ref{eq:gamma-order}), the homogeneous critical Hamiltonian
$H^{\mathrm{hom}}$ of Eq.~(\ref{eq:Hhom}) is represented, in a fixed
$\Omega$ sector, by the nonzero upper-triangular entries
\begin{equation}
 A^{\mathrm{hom}}_{a_j,b_j}=2,
 \qquad
 A^{\mathrm{hom}}_{b_j,a_{j+1}}=2\quad(1\le j<L),
 \qquad
 A^{\mathrm{hom}}_{a_1,b_L}=2\Omega,
 \label{appA:homogeneous-matrix}
\end{equation}
with the lower-triangular entries fixed by antisymmetry.  This is simply the
component form of $-X_j=i a_jb_j$, $-Z_jZ_{j+1}=i b_ja_{j+1}$, and
$-Z_LZ_1=-i\Omega b_La_1$.  In particular, the intersite coupling is a
single matrix element rather than a full $2\times2$ site block.

For the spatial defect Hamiltonian $H^{\mathrm d}$ of Eq.~(\ref{eq:Hdef})
in a fixed $\Omega$ sector, the active odd-cycle basis
$(b_1,a_2,b_2,\dots,a_L,b_L)$ of Eq.~(\ref{eq:odd-cycle-order}) gives the
$(2L-1)\times(2L-1)$ antisymmetric matrix $A^{\mathrm d}_{\Omega}$ with
entries
\begin{equation}
 (A^{\mathrm d}_{\Omega})_{m,m+1}=2
 \quad(1\le m<2L-1),
 \qquad
 (A^{\mathrm d}_{\Omega})_{2L-1,1}=2\Omega,
 \label{appA:defect-matrix}
\end{equation}
and all other entries zero.  The spectator $a_1$ corresponds to an all-zero
row and column in the full $2L\times2L$ matrix.  The odd cycle has exactly one
zero eigenvalue for every $L$.

\subsection{Jordan--Wigner boundary-sign table}

Table~\ref{tab:jw-signs} records the Pauli-to-Majorana identifications used
in the fixed-charge defect representation.

\begin{table}[ht]
\centering
\begin{tabular}{@{}lll@{}}
\toprule
Spin operator & Majorana form & Sites involved\\
\midrule
$-X_j$ & $i a_j b_j$ & $j$\\
$-Z_j Z_{j+1}$ & $i b_j a_{j+1}$ & $j,j+1$\\
$-Z_L Y_1$ & $i\Omega b_L b_1$ & $1,L$\\
$\Omega$ & $\prod_j X_j$ & all\\
\bottomrule
\end{tabular}
\caption{Fixed-charge Jordan--Wigner identities in the dressed defect
convention of Eq.~(\ref{eq:Hdef}).  The closing sign depends on the global
charge $\Omega=\pm1$.}
\label{tab:jw-signs}
\end{table}

Reversing the dressed defect sign (replacing $-Z_L Y_1$ by $+Z_L Y_1$)
reverses the closing orientation of the odd cycle but preserves its
single-particle spectrum.  Reversing the right-translation convention
reverses the virtual-index ordering in the MPO kernel and interchanges
$T$ and $T^{-1}$ in the fusion law.

\subsection{Line orientation and quantization}

Let $x$ run along the spatial circle and $\tau$ along Euclidean time.  A line
wound along $x$ lies within one time slice and so appears as an insertion in a
time-evolution amplitude on $\Hu$; a line parallel to $\tau$ cuts through every
slice and replaces the slice Hilbert space by the twisted space $\Hs$.  Which
finite operator belongs to which quantization is tabulated in
Table~\ref{tab:orientation-dictionary}; the slices themselves are drawn in
Figure~\ref{fig:kw-quantization}.

\begin{table}[ht]
\centering
\begin{tabular}{@{}lll@{}}
\toprule
Object & Spatially wrapped line & Time-directed line\\
\midrule
Line direction & Spatial circle $x$ & Euclidean time $\tau$\\
Hilbert space & $\Hu$ (untwisted) & $\Hs$ (twisted)\\
Operator & $D_\sigma:\Hu\to\Hu$ & $H^{\mathrm d},T_\sigma:\Hs\to\Hs$\\
Role & Operator insertion and fusion & State preparation and spectrum\\
\bottomrule
\end{tabular}
\caption{The two quantizations of the lattice KW line, in standard spacetime
terminology.  Interchanging the two Euclidean cycles changes the Hilbert-space
problem rather than merely relabelling it.  The twisted-sector operator
$H^{\mathrm d}$ remains a Hamiltonian on a spatial chain, a different object
from the line insertion $D_\sigma$.}
\label{tab:orientation-dictionary}
\end{table}

\section{Technical analytic lemmas}\label{app:technical}

\subsection{Fisher--Hartwig hypotheses for the sign symbol}

The Toeplitz symbol $h(e^{i\theta})=-\operatorname{sgn}\theta$ of
Eq.~(\ref{eq:sign-symbol}) has two pure jump discontinuities at
$\theta=0$ and $\theta=\pm\pi$.  For the characteristic symbol
$f_z(e^{i\theta})=z-h(e^{i\theta})$ with $z\in\mathbb C\setminus[-1,1]$,
the jump parameters may be chosen as
\begin{equation}
 \alpha_0=\alpha_\pi=0,
 \qquad
 \beta_0(z)=\beta(z),\quad \beta_\pi(z)=-\beta(z),
 \qquad
 \beta(z)=\frac{1}{2\pi i}\log\frac{z+1}{z-1},
 \label{appB:jump-parameters}
\end{equation}
in the notation of the scalar Fisher--Hartwig theorem
\cite{DeiftItsKrasovsky2011}.  For the principal logarithm,
$-1/2<\operatorname{Re}\beta(z)<1/2$ throughout
$\mathbb C\setminus[-1,1]$.  Hence on every compact set $K$ in that domain
there is a $\delta_K>0$ such that
\begin{equation}
 \sup_{z\in K}|\operatorname{Re}\beta(z)|\le\frac12-\delta_K,
 \qquad
 \inf_{z\in K}|E(z)|>0.
 \label{appB:domain}
\end{equation}
The Fisher--Hartwig formula therefore gives the consecutive determinant-ratio
limit for each fixed $z$ off the cut.  Local uniformity is obtained below from
the exact Schur complement rather than assumed from parameter-uniform
Fisher--Hartwig asymptotics.

\subsection{Cauchy estimates for the log-derivative ratio}

Define the unnormalized consecutive ratio
\begin{equation}
 \widehat R_N(z)=\frac{D_N(z)}{D_{N-1}(z)}.
 \label{appB:normalized-ratio}
\end{equation}
For the principal extension
$H_N=\left(\begin{smallmatrix}H_{N-1}&u_N\\u_N^*&0\end{smallmatrix}\right)$,
the Schur complement gives
\begin{equation}
 \widehat R_N(z)=z-u_N^*(zI-H_{N-1})^{-1}u_N.
 \label{appB:schur-ratio}
\end{equation}
The spectral support of every $H_N$ lies in $[-1,1]$ and $\|u_N\|\leq1$.
The Hermitian resolvent estimate therefore makes $\{\widehat R_N\}$ locally
bounded on the cut plane.  The fixed-parameter Fisher--Hartwig limit and
Vitali--Porter now imply
\begin{equation}
 \widehat R_N(z)\longrightarrow G(f_z)=\sqrt{z^2-1}
 \quad\text{locally uniformly on }\mathbb C\setminus[-1,1].
 \label{appB:vitali-ratio}
\end{equation}
Both sides are zero-free there.  On compact sets, local uniform convergence
therefore supplies a uniform lower bound for $|\widehat R_N|$ at large $N$;
Cauchy's derivative estimate gives
\begin{equation}
 \frac{\widehat R_N'(z)}{\widehat R_N(z)}
 \longrightarrow\frac{z}{z^2-1}.
 \label{appB:cauchy}
\end{equation}
This proves Eq.~(\ref{eq:log-derivative-ratio}) locally uniformly.  All Cauchy
disks stay inside the cut plane; no contour crosses the spectrum.

\subsection{Polynomial density in $W^{1,1}(-1,1)$}

Let $\|\varphi\|_*=|\varphi(1)|+\|\varphi'\|_{L^1(-1,1)}$.  For any
$\varphi\in W^{1,1}(-1,1)$, approximate $\varphi'$ in $L^1$ by a polynomial
$q$ with $\|\varphi'-q\|_{L^1}\le\varepsilon$.  Define
\begin{equation}
 p(x)=\varphi(1)+\int_1^x q(t)\,dt.
 \label{appB:poly-approx}
\end{equation}
Then $p$ is a polynomial with $p(1)=\varphi(1)$ and $p'=q$.  Moreover,
\begin{equation}
 \|p-\varphi\|_*
 =|(p-\varphi)(1)|+\|p'-\varphi'\|_{L^1}
 =\|q-\varphi'\|_{L^1}\le\varepsilon.
 \label{appB:approx-bound}
\end{equation}
Thus polynomials are dense in $(W^{1,1},\|\cdot\|_*)$.  The uniform bound
$|L_N(\varphi)|\le\|\varphi\|_*$ of Eq.~(\ref{eq:w11-spectral-shift})
then extends the limit functional $L_\infty(p)=\frac12p(-1)+\frac12p(1)$
from polynomials to all $W^{1,1}$ functions by continuity.  In particular,
Eq.~(\ref{eq:entropy-w11}) applies to $\eta$ because
$\eta'(x)=\frac12\log[(1-x)/(1+x)]\in L^1(-1,1)$ and $\eta(\pm1)=0$.

\subsection{Finite-sum trigonometric identities}

The sine sums used in the exact ground-energy formulas are derived from the
standard telescoping identity
\begin{equation}
 \sum_{n=1}^{N-1}\sin(n\theta)
 =\frac{\sin\frac{N\theta}{2}\,\sin\frac{(N-1)\theta}{2}}
      {\sin\frac{\theta}{2}}.
 \label{appB:telescope}
\end{equation}
Setting $N=(M+1)/2$ and $\theta=\pi/M$ gives the odd-cycle sum
\begin{equation}
 \sum_{n=1}^{(M-1)/2}\sin\frac{\pi n}{M}
 =\frac12\cot\frac{\pi}{2M},
 \label{appB:odd-sum}
\end{equation}
which is Eq.~(\ref{eq:odd-sine-sum}) with $M=2L-1$.  Setting
$N=L+1$ and $\theta=\pi/L$ with the antiperiodic shift gives
\begin{equation}
 \sum_{n=1}^{L}\sin\frac{(2n-1)\pi}{2L}
 =\csc\frac{\pi}{2L},
 \label{appB:even-sum}
\end{equation}
which is Eq.~(\ref{eq:anti-sine-sum}).  Both are exact for all $L\ge2$.

\subsection{Cauchy-block rank and determinant formulas}

The even critical-Majorana block $H_{2n}=iT_{2n}$ can be reordered into a
$2\times2$ Hermitian block matrix with $n\times n$ blocks
\begin{equation}
 H_{2n}=\begin{pmatrix}0&iB\\-iB^{\mathsf T}&0\end{pmatrix},
 \qquad
 B_{ab}=\frac{2}{\pi[2(b-a)+1]},\quad0\le a,b<n.
 \label{appB:even-block}
\end{equation}
The real matrix $B$ is of Cauchy type because, with $x_a=2a$ and $y_b=2b+1$,
$B_{ab}=(2/\pi)/(y_b-x_a)$.  Cauchy's determinant formula gives
\begin{equation}
 \det B
 =\left(\frac{2}{\pi}\right)^n
 \frac{\prod_{a<a'}(x_{a'}-x_a)\prod_{b<b'}(y_b-y_{b'})}
      {\prod_{a,b}(y_b-x_a)}
 \neq0,
 \label{appB:cauchy-det}
\end{equation}
because the even and odd node sets are pairwise distinct.  Hence $B$ is
nonsingular and $H_{2n}$ has no zero eigenvalue.  The odd block
$H_{2n-1}$ has the analogous $n\times(n-1)$ off-diagonal block; each maximal
minor is a nonsingular Cauchy determinant, so that block has rank $n-1$ and
the full $(2n-1)\times(2n-1)$ Hermitian matrix has nullity exactly one.  This
justifies the bookkeeping in Eq.~(\ref{eq:odd-zero-bookkeeping}).

\subsection{Clifford circuit power-law derivation}

The modified-translation power identity of Eq.~(\ref{eq:Tsigma-power}) follows
directly from Clifford-circuit composition.  The ordered circuit is
\begin{equation}
 T_\sigma=V\,T\,\bigl(H_L\,\mathrm{CZ}_{L,1}\bigr)V^\dagger,
 \qquad V=e^{-i\pi Z_1/4},
 \label{appB:Tsigma-circuit}
\end{equation}
with products acting from right to left.  For the dressing rotation,
\begin{equation}
 V:\quad X_1\mapsto Y_1,\qquad
 Y_1\mapsto -X_1,\qquad Z_1\mapsto Z_1
 \label{appB:V-action}
\end{equation}
under $P\mapsto VPV^\dagger$.  The other elementary actions are
\begin{equation}
 \begin{aligned}
 H_L:&\quad Z_L\mapsto X_L,\ X_L\mapsto Z_L,\\
 \mathrm{CZ}_{L,1}:&\quad Z_L\mapsto Z_L,\ Z_1\mapsto Z_1,\
 X_L\mapsto X_LZ_1,\ X_1\mapsto Z_LX_1,\\
 T:&\quad X_j\mapsto X_{j+1},\ Z_j\mapsto Z_{j+1}.
 \end{aligned}
 \label{appB:clifford-action}
\end{equation}
Composing these maps in the order fixed by
Eq.~\eqref{appB:Tsigma-circuit} gives $T_\sigma O_jT_\sigma^\dagger=O_{j+1}$
for the ordered list of $2L-1$ positive local terms of $-H^{\mathrm d}$.
After $2L-1$ iterations, the induced Pauli automorphism agrees with
conjugation by $e^{i\pi\Omega/4}$.  Their ratio is therefore a scalar, and a
single nonzero computational-basis matrix element fixes it to one.  Thus
\begin{equation}
 T_\sigma^{2L-1}=e^{i\pi\Omega/4}
 =\frac{1+i\Omega}{\sqrt2},
\end{equation}
as used in Eq.~(\ref{eq:Tsigma-power}).

\section{Physical ensemble, exact finite-size obstruction, and Kramers descent}
\label{app:physical-rdm}

Let $L\geq2$, let the global charge be
$\Omega=\prod_{j=1}^{L}X_j$, and
$P_\omega=(I+\omega\Omega)/2$ projects onto the sector
$\Omega=\omega$, where $\omega=\pm1$.  Let $\rho_{\omega,L}$ denote the
unique ground-state projector of $H_L^{\rm d}$ in that sector.  The
physical charge-unresolved preparation is the equal-weight mixture
\begin{equation}
 \rho_{{\rm can},L}=\frac12\bigl(\rho_{+,L}+\rho_{-,L}\bigr).
 \label{appC:physical-ensemble}
\end{equation}
It is positive, normalized, stationary, and ground-supported, with
$\operatorname{Tr}(P_\omega\rho_{{\rm can},L})=1/2$ and
$\operatorname{Tr}(\Omega\rho_{{\rm can},L})=0$.

For the complete proper prefix $A_m=\{1,\ldots,m\}$, $1\leq m<L$, the
ordinary spin algebra is the full matrix algebra
$B((\mathbb C^2)^{\otimes m})$.  It has trivial center and is generated by
the first $2m$ Jordan--Wigner Majoranas.  Hence the physical prefix RDM is
the ordinary partial trace,
\begin{equation}
 \rho_{{\rm can},A_m}^{(L)}
 =\frac12\left(\rho_{+,A_m}^{(L)}+\rho_{-,A_m}^{(L)}\right).
 \label{appC:prefix-mixture}
\end{equation}
Sector labels here name the preparation that was reduced, not a splitting of
the local algebra, and nothing identifies the two reductions in
Eq.~\eqref{appC:prefix-mixture} with orthogonal local parity blocks.  Their
overlap need not vanish.  In the stated spin convention the one-site
reductions at $L=2$ satisfy
\begin{equation}
 \Tr\!\left(\rho_{+,A_1}^{(2)}\rho_{-,A_1}^{(2)}\right)=\frac13.
 \label{appC:global-charge-overlap}
\end{equation}
By contrast, the local parity restrictions constructed below are supported on orthogonal subspaces, and their normalized entropies obey
Corollary~\ref{cor:parity-resolved-moments}.  Entropy accounting therefore
rests on the actual blocks of the reduced state, and no Shannon term is
attributed to the mixture in Eq.~\eqref{appC:prefix-mixture}.

\subsection{An exact finite-size obstruction}

For each charge sector, the sign-zero prescription defines the real
antisymmetric covariance
\begin{equation}
 \Gamma_{\omega,0}^{(L)}
 =-i\,\operatorname{sign}(iA_{\omega,L}),
 \qquad \operatorname{sign}(0)=0.
 \label{appC:sign-zero-covariance}
\end{equation}
This covariance is auxiliary, since $\operatorname{Pf}\Gamma_{\omega,0}^{(L)}=0$,
whereas a state supported in the sector $\Omega=\omega$ has charge
expectation $\omega$.  The physical sector ground state instead uses the
unique pure completion with
$\operatorname{Pf}\Gamma_{\omega,{\rm phys}}^{(L)}=\omega$.
Let $\rho_{G,\omega,0;A_m}^{(L)}$ denote the restriction to $A_m$ of the
Gaussian state associated with \eqref{appC:sign-zero-covariance} in the same
Jordan--Wigner CAR chart.

The first complete-prefix mismatch occurs at $L=3$ and
$A_2=\{1,2\}$; the one-site coincidence is an accidental
low-dimensional identity.  Direct diagonalization in the spin algebra gives
\begin{equation}
 \operatorname{spec}\rho_{{\rm can},A_2}^{(3)}
 =\left\{\frac14\left(1\pm
 \frac{\sqrt{15+4\sqrt5}}5\right)\right\},
 \label{appC:physical-spectrum}
\end{equation}
with each displayed eigenvalue occurring twice.  For either
$\omega=\pm1$, the auxiliary Gaussian proxy has instead
\begin{equation}
 \operatorname{spec}\rho_{G,\omega,0;A_2}^{(3)}
 =\left\{\frac14\left(1\pm
 \frac{\sqrt{15+2\sqrt5}}5\right)\right\},
 \label{appC:proxy-spectrum}
\end{equation}
again with multiplicity two.  The respective purities are
$(10+\sqrt5)/25$ and $(20+\sqrt5)/50$, and therefore
\begin{equation}
 \left|\operatorname{Tr}\!\left[
   \bigl(\rho_{{\rm can},A_2}^{(3)}\bigr)^2\right]
 -\operatorname{Tr}\!\left[
   \bigl(\rho_{G,\omega,0;A_2}^{(3)}\bigr)^2\right]\right|
 =\frac{\sqrt5}{50}\neq0.
 \label{appC:purity-mismatch}
\end{equation}
Thus the physical mixture cannot universally collapse at finite $L$ to
either sector-derived sign-zero proxy.  In particular, no physical finite-size
single-particle entanglement level $\xi=0$ follows from the sign-zero
covariance.  This obstruction concerns finite size only; it does not affect
the separately proved common limit obtained by sending $L\to\infty$ at each
fixed prefix.

\subsection{Global sector exchange}

The finite-size obstruction coexists with an exact antiunitary symmetry.
Let $K_L$ denote complex conjugation in the computational basis and set
$\mathcal C_L=Z_1K_L$.  Complex conjugation reverses the sign of the sole
$Y_1$ in $H_L^{\rm d}$, while conjugation by $Z_1$ reverses it once more and
leaves the remaining terms invariant.  Consequently,
\begin{equation}
 \mathcal C_L H_L^{\rm d}\mathcal C_L^{-1}=H_L^{\rm d},
 \qquad
 \mathcal C_L\Omega\mathcal C_L^{-1}=-\Omega,
 \qquad \mathcal C_L^2=I.
 \label{appC:global-sector-exchange-symmetry}
\end{equation}
The uniqueness of the ground state in each charge sector then
implies
\begin{equation}
 \rho_{-,L}=Z_1\rho_{+,L}^{*}Z_1.
 \label{appC:sector-exchange}
\end{equation}
Sector support also gives
$[\rho_{{\rm can},L},\Omega]=0$.  Defining
$\Theta_L:=\Omega\mathcal C_L$, and choosing an immaterial overall phase,
we may write
\begin{equation}
 \Theta_L=\left(Y_1\prod_{j=2}^{L}X_j\right)K_L,
 \qquad
 \Theta_L\rho_{{\rm can},L}\Theta_L^{-1}
 =\rho_{{\rm can},L}.
 \label{appC:global-kramers-symmetry}
\end{equation}

\subsection{Partial-trace descent and parity spectral doubling}

To transfer this global symmetry to the entanglement spectrum, we now trace
its action to the prefix RDM.  Write $\Theta_L=U_LK_L$.  In the computational
product basis, both factors split across every proper prefix as
$U_L=U_{A_m}\otimes U_{A_m^c}$ and
$K_L=K_{A_m}\otimes K_{A_m^c}$, with
\begin{equation}
 U_{A_m}=Y_1\prod_{j=2}^{m}X_j.
 \label{appC:prefix-unitary}
\end{equation}
For every operator $M$, invariance of the partial trace under unitary
conjugation on the traced factor, together with its commutation with
entrywise conjugation in this basis, gives
\begin{align}
 &\operatorname{Tr}_{A_m^c}\!\left[
 (U_{A_m}\otimes U_{A_m^c})M^*
 (U_{A_m}\otimes U_{A_m^c})^\dagger\right]
 \nonumber\\
 &\hspace{30mm}=U_{A_m}\left(\operatorname{Tr}_{A_m^c}M\right)^*
 U_{A_m}^\dagger.
 \label{appC:partial-trace-antiunitary}
\end{align}
Applying this identity to \eqref{appC:global-kramers-symmetry} yields the
reduced antiunitary
\begin{equation}
 \Theta_{A_m}:=U_{A_m}K_{A_m},
 \qquad
 \Theta_{A_m}\rho_{{\rm can},A_m}^{(L)}\Theta_{A_m}^{-1}
 =\rho_{{\rm can},A_m}^{(L)},
 \qquad
 \Theta_{A_m}^{2}=U_{A_m}U_{A_m}^{*}=-I.
 \label{eq:prefix-kramers-symmetry}
\end{equation}
The last equality follows because $Y_1^*=-Y_1$, whereas every $X_j$ is
real.  If $v$ is an eigenvector of the Hermitian RDM, then
$\Theta_{A_m}v$ belongs to the same eigenspace and
$\langle v,\Theta_{A_m}v\rangle=0$.  Every eigenvalue, including zero,
therefore has even multiplicity.

Let $\Omega_{A_m}=\prod_{j=1}^{m}X_j$ be the local parity.  Since
$\Omega=\Omega_{A_m}\otimes\Omega_{A_m^c}$ and
$[\rho_{{\rm can},L},\Omega]=0$, unitary invariance of the partial trace
implies
\begin{equation}
 [\rho_{{\rm can},A_m}^{(L)},\Omega_{A_m}]=0,
 \qquad
 \Theta_{A_m}\Omega_{A_m}\Theta_{A_m}^{-1}=-\Omega_{A_m}.
 \label{appC:local-parity-relations}
\end{equation}
Thus $\Theta_{A_m}$ exchanges the two local-parity blocks of the RDM.  With
$P_{A_m,\pm}=(I\pm\Omega_{A_m})/2$, let
$R_+$ be the restriction of
$P_{A_m,+}\rho_{{\rm can},A_m}^{(L)}P_{A_m,+}$ to
$P_{A_m,+}\mathcal H_{A_m}$.  The two restrictions are antiunitarily
equivalent, so for every $L\geq2$ and $1\leq m<L$,
\begin{equation}
 \det\!\left(tI-\rho_{{\rm can},A_m}^{(L)}\right)
 =\det(tI-R_+)^2.
 \label{appC:parity-square}
\end{equation}
Choosing paired orthonormal bases of the two blocks gives the abstract
unitary equivalence
\begin{equation}
 \rho_{{\rm can},A_m}^{(L)}
 \cong\frac{I_2}{2}\otimes\sigma_{L,m},
 \qquad \sigma_{L,m}=2R_+,
 \qquad \operatorname{Tr}\sigma_{L,m}=1.
 \label{eq:abstract-qubit-factor}
\end{equation}
Consequently,
$S(\rho_{{\rm can},A_m}^{(L)})=\log2+S(\sigma_{L,m})$.
This is an exact many-body Kramers doubling, not a Gaussian factorization.
It does not define a canonical local qubit factorization or extraction
circuit.
Moreover, the abstract factor does not determine the ordered entropy
increment, because the compensating contribution from $\sigma_{L,m}$ is fixed only
by the separate Toeplitz/Fisher--Hartwig comparison.

\section{Full-RDM bridge and Toeplitz endpoint}
\label{app:entropy-bridge}

The fixed-prefix covariance estimates determine the complete physical reduced
density matrices and their entropy in the strict order first $L\to\infty$
at fixed prefix length $m$, and only afterwards $m\to\infty$.  Throughout,
$A_m=\{1,\ldots,m\}$ is a complete proper prefix, $1\leq m<L$, with its
ordinary full spin matrix algebra.  On the $2m$ prefix Majoranas,
$\rho_G(\Gamma)$ denotes the Gaussian density matrix whose moments are fixed
by an admissible covariance $\Gamma$.

\subsection{Sectorwise correlators and complete-RDM convergence}

Let $N=2L-1$.  In the literal prefix order used in the main text, the common
limiting covariance of the two physical charge sectors is
\begin{equation}
 \Gamma_{\infty,m}=0\oplus T_{2m-1},
 \qquad
 (T_n)_{jk}=
 \begin{cases}
  \dfrac{2}{\pi(k-j)},&k-j\ \text{odd},\\[4pt]
  0,&k-j\ \text{even},
 \end{cases}
 \quad 0\leq j,k<n.
 \label{appD:limiting-covariances}
\end{equation}
Put $H_n=iT_n$.  This Hermitian matrix is the finite compression
$H_n=P_nM_hP_n$ of the $\{\pm1\}$-valued sign symbol $h$.  Consequently,
\begin{equation}
 H_n=H_n^*,\qquad \|H_n\|\leq1,
 \qquad \sigma(H_n)\subset[-1,1].
 \label{appD:toeplitz-support}
\end{equation}
Thus $T_n$ and $0\oplus T_n$ are admissible covariance matrices.  Moreover,
$J_nH_nJ_n=-H_n$ for $J_n=\operatorname{diag}(1,-1,\ldots)$, so the spectrum
of $H_n$ is symmetric about zero.

For either charge $\omega=\pm1$, the finite-cycle quadrature error and the
restriction of the physical zero-mode completion give the entrywise estimate
\begin{equation}
 \max_{1\leq a,b\leq 2m}
 \left|\bigl(\Gamma_{\omega,A_m}^{(L)}
       -\Gamma_{\infty,m}\bigr)_{ab}\right|
 \leq f_m(L),
 \qquad
 f_m(L)=\sqrt{\frac{2m-1}{2L-1}}
       +\frac{\pi(2m-1)}{2L-1}.
 \label{appD:entrywise-rate}
\end{equation}
The first term is the delocalized completion restricted to the prefix, while
the second is the Fourier-grid error.  Both terms vanish as $L\to\infty$ for
every fixed $m$; no uniformity in growing $m$ is asserted.

Fix an increasing even Majorana set $I\subset\{1,\ldots,2m\}$ with
$|I|=2r$.  Each physical sector ground state is the unique quasifree state
selected by its pure Pfaffian completion and is also a parity eigenstate.
Wick's rule therefore gives
\begin{equation}
 \langle\gamma_I\rangle_{\omega,L}
 =i^r\operatorname{Pf}\!\left[
   (\Gamma_{\omega,A_m}^{(L)})_I\right],
 \qquad
 \langle\gamma_I\rangle_{G}
 =i^r\operatorname{Pf}\!\left[(\Gamma_{\infty,m})_I\right].
 \label{appD:wick-pfaffian}
\end{equation}
All covariance entries have modulus at most one.  Expanding each Pfaffian over
perfect matchings and replacing one factor at a time therefore yields
\begin{equation}
 \left|\langle\gamma_I\rangle_{\omega,L}
       -\langle\gamma_I\rangle_G\right|
 \leq r(2r-1)!!\,f_m(L).
 \label{appD:pfaffian-telescope}
\end{equation}
Odd moments vanish in both states.  Since the Majorana monomials form an
orthogonal basis of the full prefix operator algebra, every prefix density
matrix has the exact expansion
\begin{equation}
 \rho=2^{-m}\sum_{I\subset\{1,\ldots,2m\}}
       \langle\gamma_I\rangle^*\gamma_I.
 \label{appD:density-expansion}
\end{equation}
Using $\|\gamma_I\|_1=2^m$ and summing
Eq.~\eqref{appD:pfaffian-telescope} over the even nonempty sets gives, for
each $\omega$,
\begin{equation}
 \left\|\rho_{\omega,A_m}^{(L)}
       -\rho_G(0\oplus T_{2m-1})\right\|_1
 \leq C_m f_m(L),
 \qquad
 C_m=\sum_{r=1}^{m}\binom{2m}{2r}r(2r-1)!!.
 \label{appD:trace-bound}
\end{equation}
This is an all-correlator reconstruction of the complete density matrix, not
an averaging of sector covariances.

The physical charge-unresolved prefix state is the ordinary convex sum
\begin{equation}
 \rho_{{\rm can},A_m}^{(L)}
 =\frac12\left(\rho_{+,A_m}^{(L)}+\rho_{-,A_m}^{(L)}\right).
 \label{appD:physical-prefix-mixture}
\end{equation}
Because Eq.~\eqref{appD:trace-bound} has the same comparator and bound in both
sectors, the triangle inequality gives
\begin{equation}
 \left\|\rho_{{\rm can},A_m}^{(L)}
       -\rho_G(0\oplus T_{2m-1})\right\|_1
 \leq C_m f_m(L)
 \xrightarrow[\substack{L\to\infty\\L>m}]{}0.
 \label{eq:physical-full-rdm-limit}
\end{equation}
The physical homogeneous and invertible-$\eta$ ground states are both
parity-even quasifree states.  If
$\Gamma_{{\rm hom},A_m}^{(L)}$ and
$\Gamma_{\eta,A_m}^{(L)}$ denote their prefix covariances, the antiperiodic and
periodic Fourier quadratures have the same $T_{2m}$ limit.  The physical
$\eta$ completion also contains a rank-two periodic zero-mode term whose
entries on a fixed prefix are $O(L^{-1})$.  Hence
\begin{equation}
 \max\!\left\{
 \max_{1\leq a,b\leq2m}
 \left|(\Gamma_{{\rm hom},A_m}^{(L)}-T_{2m})_{ab}\right|,
 \max_{1\leq a,b\leq2m}
 \left|(\Gamma_{\eta,A_m}^{(L)}-T_{2m})_{ab}\right|
 \right\}
 \xrightarrow[\substack{L\to\infty\\L>m}]{}0
 \qquad (m\ \text{fixed}).
 \label{appD:even-controls-entrywise-limit}
\end{equation}  Their odd
moments vanish, while the same Pfaffian telescoping controls every even moment.
Reconstruction in the complete monomial basis therefore gives
\begin{align}
 \left\|\rho_{{\rm hom},A_m}^{(L)}-\rho_G(T_{2m})\right\|_1
 &\xrightarrow[\substack{L\to\infty\\L>m}]{}0,
 \label{eq:hom-full-rdm-limit}\\
 \left\|\rho_{\eta,A_m}^{(L)}-\rho_G(T_{2m})\right\|_1
 &\xrightarrow[\substack{L\to\infty\\L>m}]{}0.
 \label{eq:eta-full-rdm-limit}
\end{align}
Equations~\eqref{eq:physical-full-rdm-limit},
\eqref{eq:hom-full-rdm-limit}, and \eqref{eq:eta-full-rdm-limit} concern
complete prefixes, rather than wrapping
or otherwise disconnected regions.  Their constants are not uniform in the
prefix size, so these limits cannot be promoted to growing or fixed-ratio
prefixes, nor can the two limits below be exchanged.  In particular, neither
statement identifies the finite-size physical mixture with a Gaussian state.

\subsection{Finite-cycle kernel and the fixed-prefix rate}

For a cycle of length $N$, choose the periodic or antiperiodic grid
$q_s=(2\pi s+\vartheta)/N$, with $\vartheta=0$ or $\pi$, and set the spectral
sign to zero at its jumps.  The real covariance kernel can be written
\begin{equation}
 C_{N,\vartheta}(d)=\operatorname{Im}\!\left[
 \frac1N\sum_{s=0}^{N-1}
 \bigl(-\operatorname{sgn}(\sin q_s)\bigr)e^{-iq_sd}\right].
 \label{appD:cycle-kernel}
\end{equation}
After centering the grid on $(-\pi,\pi]$, a cell-by-cell comparison with the
integral gives, for every fixed integer $d$,
\begin{equation}
 \left|C_{N,\vartheta}(d)-C_\infty(d)\right|
 \leq\frac{\pi(|d|+1)}{N},
 \qquad
 C_\infty(d)=
 \begin{cases}
  2/(\pi d),&d\ \text{odd},\\
  0,&d\ \text{even}.
 \end{cases}
 \label{appD:kernel-limit}
\end{equation}
At a jump, the displayed imaginary kernel is insensitive to the assigned
real sign value.  The largest separation inside $A_m$ is at most $2m-2$.
Combining Eq.~\eqref{appD:kernel-limit} with the restricted completion mass
and the all-correlator estimate gives the explicit fixed-prefix rate
\begin{equation}
 \left\|\rho_{\omega,A_m}^{(L)}
       -\rho_G(0\oplus T_{2m-1})\right\|_1
 \leq C_m\!\left[
 \sqrt{\frac{2m-1}{2L-1}}
 +\frac{\pi(2m-1)}{2L-1}\right].
 \label{eq:full-rdm-rate}
\end{equation}
The prefactor $C_m$ grows with $m$.  Thus
Eq.~\eqref{eq:full-rdm-rate} is deliberately a fixed-dimensional estimate and
supplies no simultaneous or fixed-ratio limit.

\subsection{Cross-sector reduction and general ground-state preparations}
\label{appD:coherent-preparations}

Let $|\psi_{\pm,L}\rangle$ denote the two physical charge-sector ground
vectors.  The diagonal estimates do not control a coherent superposition,
whose reduction also contains
\begin{equation}
 X_{m,L}=\operatorname{Tr}_{A_m^c}
 |\psi_{+,L}\rangle\langle\psi_{-,L}|.
 \label{appD:cross-rdm-definition}
\end{equation}
We control this operator on the same complete fixed prefix, using a nonzero
Gaussian overlap to reconstruct every cross-sector moment.  No restriction
to even local observables is imposed.

\subsubsection{Physical covariances and a nonzero overlap}

In the full order $(a_1,q_1,\ldots,q_N)$, with
$q=(b_1,a_2,b_2,\ldots,a_L,b_L)$ and $N=2L-1$, write
\begin{equation}
 G_\omega=\begin{pmatrix}0&z_\omega^{\mathsf T}\\
                    -z_\omega&S_\omega\end{pmatrix},
 \qquad G_\omega^2=-I.
 \label{appD:cross-pure-covariance}
\end{equation}
The Fourier orientation $\det B_\omega=\omega$ and the physical
zero-occupation selector in Eq.~\eqref{appE:zero-selector} fix
\begin{equation}
 (z_+)_r=N^{-1/2},\qquad
 (z_-)_r=(-1)^rN^{-1/2},\qquad 1\leq r\leq N.
 \label{appD:cross-zero-signs}
\end{equation}
Pairing opposite Fourier points in Eq.~\eqref{appD:cycle-kernel} with
\begin{equation*}
 \sum_{s=0}^{M-1}\sin((2s+1)x)=\frac{\sin^2(Mx)}{\sin x}
\end{equation*}
gives, for $1\leq d\leq N-1$,
\begin{equation}
 C_{N,\vartheta}(d)=
 \begin{cases}
 N^{-1}\cot\dfrac{\pi d}{2N},&d\text{ odd},\\[3pt]
 -\cos\vartheta\,N^{-1}\tan\dfrac{\pi d}{2N},&d\text{ even},
 \end{cases}
 \qquad \vartheta\in\{0,\pi\}.
 \label{appD:cross-exact-kernels}
\end{equation}
Thus $(S_+)_{rs}=C_{N,0}(s-r)$ and
$(S_-)_{rs}=C_{N,\pi}(s-r)$, with antisymmetric extension.

The state obtained by applying the spectator Majorana has the same parity
as $|\psi_{-,L}\rangle$.  Write
\begin{equation}
 |\phi_L\rangle=a_1|\psi_{+,L}\rangle,\qquad
 R=\operatorname{diag}(1,-I_N),\qquad G_\phi=RG_+R^{\mathsf T}.
 \label{appD:cross-phi}
\end{equation}
Choose phases so that $F_L=\langle\psi_{-,L}|\phi_L\rangle$ is positive;
nonvanishing follows below.  The pure-Gaussian overlap formula, reduced to
its bipartite Pfaffian blocks, gives
\begin{equation}
 F_L^2=2^{-L}\sqrt{\det(I-G_-G_\phi)}
      =|\det C_L|=:D_L,
 \label{appD:cross-overlap-determinant}
\end{equation}
where the periodic/antiperiodic Fourier overlap is
\begin{equation}
 (C_L)_{jk}=\frac1N\sum_{r=1}^{N}e^{i(p_k-q_j)r},\qquad
 p_k=-\frac{2\pi k}{N},\quad q_j=-\frac{(2j+1)\pi}{N},
 \quad 0\leq j,k\leq L-1.
 \label{appD:cross-fourier-overlap}
\end{equation}
Since $e^{i(p_k-q_j)N}=-1$, the geometric sum is a diagonally scaled Cauchy
matrix.  Its trigonometric determinant is
\begin{equation}
 D_L=N^{-L}
 \frac{\displaystyle\prod_{d=1}^{L-1}
       \sin\!\left(\frac{\pi d}{N}\right)^{N+1-2d}}
      {\displaystyle\prod_{k=0}^{L-1}
       \sin\!\left(\frac{\pi(k+1/2)}{N}\right)^{N-2k}}.
 \label{appD:cross-sine-product}
\end{equation}
The numerator exponents count the two Fourier Vandermonde products; the
denominator counts odd half-grid separations.  All factors are positive.

To extract the decay, separate the linear and smooth sine factors by writing
\begin{equation}
 \sin(\pi x)=\pi x\,s(x),\qquad
 g(x)=\log s(x),\qquad a(x)=(1-2x)g(x).
 \label{appD:cross-sine-split}
\end{equation}
The linear part, expressed through the Barnes $G$-function, is
\begin{equation}
 D_{\log}(L)=
 \frac{G(L+1)^2G(3/2)^2\Gamma(L+1/2)}
      {\Gamma(1/2)G(L+3/2)^2}
 =\Theta(L^{-1/4}).
 \label{appD:cross-barnes}
\end{equation}
The remaining logarithm is a difference of endpoint and midpoint
quadratures for $a,g$ on $[0,1/2]$.  Both functions are smooth, with
$a(0)=a(1/2)=0$; shifted Euler--Maclaurin therefore yields
\begin{equation}
 \log D_{\rm sm}(L)=c_{\rm sm}+O(L^{-1}),
 \qquad D_L=D_{\log}(L)D_{\rm sm}(L).
 \label{appD:cross-smooth-factor}
\end{equation}
There is no further logarithmic power.  Consequently,
\begin{equation}
 D_L=\kappa L^{-1/4}(1+o(1)),\qquad
 F_L=\sqrt\kappa\,L^{-1/8}(1+o(1)),\qquad \kappa>0.
 \label{appD:cross-overlap-rate}
\end{equation}
This decay follows from the lattice product without a CFT scaling
assumption.

\subsubsection{Transition contractions on a fixed prefix}

For the ordered full Majoranas $\gamma$, define the antisymmetric
transition contraction by
\begin{equation}
 \mathcal K_{ij}=-i\,
 \frac{\langle\psi_-|\gamma_i\gamma_j|\phi\rangle}{F_L}
 \quad(i\ne j),\qquad \mathcal K_{ii}=0.
 \label{appD:cross-transition-definition}
\end{equation}
Gaussian bra and ket annihilation relations give
$\mathcal K(G_-+G_\phi)=-2I-iG_-+iG_\phi$, hence
\begin{equation}
 \mathcal K=(-2I-iG_-+iG_\phi)(G_-+G_\phi)^{-1}.
 \label{appD:cross-transition-kernel}
\end{equation}
The inverse exists because
\begin{equation}
 I-G_-G_\phi=-G_-(G_-+G_\phi)
 \label{appD:cross-invertibility}
\end{equation}
and the determinant in Eq.~\eqref{appD:cross-overlap-determinant} is
nonzero.  Matrix order matters; the formula reduces to $\mathcal K=G$ when
both covariances equal $G$.  Differentiating the Gaussian transition
characteristic function gives the ordered Wick rule
\begin{equation}
 \frac{\langle\psi_-|\gamma_{j_1}\cdots\gamma_{j_{2r}}|\phi\rangle}{F_L}
 =i^r\operatorname{Pf}(\mathcal K_J),\qquad j_1<\cdots<j_{2r}.
 \label{appD:cross-transition-wick}
\end{equation}

Write $N=2n+1$.  Grouping $(q_1,q_3,\ldots,q_N)$ before
$(a_1,q_2,q_4,\ldots,q_{N-1})$ gives
\begin{equation}
 \Sigma:=G_-+G_\phi=
 \begin{pmatrix}0&D\\-D^{\mathsf T}&0\end{pmatrix},
 \qquad D\in\mathbb R^{(n+1)\times(n+1)},
 \label{appD:cross-bipartite}
\end{equation}
with
\begin{equation}
 D_{a0}=\frac2{\sqrt N},\qquad
 D_{a,b+1}=-\frac2N\cot\frac{\pi(a-b-1/2)}N,
 \quad 0\leq a\leq n,\quad 0\leq b<n.
 \label{appD:cross-D}
\end{equation}
Set
\begin{equation}
 x_a=e^{2\pi ia/N},\quad y_b=e^{2\pi i(b+1/2)/N},\qquad
 \cot\frac{\pi(a-b-1/2)}N=i\frac{x_a+y_b}{x_a-y_b}.
 \label{appD:cross-cauchy-nodes}
\end{equation}
Subtracting multiples of the constant column from the cotangent columns
reduces $D$, up to diagonal factors, to a Cauchy matrix with the additional
node $y_*=0$.  Define
\begin{equation}
 R_0=1,\qquad
 R_r=\prod_{k=1}^{r}
 \frac{\sin(\pi k/N)}{\sin(\pi(k-1/2)/N)}.
 \label{appD:cross-R-product}
\end{equation}
The Cauchy inverse formula then yields
\begin{align}
 |(D^{-1})_{0a}|&=
 \frac{\sqrt N}{2R_aR_{n-a}}
 \cos\frac{\pi(n-2a)}{2N},\label{appD:cross-inverse-zero}\\
 |(D^{-1})_{b+1,a}|&=
 \frac{N/4}{R_aR_{n-a}}
 \frac{4\sin\!\frac{\pi(b+1)}N\sin\!\frac{\pi(n-b)}N}
 {2\sin\!\frac{\pi|b+1/2-a|}N\,R_{b+1}R_{n-b}}.
 \label{appD:cross-inverse-active}
\end{align}
Every denominator is nonzero.  Uniformly for $0\leq r\leq n$,
\begin{equation}
 c\sqrt{r+1}\leq R_r\leq C\sqrt{r+1},\qquad
 \prod_{k=1}^{r}\frac{k}{k-1/2}
 =\frac{\Gamma(r+1)\Gamma(1/2)}{\Gamma(r+1/2)}
 \asymp\sqrt{r+1}.
 \label{appD:cross-R-bound}
\end{equation}
The linear product gives the displayed Gamma ratio; the smooth-factor
quotient has a uniformly bounded logarithm, a half-step sum of the smooth
function $g$.

Using these bounds and $\sin(\pi r/N)\asymp r/N$ for $0\leq r\leq n$ in
the inverse formulas gives
\begin{equation}
 \|\Sigma^{-1}\|_{\max}=O(1),\qquad
 \sum_j|\Sigma^{-1}_{ij}|+\sum_j|\Sigma^{-1}_{ji}|
 =O_m(\sqrt N)
 \quad(i\text{ in the fixed prefix}).
 \label{appD:cross-inverse-bounds}
\end{equation}
For a fixed endpoint index, split the inverse sum into
$|b-a|\leq(a+1)/2$ and its complement.  The same bounds control the
transposed sum; the spectator row follows from
Eq.~\eqref{appD:cross-inverse-zero}.

Put $\Delta=G_--G_\phi$.  The active-row difference contains only
even-separation kernel terms apart from its spectator entry, and
\begin{equation}
 \sum_s|\Delta_{rs}|=O(\log N),\qquad
 \frac2N\tan\frac{\pi d}{2N}
 =\frac2N\cot\frac{\pi(N-d)}{2N}
 \leq\frac4{\pi(N-d)}\quad(d\text{ even}).
 \label{appD:cross-delta-bound}
\end{equation}
The spectator row has entries from $z_-+z_+$, each either zero or
$2/\sqrt N$.  Combining its selected-column sum with
Eq.~\eqref{appD:cross-inverse-bounds}, and the active-row bound with the
maximum inverse entry, in $\mathcal K=(-2I-i\Delta)\Sigma^{-1}$ proves
\begin{equation}
 |\mathcal K_{ij}|\leq K_m(1+\log N)
 \qquad(i,j\text{ in the fixed prefix}).
 \label{appD:cross-prefix-contractions}
\end{equation}

\subsubsection{Complete reconstruction and preparation uniformity}

Global parity gives
\begin{equation}
 \Omega\gamma_I=(-1)^{|I|}\gamma_I\Omega,\qquad
 \langle\psi_-|\gamma_I|\psi_+\rangle=0
 \quad(|I|\text{ even}).
 \label{appD:cross-parity-selection}
\end{equation}
For odd $|I|$, use $|\psi_+\rangle=a_1|\phi\rangle$ and cancel a repeated
$a_1$ if present.  Reordering leaves an even word of length at most $2m$
inside the prefix.  The transition Pfaffian therefore obeys
\begin{equation}
 |\langle\psi_-|\gamma_I|\psi_+\rangle|
 \leq F_L(2m-1)!![K_m(1+\log N)]^m,
 \label{appD:cross-moment-bound}
\end{equation}
after enlarging $K_m$ if necessary.  Orthogonality of all Majorana
monomials, not Hermiticity of $X_{m,L}$, supplies the reconstruction and
trace bound in the form
\begin{equation}
 X_{m,L}=2^{-m}\sum_I\operatorname{Tr}(\gamma_I^\dagger X_{m,L})\gamma_I,
 \qquad
 \|X_{m,L}\|_1\leq
 \sum_{|I|\text{ odd}}|\langle\psi_-|\gamma_I|\psi_+\rangle|.
 \label{appD:cross-reconstruction}
\end{equation}
Here $\|\gamma_I\|_1=2^m$.  There are finitely many terms at fixed $m$.
Together with $\|a_1\|_\infty=1$ and
$|\operatorname{Tr}(a_1X_{m,L})|=F_L$, this proves
\begin{equation}
 \Theta(N^{-1/8})=F_L\leq\|X_{m,L}\|_1
 \leq B_mN^{-1/8}(1+\log N)^m\longrightarrow0.
 \label{appD:cross-trace-bound}
\end{equation}
The logarithm is an upper-bound allowance, not a derived physical
logarithmic correction.  No sharper $O_m(N^{-1/8})$ upper rate is needed.

For a ground-space density matrix with diagonal weight $p_L$ and
coherence $c_L$, positivity gives $2|c_L|\leq1$.  Linearity of reduction
and the two diagonal bounds in Eq.~\eqref{appD:trace-bound} imply
\begin{align}
 \|\rho_{A_m}^{(L)}(p_L,c_L)-\rho_G(0\oplus T_{2m-1})\|_1
 &\leq p_L\|\rho_{+,A_m}^{(L)}-\rho_G\|_1
 +(1-p_L)\|\rho_{-,A_m}^{(L)}-\rho_G\|_1\notag\\
 &\quad+2|c_L|\|X_{m,L}\|_1\notag\\
 &\leq C_m f_m(L)+B_mN^{-1/8}(1+\log N)^m,
 \label{appD:cross-uniform-ground-space}
\end{align}
where $\rho_G$ in the first line denotes that same comparator.  The bound
is uniform over all ground-space density matrices, including $L$-dependent
amplitudes and phases.  Entropy continuity at fixed dimension $2^m$,
followed by the unchanged Toeplitz endpoint below, consequently gives
\begin{equation}
 \lim_{m\to\infty}\left\{
 \lim_{\substack{L\to\infty\\L>m}}
 \left[S(\rho_{A_m}^{(L)}(p_L,c_L))
       -S(\rho_{{\rm hom},A_m}^{(L)})\right]\right\}
 =\frac12\log2.
 \label{appD:cross-ordered-entropy}
\end{equation}
The constants depend on $m$, so the argument retains the stated order of
limits.  It establishes a common local limit, not identical finite-size
RDMs or an extension of the equal-mixture Kramers symmetry to coherent
preparations.
\subsection{Entropy continuity and the ordered physical theorem}

Let $d_m=2^m$ and
$\delta(\rho,\sigma)=\tfrac12\|\rho-\sigma\|_1$.  For fixed $m$, the
Audenaert bound
\begin{equation}
 |S(\rho)-S(\sigma)|
 \leq \delta\log(d_m-1)+h_2(\delta),
 \qquad 0\leq\delta\leq1-d_m^{-1},
 \label{appD:audenaert}
\end{equation}
applies separately to the physical KW, homogeneous, and $\eta$ fixed-prefix
trace-norm limits above.  In each case, the corresponding trace distance tends
to zero at fixed $m$ and therefore satisfies the displayed domain condition for
all sufficiently large $L$.  The bound then yields the three associated
entropy limits.

Define
\begin{equation}
 \eta(x)=-\frac{1+x}{2}\log\frac{1+x}{2}
         -\frac{1-x}{2}\log\frac{1-x}{2},
 \qquad \eta(0)=\log2.
 \label{appD:entropy-function}
\end{equation}
Let $S_n$ be the sum of $\eta(\nu)$ over the positive eigenvalues of $H_n$,
and let $S_{2m-1}^{\circ}$ use only the positive nonzero members of the odd
spectrum.  The exact Gaussian--Toeplitz dictionary is
\begin{equation}
 S\!\left[\rho_G(0\oplus T_{2m-1})\right]
 =\log2+S_{2m-1}^{\circ},
 \qquad
 S\!\left[\rho_G(T_{2m})\right]=S_{2m},
 \label{appD:entropy-dictionary}
\end{equation}
where $S_{2m-1}^{\circ}$ omits the unique zero eigenvalue of the odd active
block.  Therefore, at each fixed $m$,
\begin{equation}
 \lim_{\substack{L\to\infty\\L>m}}
 \left[S(\rho_{{\rm can},A_m}^{(L)})
      -S(\rho_{{\rm hom},A_m}^{(L)})\right]
 =\log2+S_{2m-1}^{\circ}-S_{2m}.
 \label{appD:fixed-prefix-entropy-limit}
\end{equation}
The endpoint derivation below establishes
$S_{2m}-S_{2m-1}^{\circ}\to\tfrac12\log2$ and hence proves the physical
statement
\begin{equation}
 \lim_{m\to\infty}\left\{
 \lim_{\substack{L\to\infty\\L>m}}
 \left[S(\rho_{{\rm can},A_m}^{(L)})
      -S(\rho_{{\rm hom},A_m}^{(L)})\right]\right\}
 =\frac12\log2.
 \label{appD:ordered-physical-entropy}
\end{equation}
The order displayed in Eq.~\eqref{appD:ordered-physical-entropy} is essential
because the available estimate does not control a prefix that grows with $L$.
Furthermore, the exact many-body Kramers factor from
Appendix~\ref{app:physical-rdm} gives
$S(\rho_{{\rm can},A_m}^{(L)})=\log2+S(\sigma_{L,m})$ at finite size, but this
factor alone does not determine the ordered constant.  The active component
supplies a compensating $-\tfrac12\log2$ through the separate odd--even
Toeplitz comparison.

\subsection{Independent endpoint derivation and exact nullity}

We derive the consecutive Toeplitz input directly from the Fisher--Hartwig
determinant ratio.  Define
\begin{equation}
 D_N(z)=\det(zI-H_N),\qquad
 R_N(z)=\frac{D_N(z)}{D_{N-1}(z)},\qquad
 G(z)=\sqrt{z^2-1}\sim z
 \label{appD:determinant-ratio-definitions}
\end{equation}
on $\Omega_c=\mathbb C\setminus[-1,1]$.  The fixed-parameter
Fisher--Hartwig formula for the two inverse jumps gives
\begin{equation}
 R_N(z)\longrightarrow G(z)
 \qquad (z\in\Omega_c\ \text{fixed}).
 \label{appD:ratio-pointwise}
\end{equation}
To upgrade this pointwise statement, write the principal extension as
\begin{equation}
 H_N=\begin{pmatrix}H_{N-1}&u_N\\u_N^*&0\end{pmatrix},
 \qquad
 R_N(z)=z-u_N^*(zI-H_{N-1})^{-1}u_N.
 \label{appD:schur-ratio}
\end{equation}
Equation~\eqref{appD:toeplitz-support} implies $\|u_N\|\leq1$.  The Hermitian
resolvent estimate then makes $\{R_N\}$ locally bounded on $\Omega_c$, so
Vitali--Porter upgrades Eq.~\eqref{appD:ratio-pointwise} to local uniform
convergence.  Both $R_N$ and $G$ are zero-free off the cut.  Cauchy's estimates
on a slightly larger compact set, together with the resulting uniform lower
bound, give
\begin{equation}
 \frac{R_N'(z)}{R_N(z)}\longrightarrow\frac{G'(z)}{G(z)}
 =\frac{z}{z^2-1}
 \quad\text{locally uniformly on }\Omega_c.
 \label{appD:log-derivative-ratio}
\end{equation}

Let
\begin{equation}
 F_N(x)=N_{H_N}(x)-N_{H_{N-1}}(x),
 \label{appD:counting-shift}
\end{equation}
where $N_H(x)$ counts eigenvalues not exceeding $x$.  Since $H_{N-1}$ is a
principal compression of $H_N$, Cauchy interlacing gives
$F_N(x)\in\{0,1\}$.  For every $\varphi\in W^{1,1}(-1,1)$, Stieltjes
integration by parts yields
\begin{equation}
 \operatorname{Tr}\varphi(H_N)-\operatorname{Tr}\varphi(H_{N-1})
 =\varphi(1)-\int_{-1}^{1}\varphi'(x)F_N(x)\,dx,
 \label{appD:spectral-shift}
\end{equation}
and hence the uniform bound
$|\varphi(1)|+\|\varphi'\|_1$.  The Laurent coefficients of
Eq.~\eqref{appD:log-derivative-ratio} at infinity give, for every polynomial
$p$,
\begin{equation}
 \operatorname{Tr}p(H_N)-\operatorname{Tr}p(H_{N-1})
 \longrightarrow\frac{p(-1)+p(1)}2.
 \label{appD:polynomial-endpoint}
\end{equation}
Polynomials are dense in the norm
$\|\varphi\|_*=|\varphi(1)|+\|\varphi'\|_1$.  Indeed, approximate $\varphi'$ in
$L^1$ by a polynomial $q$ and integrate $q$ from the endpoint $1$.  Thus
Eq.~\eqref{appD:polynomial-endpoint} extends to every
$\varphi\in W^{1,1}(-1,1)$.  The entropy function obeys
\begin{equation}
 \eta'(x)=\frac12\log\frac{1-x}{1+x}\in L^1(-1,1),
 \qquad \eta(\pm1)=0,
 \label{appD:entropy-w11}
\end{equation}
so
$\operatorname{Tr}\eta(H_N)-\operatorname{Tr}\eta(H_{N-1})\to0$.

It remains to account for the zero eigenvalue of the odd block.  After the
even and odd indices are grouped, $H_{2n}$ has off-diagonal block equal to the
following Cauchy matrix up to invertible diagonal phase and sign factors,
\begin{equation}
 B_{ab}=\frac{2}{\pi[2(b-a)+1]},
 \qquad 0\leq a,b<n.
 \label{appD:cauchy-block}
\end{equation}
Cauchy's determinant formula makes $B$ nonsingular, so $H_{2n}$ has nullity
zero.  For $H_{2n-1}$ the corresponding block is $n\times(n-1)$, and its
maximal minors are nonsingular Cauchy matrices.  It therefore has rank $n-1$,
which leaves exactly one zero eigenvalue in $H_{2n-1}$.  Spectral symmetry and
Eq.~\eqref{appD:entropy-function} now give
\begin{equation}
 \operatorname{Tr}\eta(H_{2n})=2S_{2n},
 \qquad
 \operatorname{Tr}\eta(H_{2n-1})=2S_{2n-1}^{\circ}+\log2.
 \label{appD:nullity-bookkeeping}
\end{equation}
Applying the endpoint limit along $N=2n$ proves
$S_{2n}-S_{2n-1}^{\circ}\to\tfrac12\log2$.  This is the independent Toeplitz
input used in Eq.~\eqref{appD:ordered-physical-entropy}.

\section{Finite joint character and Virasoro scaling limit}
\label{app:character-theorem}

The modified translation $T_\sigma$ refines the exact odd-cycle energies.  We
first formulate the finite-size statement inside a fixed
physical charge block, where the relevant fermionic generators preserve
charge.  Its low-energy scaling limit is then taken through centered marked
measures.  The finite-lattice roots determine the occurring translation
classes, centered lifts retain the conformal-spin branches, and the limiting
marked measure gives the Virasoro-character decomposition.

\subsection{Fixed-charge algebra and physical CAR modes}

Put $N=2L-1$ and let $U_\omega=P_\omega T_\sigma P_\omega$ be the modified
translation restricted to the charge block $\Omega=\omega$, $\omega=\pm1$.
In the Jordan--Wigner convention of Eq.~\eqref{eq:gamma-order}, write
\begin{equation}
 q=(q_1,\ldots,q_N)=(b_1,a_2,b_2,\ldots,a_L,b_L),
 \qquad C=Z_1Z_2=-iq_1q_2,
 \label{appE:registered-chart}
\end{equation}
with $a_1$ the spectator.  Direct composition of the circuit in
Eq.~\eqref{eq:Tsigma-circuit} gives, for
$\Phi(O)=T_\sigma OT_\sigma^\dagger$,
\begin{align}
 \Phi(a_1)&=a_1C,&
 \Phi(q_r)&=Cq_{r+2} &&(1\leq r\leq N-2),\nonumber\\
 \Phi(q_{N-1})&=-i\Omega q_2,&
 \Phi(q_N)&=+i\Omega q_1.
 \label{appE:registered-Majorana-action}
\end{align}
This action is nonlinear on bare odd Majoranas, so replacing $\Omega$ by a
number inside $q_r$ would not be defined within the physical fixed-charge
algebra.  In a fixed charge block the even edges
\begin{equation}
 e_r=-iq_rq_{r+1}\quad(1\leq r<N),
 \qquad e_N=-i\omega q_Nq_1
 \label{appE:fixed-charge-edges}
\end{equation}
generate the full matrix algebra on $P_\omega\mathcal H$.  Substitution into
Eq.~\eqref{appE:registered-Majorana-action} shows that their adjoint action is
implemented by the orthogonal shift
\begin{equation}
 R_\omega(a_1)=a_1,
 \qquad R_\omega(q_r)=q_{r+2},
 \qquad q_{r+N}=\omega q_r.
 \label{appE:orthogonal-shift}
\end{equation}
Consequently $\operatorname{Ad}_{U_\omega}$ is the Spin implementation of
$R_\omega$ on the complete fixed-charge algebra, up to one scalar.

The boundary condition in Eq.~\eqref{appE:orthogonal-shift} requires
$e^{i\kappa N}=\omega$.  Choose the printed positive-energy Fourier operators
\begin{equation}
 f_{\omega n}^\dagger=\frac1{\sqrt{2N}}
 \sum_{r=1}^{N}e^{i\kappa_{\omega n}r}q_r,
 \qquad
 \epsilon_n(L)=4\sin\frac{\pi n}{N},
 \qquad 1\leq n\leq L-1,
 \label{appE:fourier-creators}
\end{equation}
where
\begin{equation}
 \kappa_{\omega n}=
 \begin{cases}
  -\pi n/N,&(-1)^n=\omega,\\
  \pi+\pi n/N,&(-1)^n=-\omega.
 \end{cases}
 \label{appE:allowed-momenta}
\end{equation}
The bare $f_{\omega n}^\dagger$ is odd and therefore does not act within one
charge block.  With $a_1$ the spectator Majorana, the charge-preserving
physical CAR generators are instead
\begin{equation}
 \psi_{\omega n}^\dagger=a_1f_{\omega n}^\dagger,
 \qquad
 \psi_{\omega n}=-a_1f_{\omega n}.
 \label{appE:physical-CAR}
\end{equation}
They commute with $\Omega$ and satisfy the $(L-1)$-mode CAR on
$P_\omega\mathcal H$.  Substituting the Fourier sum into
Eq.~\eqref{appE:orthogonal-shift} gives
$R_\omega(f_{\omega n}^\dagger)=e^{-2i\kappa_{\omega n}}
 f_{\omega n}^\dagger$.  The two cases in
Eq.~\eqref{appE:allowed-momenta} obey
$e^{-2i\kappa_{\omega n}}=e^{2\pi i\omega(-1)^nn/N}$, while $a_1$ is fixed.
Therefore the exact physical adjoint action is
\begin{equation}
 U_\omega\psi_{\omega n}^\dagger U_\omega^\dagger
 =e^{2\pi i\ell_{\omega n}/N}\psi_{\omega n}^\dagger,
 \qquad
 \ell_{\omega n}=\omega(-1)^n n.
 \label{appE:mode-action}
\end{equation}
This fixed-charge statement is the well-defined projection of the spin-space
Clifford action in Eq.~\eqref{appE:registered-Majorana-action}; it is not
obtained by replacing $\Omega$ with a
number inside a bare odd Majorana operator.

The active real zero mode $z_\omega$ and $a_1$ form a complex zero mode
$c_{\omega0}=(a_1+iz_\omega)/2$.  To fix its physical occupation, introduce
the real Fourier pairs
\begin{equation}
 x_{\omega n}=f_{\omega n}+f_{\omega n}^\dagger,
 \qquad
 y_{\omega n}=i(f_{\omega n}^\dagger-f_{\omega n})
 =-\sqrt{\frac2N}\sum_{r=1}^{N}
   \sin(\kappa_{\omega n}r)q_r.
 \label{appE:real-Fourier-pairs}
\end{equation}
Let $B_\omega$ be the orthogonal change of basis whose ordered rows are
$(z_\omega,x_{\omega1},y_{\omega1},\ldots)$.  For $\omega=+1$ these rows are
the positively oriented real Fourier Vandermonde.  The $\omega=-1$ basis is
obtained by multiplying the $q_r$ columns by $(-1)^r$ and reversing the sign
of each of the $L-1$ printed sine rows.  Their determinant ratio is
\begin{equation}
 \prod_{r=1}^{N}(-1)^r\,(-1)^{L-1}
 =(-1)^{N(N+1)/2+L-1}=-1,
\end{equation}
and hence
\begin{equation}
 \det B_\omega=\omega.
 \label{appE:Fourier-orientation}
\end{equation}
Using
\begin{equation}
 \Omega=(-i)^La_1q_1\cdots q_N,
 \qquad
 q_1\cdots q_N=\omega z_\omega
 \prod_{n=1}^{L-1}x_{\omega n}y_{\omega n},
 \label{appE:volume-form}
\end{equation}
together with the occupation convention defined by
$c_{\omega0}$ and $f_{\omega n}$ gives the exact charge identity
\begin{equation}
 \Omega=\omega(-1)^{n_0+\sum_{n=1}^{L-1}n_n},
 \label{appE:charge-selector}
\end{equation}
where $n_0$ and $n_n$ are the zero- and positive-mode occupations.  Hence a
physical state with occupied positive-mode set
$A\subseteq\{1,\ldots,L-1\}$ has the unique zero-mode occupation
\begin{equation}
 n_0(A)=|A|\pmod2.
 \label{appE:zero-selector}
\end{equation}
There is no additional zero-mode degeneracy in a fixed charge sector.  The
CAR relations make the $2^{L-1}$ states
\begin{equation}
 |A;\omega\rangle=
 \prod_{n\in A}^{\nearrow}\psi_{\omega n}^\dagger
 |\varnothing;\omega\rangle
 \label{appE:fock-basis}
\end{equation}
linearly independent.  Their number equals
$\dim P_\omega\mathcal H=2^{L-1}$, so they form a basis of the charge block.

Let $W_\omega$ be the second-quantized implementation of
Eq.~\eqref{appE:mode-action}, normalized to fix
$|\varnothing;\omega\rangle$.  Since the physical CAR generate the full
fixed-charge matrix algebra, $U_\omega=t_{0,\omega}W_\omega$ for one
state-independent scalar.  The scalar is fixed by traces, which can be
evaluated directly in the computational $Z$ basis.  For
$\operatorname{Tr}T_\sigma$, the shift constraint leaves only the two constant
bit strings; the seam and outer-$V$ phases cancel and each diagonal amplitude
is $2^{-1/2}$.  For $\operatorname{Tr}(\Omega T_\sigma)$, the output must
instead be the bitwise complement of the input.  The shift condition then
leaves only the two alternating strings, and the Hadamard, seam, and $V$
phases give each amplitude $i(-1)^{L+1}/\sqrt2$.  Thus
\begin{equation}
 \operatorname{Tr}T_\sigma=\sqrt2,
 \qquad
 \operatorname{Tr}(\Omega T_\sigma)=i(-1)^{L+1}\sqrt2,
 \label{appE:circuit-traces}
\end{equation}
so
$\operatorname{Tr}U_\omega=e^{i\pi\omega(-1)^{L+1}/4}$.  On the other hand,
\begin{equation}
 \operatorname{Tr}W_\omega=
 \prod_{n=1}^{L-1}
 \left(1+e^{2\pi i\omega(-1)^nn/N}\right).
 \label{appE:Fock-trace}
\end{equation}
Using
$\prod_{n=1}^{L-1}2\cos(\pi n/N)=1$ and
$\sum_{n=1}^{L-1}(-1)^nn=[(-1)^{L+1}N-1]/4$ shows that
$\operatorname{Tr}W_\omega\neq0$.  Since
$U_\omega=t_{0,\omega}W_\omega$, the quotient
$\operatorname{Tr}U_\omega/\operatorname{Tr}W_\omega$ fixes the vacuum
eigenvalue
\begin{equation}
 t_{0,\omega}=\alpha_{\omega,L}
 =e^{i\pi\omega/(4N)}.
 \label{appE:vacuum-scalar}
\end{equation}
The power relation $U_\omega^N=e^{i\pi\omega/4}$ is consistent with this
value, but by itself would determine only an $N$th-root class, not the
occurring scalar.

For a subset $A$, define
\begin{equation}
 E_{\omega,L}(A)=E_{0,L}+\sum_{n\in A}\epsilon_n(L),
 \qquad
 r_\omega(A)=\sum_{n\in A}\omega(-1)^nn,
 \qquad
 k_{\omega,L}(A)=\operatorname{cent}_N r_\omega(A),
 \label{appE:energy-root-labels}
\end{equation}
where
$\operatorname{cent}_N:\mathbb Z_N\to\{-(L-1),\ldots,L-1\}$ is the
centered representative.  Equations~\eqref{appE:mode-action} and
\eqref{appE:vacuum-scalar} give the exact relative root
\begin{equation}
 \frac{t_{\omega,L}(A)}{t_{0,\omega}}
 =\exp\!\left(\frac{2\pi i}{N}r_\omega(A)\right)
 =\exp\!\left(\frac{2\pi i}{N}k_{\omega,L}(A)\right).
 \label{eq:exact-root-ratio}
\end{equation}
With a formal variable $z$ satisfying $z^N=1$, the finite occupation--root
character is
\begin{equation}
 \mathcal C_{\omega,L}(\mathbf x,z)=
 \prod_{n=1}^{L-1}\left(1+x_nz^{\omega(-1)^nn}\right).
 \label{eq:finite-joint-character}
\end{equation}
Every monomial records an actual physical Fock state.  The exact joint
multiplicity is
\begin{equation}
 M_{\omega,L}(E,k)=\#\left\{A\subseteq\{1,\ldots,L-1\}:
 E_{\omega,L}(A)=E,\ k_{\omega,L}(A)=k\right\}.
 \label{appE:joint-multiplicity}
\end{equation}
This grouping remains valid through accidental energy collisions because both
energy and centered root are fixed.  The product proves occurrence and
multiplicity of its roots; it does not claim that every algebraically allowed
$N$th root occurs.

\subsection{Branch-safe spin classes and finite witnesses}

A finite witness, in the sense used here, is an explicit physical state at
finite size exhibited in order to show that a claimed spin class is actually
attained.  Each state constructed below is available at every allowed
circumference, so the construction adds what the algebraic power relation
cannot supply on its own; it does not trade an all-length argument for a
numerical sample.

The effective circumference determined independently by
Eq.~\eqref{eq:Ekw-series} is
$L_{\rm eff}=L-\tfrac12=N/2$.  Together with the velocity normalization in
Eq.~\eqref{eq:circle-Casimir}, this fixes the scaling cylinder used below.
One application of the dressed right shift $R_\omega$ in
Eq.~\eqref{appE:orthogonal-shift} advances the local generators by one
original-lattice spacing.
A scaling state of conformal spin $s=h-\bar h$ therefore acquires the phase
\begin{equation}
 e^{2\pi i s/L_{\rm eff}}=e^{4\pi i s/N}.
 \label{appE:translation-spin-normalization}
\end{equation}
Combining this normalization with the exact eigenphase
$\alpha_{\omega,L}e^{2\pi ik/N}$ gives the convention-fixed spin class
\begin{equation}
 s(\omega,k)=\frac{\omega}{16}+\frac{k}{2}\pmod1.
 \label{appE:spin-lift}
\end{equation}
Equation~\eqref{appE:spin-lift} is the spin class obtained after the declared
centered-section choice for an occurring finite-size label.  It does not
define a global additive lift of cyclic momentum through centered wraps.

The physical subsets $A=\varnothing$ and $A=\{1\}$ exist in each charge block
for every $L\geq2$.  Their labels are
\begin{equation}
\begin{array}{c|c|c|c}
 \omega&A&k_{\omega,L}(A)&s(\omega,k)\\ \hline
 +1&\varnothing&0&1/16\\
 +1&\{1\}&-1&-7/16\\
 -1&\varnothing&0&-1/16\\
 -1&\{1\}&+1&7/16
\end{array}
\label{appE:four-witnesses}
\end{equation}
Thus all four classes
$\{1/16,-1/16,7/16,-7/16\}\pmod1$ occur at finite lattice size.  This
conclusion uses the explicit Fock basis and scalar, not merely the power
identity for $T_\sigma$.

\subsection{Eventual no-wrap and the marked scaling measure}

For scaling excitations, write
\begin{equation}
 d_{n,L}=\frac{N}{2\pi}\sin\frac{\pi n}{N},
 \qquad
 D_L(A)=\sum_{n\in A}d_{n,L}.
 \label{appE:scaled-energy}
\end{equation}
The centered excitation-spin marker and its two chiral coordinates are
\begin{equation}
 p_{\omega,L}(A)=\frac{k_{\omega,L}(A)}2,
 \qquad
 u_{\omega,L}(A)=\frac{D_L(A)+p_{\omega,L}(A)}2,
 \qquad
 v_{\omega,L}(A)=\frac{D_L(A)-p_{\omega,L}(A)}2.
 \label{appE:chiral-coordinates}
\end{equation}
These are finite-size spectral markers, not eigenvalues of separately defined
finite-size chiral Hamiltonians.  Define the locally finite marked measure
\begin{equation}
 \mu_{\omega,L}=\sum_{A\subseteq\{1,\ldots,L-1\}}
 \delta_{(u_{\omega,L}(A),v_{\omega,L}(A))}.
 \label{appE:marked-measure}
\end{equation}

For $1\leq n\leq L-1$, the chord bound
$\sin(\pi n/N)\geq2n/N$ gives
\begin{equation}
 d_{n,L}\geq\frac n\pi.
 \label{appE:energy-lower-bound}
\end{equation}
For a compact set $K\subset\mathbb R^2$, put
$R=\max\{0,\sup_{(u,v)\in K}(u+v)\}$.  Since
$D_L(A)=u_{\omega,L}(A)+v_{\omega,L}(A)$, every atom in $K$ satisfies
$D_L(A)\leq R$.  Equation~\eqref{appE:energy-lower-bound} then gives
$\sum_{n\in A}n\leq\pi R$.  Only finitely many subsets can occur in that
window, uniformly in $L$, and their uncentered labels $r_\omega(A)$ are
uniformly bounded.  For all sufficiently large $L$,
\begin{equation}
 k_{\omega,L}(A)=r_\omega(A)
 \label{appE:eventual-no-wrap}
\end{equation}
for every such subset.  Since $d_{n,L}\to n/2$ mode by mode,
integration against a compactly supported continuous function becomes a
uniformly finite sum of convergent terms.  Therefore
\begin{equation}
 \mu_{\omega,L}\xrightarrow[L\to\infty]{\rm vague}\mu_\omega,
 \label{appE:vague-convergence}
\end{equation}
where
\begin{equation}
 \mu_\omega=\sum_{A\Subset\mathbb N}
 \delta_{(u_\omega(A),v_\omega(A))},
 \quad
 \begin{aligned}
 u_\omega(A)&=\frac14\sum_{n\in A}
     [n+\omega(-1)^nn],\\
 v_\omega(A)&=\frac14\sum_{n\in A}
     [n-\omega(-1)^nn].
 \end{aligned}
 \label{appE:limiting-measure}
\end{equation}
In the limiting measure, each summand in $u_\omega$ and $v_\omega$ is either
zero or $n/2$.  Hence both coordinates are nonnegative and
$u_\omega(A)+v_\omega(A)=\tfrac12\sum_{n\in A}n$, which proves local
finiteness.  The convergence is stabilization below every fixed excitation
cutoff, not coefficientwise convergence of finite products with moving sine
exponents in one formal-series ring.

For a single occupied mode, Eq.~\eqref{appE:limiting-measure} allocates
\begin{equation}
(\delta h,\delta\bar h)=
\begin{cases}
 (n/2,0),&\omega(-1)^n=+1,\\
 (0,n/2),&\omega(-1)^n=-1.
\end{cases}
\label{appE:single-mode-chirality}
\end{equation}
For $\omega=+1$, the even modes $n=2m$ occupy the left chirality and the odd
modes $n=2m-1$ occupy the right chirality; for $\omega=-1$ the allocation is
reversed.  Thus the right shift $R_\omega$, before any CFT character is named,
fixes the limiting products
\begin{align}
 G_+(q,\bar q)&=\prod_{m\geq1}(1+q^m)
                 \prod_{m\geq1}(1+\bar q^{m-1/2}),
 \label{appE:Gplus}\\
 G_-(q,\bar q)&=\prod_{m\geq1}(1+q^{m-1/2})
                 \prod_{m\geq1}(1+\bar q^m).
 \label{appE:Gminus}
\end{align}
Replacing the right shift by its inverse would exchange the two chiralities
and would describe a different operator.

\subsection{Virasoro characters and four primary pairs}

The vacuum witnesses in Eq.~\eqref{appE:four-witnesses}, together with the
independent Casimir result $\Delta_{\rm def}=1/16$ from
Eq.~\eqref{eq:Delta-sigma}, select
$(h_{+,0},\bar h_{+,0})=(1/16,0)$ and
$(h_{-,0},\bar h_{-,0})=(0,1/16)$.  Restoring the $c=1/2$ vacuum powers gives
\begin{align}
 Z_+(q,\bar q)&=q^{1/16-c/24}\bar q^{-c/24}G_+(q,\bar q),
 \label{appE:Zplus}\\
 Z_-(q,\bar q)&=q^{-c/24}\bar q^{1/16-c/24}G_-(q,\bar q).
 \label{appE:Zminus}
\end{align}
Jacobi's products for the Ising characters read
\begin{align}
 \chi_\sigma(q)&=q^{1/24}\prod_{m\geq1}(1+q^m),
 \label{appE:sigma-product}\\
 \chi_1(q)+\chi_\epsilon(q)
 &=q^{-1/48}\prod_{m\geq1}(1+q^{m-1/2}).
 \label{appE:NS-product}
\end{align}
There is no additional factor of two in a fixed charge sector.  The Ramond
normalization is already contained in Eq.~\eqref{appE:sigma-product}.  To
resolve the Neveu--Schwarz modules, set
\begin{equation}
 P_{\rm NS}^{\pm}(q)=\prod_{m\geq1}(1\pm q^{m-1/2}).
\end{equation}
Then
\begin{equation}
 \chi_1(q)=q^{-1/48}\frac{P_{\rm NS}^{+}(q)+P_{\rm NS}^{-}(q)}2,
 \qquad
 \chi_\epsilon(q)=q^{-1/48}
 \frac{P_{\rm NS}^{+}(q)-P_{\rm NS}^{-}(q)}2.
 \label{appE:NS-parity-projection}
\end{equation}
Define the finite-state NS parity in either charge block by
$|A\cap(2\mathbb N-1)|\pmod2$, because the half-integer factors in both
Eqs.~\eqref{appE:Gplus} and~\eqref{appE:Gminus} arise from the odd mode
indices $n=2m-1$.  Even NS occupation parity gives the identity module,
whereas odd NS parity gives the energy module.  Combining these identities with
Eqs.~\eqref{appE:Zplus}--\eqref{appE:Zminus} yields
\begin{equation}
 \begin{aligned}
 Z_{\rm KW}(q,\bar q)&=Z_+(q,\bar q)+Z_-(q,\bar q)\\
 &=\chi_\sigma(q)
   [\bar\chi_1(\bar q)+\bar\chi_\epsilon(\bar q)]\\
 &\quad+[\chi_1(q)+\chi_\epsilon(q)]
   \bar\chi_\sigma(\bar q).
 \end{aligned}
 \label{appE:Virasoro-decomposition}
\end{equation}
The four primary pairs and their finite witnesses are therefore
\begin{equation}
\begin{array}{c|c|c|c|c}
 \omega&\text{NS parity}&(h,\bar h)&\text{primary pair}&
 \text{finite witness}\\ \hline
 +1&\text{even}&(1/16,0)&(\sigma,1)&\varnothing\\
 +1&\text{odd}&(1/16,1/2)&(\sigma,\epsilon)&\{1\}\\
 -1&\text{even}&(0,1/16)&(1,\sigma)&\varnothing\\
 -1&\text{odd}&(1/2,1/16)&(\epsilon,\sigma)&\{1\}
\end{array}
\label{appE:four-towers}
\end{equation}
These four primary pairs realize
Eq.~\eqref{appE:Virasoro-decomposition}.  The one-mode witness occupies the
lowest NS half-integer mode and changes NS parity.  It is the
character-theoretic bottom of the $\epsilon$ module in that chirality, not a
descendant of the identity primary.

Diagonal specialization gives
\begin{equation}
 G_+(x,x)+G_-(x,x)=2\prod_{n\geq1}(1+x^{n/2}).
 \label{appE:diagonal-control}
\end{equation}
The diagonal identity shows that energy multiplicities alone cannot determine
chirality, the signs of conformal spin, or the placement of
the four primary pairs.  The finite-size joint polynomial and the marked
scaling-counting and chirality statements above are unconditional in the
lattice convention fixed above.  Naming the limiting products as Virasoro
characters and decomposing them into the four primary pairs additionally use
the standard $c=1/2$ Ising/Jacobi identities and the effective-length,
velocity, and translation conventions fixed by Eqs.~\eqref{eq:Ekw-series},
\eqref{eq:circle-Casimir}, and \eqref{appE:orthogonal-shift}.  No statement
constructs finite-size
Virasoro operators or asserts global additive lifts, root saturation, or
formal-series convergence.

\section{Finite-size Gaussian single-particle levels}\label{app:lowest-level-illustration}

Table~\ref{tab:state-prescriptions} records the sectorwise physical Gaussian
states and auxiliary Gaussian comparators used to illustrate the distinction
proved in Section~\ref{sec:entanglement-es}; none is the charge-unresolved
physical RDM.  The tabulated $\nu_1$ and $\xi_1$ are Gaussian single-particle
levels, not the many-body entanglement spectrum of
Ref.~\cite{LiHaldane2008}; they are obtained by the correlation-matrix
method~\cite{Peschel2003,PeschelEisler2009}.
\begin{table}[ht]
\centering
\begin{tabular}{@{}lll@{}}
\toprule
system & covariance object & status and finite-$L$ lowest level\\
\midrule
homogeneous & physical ground covariance & physical; $\nu_1>0$ on the grid\\
$\eta$ defect & $\operatorname{Pf}(\Gamma)=+1$ completion & physical; $\nu_1>0$ on the grid\\
KW, $\Omega=\omega$ & $\Gamma_{\omega,{\rm phys}}$ & physical; infrared-soft level\\
KW, $\Omega=\omega$ & $\Gamma_{\omega,0}$ & auxiliary; exact $\nu_1=\xi_1=0$\\
\bottomrule
\end{tabular}
\caption{Finite-size covariance objects used in the Gaussian single-particle
illustration.  Here $\nu_1$ is the smallest nonnegative eigenvalue of
$i\Gamma_A$ and $\xi_1=\log[(1+\nu_1)/(1-\nu_1)]$.  The sign-zero
KW row is auxiliary and is not the physical charge-unresolved RDM.}
\label{tab:state-prescriptions}
\end{table}

The finite grid, on which every restriction retains a complete non-wrapping
prefix of $2\ell$ spin sites, is
\begin{equation}
 L\in\{64,128,256\},\qquad \ell\in\{4,8,16\}.
 \label{eq:ES-grid}
\end{equation}
Every physical pure-state spectrum is computed from the complete site
restriction; no Majorana is deleted.  The auxiliary sign-zero comparator uses
the same complete CAR restriction.  Figure~\ref{fig:level-resolved-es}
compares the resulting lowest single-particle level across the three physical
systems and the auxiliary comparator on this finite grid.  Line style identifies
the covariance object, while marker shape identifies the block size $\ell$.

\begin{figure}[ht]
 \centering
 \includegraphics[width=0.76\textwidth]{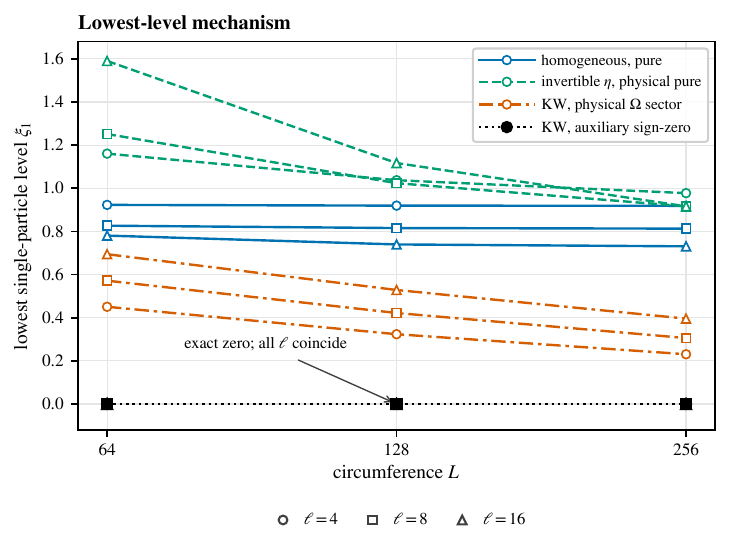}
 \caption{\textbf{Lowest Gaussian single-particle level.}  The level $\xi_1$
 for the homogeneous physical ground state (solid blue), the invertible-$\eta$
 physical pure state (dashed green), the Pfaffian-selected physical state in
either KW charge block (dash-dotted orange), and the auxiliary sign-zero KW
covariance (dotted black).  Circles, squares, and triangles denote
$\ell=4,8,16$, respectively; all three auxiliary sign-zero series coincide
exactly at $\xi_1=0$.  The auxiliary comparator has exact $\xi_1=0$,
whereas either physical KW sector obeys the analytic softening bound in
Eq.~\eqref{eq:pure-xi-bound}; the two mathematical completions within one
charge block are Gaussian-isospectral on complete prefixes.}
 \label{fig:level-resolved-es}
\end{figure}

At fixed block, the lines in Figure~\ref{fig:level-resolved-es} are guides
connecting the three sampled circumferences.  On this finite grid, the
homogeneous, physical $\eta$, and physical KW levels are positive, while the
auxiliary sign-zero level is exactly zero.  The physical KW level also obeys
the analytic softening ceiling in Eq.~\eqref{eq:pure-xi-bound}; no nonzero
finite-$L$ lower bound is claimed.  Figure~\ref{fig:level-resolved-es} is
therefore an auxiliary finite-grid illustration, not an input to any theorem.
It makes no statement about higher entanglement levels or about the many-body
spectrum of the physical convex mixture.

\paragraph{Code availability.}
Python/NumPy code for the finite-size reduced-state, ordered-entropy, and joint energy--modified-translation calculations is available in the Zenodo software record at \href{https://doi.org/10.5281/zenodo.21523746}{\texttt{doi:10.5281/zenodo.21523746}}.

\acknowledgments

During the preparation of this manuscript, the author used ChatGPT (version 5.5) to improve its linguistic clarity and readability. The author subsequently reviewed and edited the entire manuscript and takes full responsibility for its final content.

\bibliographystyle{unsrturl}
\bibliography{refs}

\begin{thebibliography}{10}

\bibitem{BhardwajEtAl2024}
Lakshya Bhardwaj, Lea~E. Bottini, Ludovic Fraser-Taliente, Liam Gladden, Dewi
  S.~W. Gould, Arthur Platschorre, and Hannah Tillim.
\newblock Lectures on generalized symmetries.
\newblock {\em Phys. Rept.}, 1051:1--87, 2024.
\newblock \href {http://arxiv.org/abs/2307.07547} {\path{arXiv:2307.07547}},
  \href {https://doi.org/10.1016/j.physrep.2023.11.002}
  {\path{doi:10.1016/j.physrep.2023.11.002}}.

\bibitem{Shao2023TASI}
Shu-Heng Shao.
\newblock What's done cannot be undone: {TASI} lectures on non-invertible
  symmetries.
\newblock \href{https://arxiv.org/abs/2308.00747}{arXiv:2308.00747}, 2023.
\newblock {TASI} 2023 lecture notes.
\newblock URL: \url{https://arxiv.org/abs/2308.00747}, \href
  {http://arxiv.org/abs/2308.00747} {\path{arXiv:2308.00747}}.

\bibitem{SchaferNameki2024}
Sakura Sch\"{a}fer-Nameki.
\newblock {ICTP} lectures on (non-)invertible generalized symmetries.
\newblock {\em Phys. Rept.}, 1063:1--55, 2024.
\newblock \href {http://arxiv.org/abs/2305.18296} {\path{arXiv:2305.18296}},
  \href {https://doi.org/10.1016/j.physrep.2024.01.007}
  {\path{doi:10.1016/j.physrep.2024.01.007}}.

\bibitem{BhardwajTachikawa2018}
Lakshya Bhardwaj and Yuji Tachikawa.
\newblock On finite symmetries and their gauging in two dimensions.
\newblock {\em J. High Energy Phys.}, 03(03):189, 2018.
\newblock \href {http://arxiv.org/abs/1704.02330} {\path{arXiv:1704.02330}},
  \href {https://doi.org/10.1007/JHEP03(2018)189}
  {\path{doi:10.1007/JHEP03(2018)189}}.

\bibitem{KramersWannier1941I}
H.~A. Kramers and G.~H. Wannier.
\newblock Statistics of the two-dimensional ferromagnet. part i.
\newblock {\em Phys. Rev.}, 60(3):252--262, 1941.
\newblock \href {https://doi.org/10.1103/PhysRev.60.252}
  {\path{doi:10.1103/PhysRev.60.252}}.

\bibitem{KramersWannier1941II}
H.~A. Kramers and G.~H. Wannier.
\newblock Statistics of the two-dimensional ferromagnet. part ii.
\newblock {\em Phys. Rev.}, 60(3):263--276, 1941.
\newblock \href {https://doi.org/10.1103/PhysRev.60.263}
  {\path{doi:10.1103/PhysRev.60.263}}.

\bibitem{Onsager1944}
Lars Onsager.
\newblock Crystal statistics. {I}. {A} two-dimensional model with an
  order-disorder transition.
\newblock {\em Phys. Rev.}, 65:117--149, 1944.
\newblock \href {https://doi.org/10.1103/PhysRev.65.117}
  {\path{doi:10.1103/PhysRev.65.117}}.

\bibitem{BelavinPolyakovZamolodchikov1984}
A.~A. Belavin, A.~M. Polyakov, and A.~B. Zamolodchikov.
\newblock Infinite conformal symmetry in two-dimensional quantum field theory.
\newblock {\em Nucl. Phys. B}, 241:333--380, 1984.
\newblock \href {https://doi.org/10.1016/0550-3213(84)90052-X}
  {\path{doi:10.1016/0550-3213(84)90052-X}}.

\bibitem{FrohlichFuchsRunkelSchweigert2004}
J\"{u}rg Fr\"{o}hlich, J\"{u}rgen Fuchs, Ingo Runkel, and Christoph Schweigert.
\newblock Kramers-wannier duality from conformal defects.
\newblock {\em Phys. Rev. Lett.}, 93:070601, 2004.
\newblock \href {http://arxiv.org/abs/cond-mat/0404051}
  {\path{arXiv:cond-mat/0404051}}, \href
  {https://doi.org/10.1103/PhysRevLett.93.070601}
  {\path{doi:10.1103/PhysRevLett.93.070601}}.

\bibitem{Ishibashi1989}
Nobuyuki Ishibashi.
\newblock The boundary and crosscap states in conformal field theories.
\newblock {\em Mod. Phys. Lett. A}, 4:251--264, 1989.
\newblock \href {https://doi.org/10.1142/S0217732389000320}
  {\path{doi:10.1142/S0217732389000320}}.

\bibitem{Cardy1989}
John~L. Cardy.
\newblock Boundary conditions, fusion rules and the {Verlinde} formula.
\newblock {\em Nucl. Phys. B}, 324(3):581--596, 1989.
\newblock \href {https://doi.org/10.1016/0550-3213(89)90521-X}
  {\path{doi:10.1016/0550-3213(89)90521-X}}.

\bibitem{PetkovaZuber2001}
Valentina~B. Petkova and Jean-Bernard Zuber.
\newblock Generalised twisted partition functions.
\newblock {\em Phys. Lett. B}, 504:157--164, 2001.
\newblock \href {http://arxiv.org/abs/hep-th/0011021}
  {\path{arXiv:hep-th/0011021}}, \href
  {https://doi.org/10.1016/S0370-2693(01)00276-3}
  {\path{doi:10.1016/S0370-2693(01)00276-3}}.

\bibitem{BachasBrunner2008}
Constantin Bachas and Ilka Brunner.
\newblock Fusion of conformal interfaces.
\newblock {\em J. High Energy Phys.}, 02:085, 2008.
\newblock \href {http://arxiv.org/abs/0712.0076} {\path{arXiv:0712.0076}},
  \href {https://doi.org/10.1088/1126-6708/2008/02/085}
  {\path{doi:10.1088/1126-6708/2008/02/085}}.

\bibitem{FrohlichFuchsRunkelSchweigert2007}
J\"{u}rg Fr\"{o}hlich, J\"{u}rgen Fuchs, Ingo Runkel, and Christoph Schweigert.
\newblock Duality and defects in rational conformal field theory.
\newblock {\em Nucl. Phys. B}, 763:354--430, 2007.
\newblock \href {http://arxiv.org/abs/hep-th/0607247}
  {\path{arXiv:hep-th/0607247}}, \href
  {https://doi.org/10.1016/j.nuclphysb.2006.11.017}
  {\path{doi:10.1016/j.nuclphysb.2006.11.017}}.

\bibitem{OshikawaAffleck1996}
Masaki Oshikawa and Ian Affleck.
\newblock Defect lines in the {Ising} model and boundary states on orbifolds.
\newblock {\em Phys. Rev. Lett.}, 77(13):2604--2607, 1996.
\newblock \href {http://arxiv.org/abs/hep-th/9606177}
  {\path{arXiv:hep-th/9606177}}, \href
  {https://doi.org/10.1103/PhysRevLett.77.2604}
  {\path{doi:10.1103/PhysRevLett.77.2604}}.

\bibitem{OshikawaAffleck1997}
Masaki Oshikawa and Ian Affleck.
\newblock Boundary conformal field theory approach to the critical
  two-dimensional {Ising} model with a defect line.
\newblock {\em Nucl. Phys. B}, 495(3):533--582, 1997.
\newblock \href {http://arxiv.org/abs/cond-mat/9612187}
  {\path{arXiv:cond-mat/9612187}}, \href
  {https://doi.org/10.1016/S0550-3213(97)00219-8}
  {\path{doi:10.1016/S0550-3213(97)00219-8}}.

\bibitem{AasenMongFendley2016}
David Aasen, Roger S.~K. Mong, and Paul Fendley.
\newblock Topological defects on the lattice: I. the {Ising} model.
\newblock {\em J. Phys. A: Math. Theor.}, 49(35):354001, 2016.
\newblock \href {http://arxiv.org/abs/1601.07185} {\path{arXiv:1601.07185}},
  \href {https://doi.org/10.1088/1751-8113/49/35/354001}
  {\path{doi:10.1088/1751-8113/49/35/354001}}.

\bibitem{AasenFendleyMong2020}
David Aasen, Paul Fendley, and Roger S.~K. Mong.
\newblock Topological defects on the lattice: Dualities and degeneracies.
\newblock \href{https://arxiv.org/abs/2008.08598}{arXiv:2008.08598}, 2020.
\newblock URL: \url{https://arxiv.org/abs/2008.08598}, \href
  {http://arxiv.org/abs/2008.08598} {\path{arXiv:2008.08598}}.

\bibitem{Grimm2002}
Uwe Grimm.
\newblock Spectrum of a duality-twisted {Ising} quantum chain.
\newblock {\em J. Phys. A: Math. Gen.}, 35(3):L25--L30, 2002.
\newblock \href {http://arxiv.org/abs/hep-th/0111157}
  {\path{arXiv:hep-th/0111157}}, \href
  {https://doi.org/10.1088/0305-4470/35/3/101}
  {\path{doi:10.1088/0305-4470/35/3/101}}.

\bibitem{CalabreseCardy2004}
Pasquale Calabrese and John Cardy.
\newblock Entanglement entropy and quantum field theory.
\newblock {\em J. Stat. Mech.}, 0406:P06002, 2004.
\newblock \href {http://arxiv.org/abs/hep-th/0405152}
  {\path{arXiv:hep-th/0405152}}, \href
  {https://doi.org/10.1088/1742-5468/2004/06/P06002}
  {\path{doi:10.1088/1742-5468/2004/06/P06002}}.

\bibitem{CalabreseCardy2009}
Pasquale Calabrese and John Cardy.
\newblock Entanglement entropy and conformal field theory.
\newblock {\em J. Phys. A: Math. Theor.}, 42:504005, 2009.
\newblock \href {http://arxiv.org/abs/0905.4013} {\path{arXiv:0905.4013}},
  \href {https://doi.org/10.1088/1751-8113/42/50/504005}
  {\path{doi:10.1088/1751-8113/42/50/504005}}.

\bibitem{SeibergShao2024}
Nathan Seiberg and Shu-Heng Shao.
\newblock Majorana chain and {Ising} model---(non-invertible) translations,
  anomalies, and emanant symmetries.
\newblock {\em SciPost Phys.}, 16:064, 2024.
\newblock \href {http://arxiv.org/abs/2307.02534} {\path{arXiv:2307.02534}},
  \href {https://doi.org/10.21468/SciPostPhys.16.3.064}
  {\path{doi:10.21468/SciPostPhys.16.3.064}}.

\bibitem{WilliamsonEtAl2016}
Dominic~J. Williamson, Nick Bultinck, Micha\"{e}l Mari\"{e}n, Mehmet~B.
  \c{S}ahino\u{g}lu, Jutho Haegeman, and Frank Verstraete.
\newblock Matrix product operators for symmetry-protected topological phases:
  Gauging and edge theories.
\newblock {\em Phys. Rev. B}, 94:205150, 2016.
\newblock \href {http://arxiv.org/abs/1412.5604} {\path{arXiv:1412.5604}},
  \href {https://doi.org/10.1103/PhysRevB.94.205150}
  {\path{doi:10.1103/PhysRevB.94.205150}}.

\bibitem{BultinckEtAl2017}
Nick Bultinck, Micha\"{e}l Mari\"{e}n, Dominic~J. Williamson, Mehmet~B.
  \c{S}ahino\u{g}lu, Jutho Haegeman, and Frank Verstraete.
\newblock Anyons and matrix product operator algebras.
\newblock {\em Ann. Phys. (N.Y.)}, 378:183--233, 2017.
\newblock \href {http://arxiv.org/abs/1511.08090} {\path{arXiv:1511.08090}},
  \href {https://doi.org/10.1016/j.aop.2017.01.004}
  {\path{doi:10.1016/j.aop.2017.01.004}}.

\bibitem{CiracEtAl2021}
J.~Ignacio Cirac, David P\'{e}rez-Garc\'{i}a, Norbert Schuch, and Frank
  Verstraete.
\newblock Matrix product states and projected entangled pair states: Concepts,
  symmetries, theorems.
\newblock {\em Rev. Mod. Phys.}, 93:045003, 2021.
\newblock \href {http://arxiv.org/abs/2011.12127} {\path{arXiv:2011.12127}},
  \href {https://doi.org/10.1103/RevModPhys.93.045003}
  {\path{doi:10.1103/RevModPhys.93.045003}}.

\bibitem{LootensDelcampOrtizVerstraete2023}
Laurens Lootens, Cl{\'e}ment Delcamp, Gerardo Ortiz, and Frank Verstraete.
\newblock Dualities in one-dimensional quantum lattice models: Symmetric
  hamiltonians and matrix product operator intertwiners.
\newblock {\em PRX Quantum}, 4(2):020357, 2023.
\newblock \href {http://arxiv.org/abs/2112.09091} {\path{arXiv:2112.09091}},
  \href {https://doi.org/10.1103/PRXQuantum.4.020357}
  {\path{doi:10.1103/PRXQuantum.4.020357}}.

\bibitem{SakaiSatoh2008}
Kazuhiro Sakai and Yuji Satoh.
\newblock Entanglement through conformal interfaces.
\newblock {\em J. High Energy Phys.}, 12:001, 2008.
\newblock \href {http://arxiv.org/abs/0809.4548} {\path{arXiv:0809.4548}},
  \href {https://doi.org/10.1088/1126-6708/2008/12/001}
  {\path{doi:10.1088/1126-6708/2008/12/001}}.

\bibitem{BachasBrunnerRoggenkamp2013}
Constantin Bachas, Ilka Brunner, and Daniel Roggenkamp.
\newblock Fusion of critical defect lines in the {2D} {Ising} model.
\newblock {\em J. Stat. Mech.}, 1308:P08008, 2013.
\newblock \href {http://arxiv.org/abs/1303.3616} {\path{arXiv:1303.3616}},
  \href {https://doi.org/10.1088/1742-5468/2013/08/P08008}
  {\path{doi:10.1088/1742-5468/2013/08/P08008}}.

\bibitem{RoySaleur2022}
Ananda Roy and Hubert Saleur.
\newblock Entanglement entropy in the {Ising} model with topological defects.
\newblock {\em Phys. Rev. Lett.}, 128:090603, 2022.
\newblock \href {http://arxiv.org/abs/2111.04534} {\path{arXiv:2111.04534}},
  \href {https://doi.org/10.1103/PhysRevLett.128.090603}
  {\path{doi:10.1103/PhysRevLett.128.090603}}.

\bibitem{Rogerson2022}
David Rogerson, Frank Pollmann, and Ananda Roy.
\newblock Entanglement entropy and negativity in the {Ising} model with
  defects.
\newblock {\em J. High Energy Phys.}, 06:165, 2022.
\newblock \href {http://arxiv.org/abs/2204.03601} {\path{arXiv:2204.03601}},
  \href {https://doi.org/10.1007/JHEP06(2022)165}
  {\path{doi:10.1007/JHEP06(2022)165}}.

\bibitem{Rockwood2025}
Gavin Rockwood.
\newblock Entanglement {H}amiltonians for periodic free fermion chains with
  defects.
\newblock {\em J. Stat. Mech.}, 2025:073101, 2025.
\newblock \href {http://arxiv.org/abs/2408.08281} {\path{arXiv:2408.08281}},
  \href {https://doi.org/10.1088/1742-5468/ade2ec}
  {\path{doi:10.1088/1742-5468/ade2ec}}.

\bibitem{NortheRossi2025}
Christian Northe and Paolo Rossi.
\newblock Entanglement through topological defects: Reconciling theory with
  numerics, 2025.
\newblock \href {http://arxiv.org/abs/2512.07974} {\path{arXiv:2512.07974}}.

\bibitem{NorthePolettiRossi2026Interfaces}
Christian Northe, Riccardo Poletti, and Paolo Rossi.
\newblock Entanglement in presence of topological interfaces and dualities,
  2026.
\newblock \href {http://arxiv.org/abs/2607.22451} {\path{arXiv:2607.22451}}.

\bibitem{GoldsteinSela2018}
Moshe Goldstein and Eran Sela.
\newblock Symmetry-resolved entanglement in many-body systems.
\newblock {\em Phys. Rev. Lett.}, 120:200602, 2018.
\newblock \href {http://arxiv.org/abs/1711.09418} {\path{arXiv:1711.09418}},
  \href {https://doi.org/10.1103/PhysRevLett.120.200602}
  {\path{doi:10.1103/PhysRevLett.120.200602}}.

\bibitem{XavierAlcarazSierra2018}
Jos\'{e}~C. Xavier, Francisco~C. Alcaraz, and Germ\'{a}n Sierra.
\newblock Equipartition of the entanglement entropy.
\newblock {\em Phys. Rev. B}, 98:041106(R), 2018.
\newblock \href {http://arxiv.org/abs/1804.06357} {\path{arXiv:1804.06357}},
  \href {https://doi.org/10.1103/PhysRevB.98.041106}
  {\path{doi:10.1103/PhysRevB.98.041106}}.

\bibitem{BonsignoriRuggieroCalabrese2019}
Riccarda Bonsignori, Paola Ruggiero, and Pasquale Calabrese.
\newblock Symmetry resolved entanglement in free fermionic systems.
\newblock {\em J. Phys. A: Math. Theor.}, 52:475302, 2019.
\newblock \href {http://arxiv.org/abs/1907.02084} {\path{arXiv:1907.02084}},
  \href {https://doi.org/10.1088/1751-8121/ab4b77}
  {\path{doi:10.1088/1751-8121/ab4b77}}.

\bibitem{MurcianoDiGiulioCalabrese2020}
Sara Murciano, Giuseppe Di~Giulio, and Pasquale Calabrese.
\newblock Entanglement and symmetry resolution in two dimensional free quantum
  field theories.
\newblock {\em J. High Energy Phys.}, 08:073, 2020.
\newblock \href {http://arxiv.org/abs/2006.09069} {\path{arXiv:2006.09069}},
  \href {https://doi.org/10.1007/JHEP08(2020)073}
  {\path{doi:10.1007/JHEP08(2020)073}}.

\bibitem{Saura2024}
Pablo Saura-Bastida, Arpit Das, Germ\'an Sierra, and Javier Molina-Vilaplana.
\newblock Categorical-symmetry resolved entanglement in conformal field theory.
\newblock {\em Phys. Rev. D}, 109:105026, 2024.
\newblock \href {http://arxiv.org/abs/2402.06322} {\path{arXiv:2402.06322}},
  \href {https://doi.org/10.1103/PhysRevD.109.105026}
  {\path{doi:10.1103/PhysRevD.109.105026}}.

\bibitem{Das2024}
Arpit Das, Javier Molina-Vilaplana, and Pablo Saura-Bastida.
\newblock Generalized symmetry resolution of entanglement in conformal field
  theory for twisted and anyonic sectors.
\newblock {\em Phys. Rev. D}, 110(12):125005, 2024.
\newblock \href {http://arxiv.org/abs/2409.02162} {\path{arXiv:2409.02162}},
  \href {https://doi.org/10.1103/PhysRevD.110.125005}
  {\path{doi:10.1103/PhysRevD.110.125005}}.

\bibitem{ChoiRayhaunZheng2024}
Yichul Choi, Brandon~C. Rayhaun, and Yunqin Zheng.
\newblock Noninvertible symmetry-resolved affleck-ludwig-cardy formula and
  entanglement entropy from the boundary tube algebra.
\newblock {\em Phys. Rev. Lett.}, 133(25):251602, 2024.
\newblock \href {http://arxiv.org/abs/2409.02806} {\path{arXiv:2409.02806}},
  \href {https://doi.org/10.1103/PhysRevLett.133.251602}
  {\path{doi:10.1103/PhysRevLett.133.251602}}.

\bibitem{HeymannQuella2025}
Jared Heymann and Thomas Quella.
\newblock Revisiting the symmetry-resolved entanglement for noninvertible
  symmetries in $1{+}1$d conformal field theories.
\newblock {\em Phys. Rev. D}, 112:025004, 2025.
\newblock \href {http://arxiv.org/abs/2409.02315} {\path{arXiv:2409.02315}},
  \href {https://doi.org/10.1103/lr47-yv3j} {\path{doi:10.1103/lr47-yv3j}}.

\bibitem{ChungPeschel2001}
Ming-Chiang Chung and Ingo Peschel.
\newblock Density-matrix spectra of solvable fermionic systems.
\newblock {\em Phys. Rev. B}, 64:064412, 2001.
\newblock \href {http://arxiv.org/abs/cond-mat/0103301}
  {\path{arXiv:cond-mat/0103301}}, \href
  {https://doi.org/10.1103/PhysRevB.64.064412}
  {\path{doi:10.1103/PhysRevB.64.064412}}.

\bibitem{CheongHenley2004}
Siew-Ann Cheong and Christopher~L. Henley.
\newblock Many-body density matrices for free fermions.
\newblock {\em Phys. Rev. B}, 69:075111, 2004.
\newblock \href {http://arxiv.org/abs/cond-mat/0206196}
  {\path{arXiv:cond-mat/0206196}}, \href
  {https://doi.org/10.1103/PhysRevB.69.075111}
  {\path{doi:10.1103/PhysRevB.69.075111}}.

\bibitem{Peschel2003}
Ingo Peschel.
\newblock Calculation of reduced density matrices from correlation functions.
\newblock {\em J. Phys. A: Math. Gen.}, 36(14):L205, 2003.
\newblock \href {http://arxiv.org/abs/cond-mat/0212631}
  {\path{arXiv:cond-mat/0212631}}, \href
  {https://doi.org/10.1088/0305-4470/36/14/101}
  {\path{doi:10.1088/0305-4470/36/14/101}}.

\bibitem{PeschelEisler2009}
Ingo Peschel and Viktor Eisler.
\newblock Reduced density matrices and entanglement entropy in free lattice
  models.
\newblock {\em J. Phys. A: Math. Theor.}, 42(50):504003, 2009.
\newblock \href {http://arxiv.org/abs/0906.1663} {\path{arXiv:0906.1663}},
  \href {https://doi.org/10.1088/1751-8113/42/50/504003}
  {\path{doi:10.1088/1751-8113/42/50/504003}}.

\bibitem{VidalLatorreRicoKitaev2003}
Guifr\'{e} Vidal, Jos\'{e}~I. Latorre, Enrique Rico, and Alexei Kitaev.
\newblock Entanglement in quantum critical phenomena.
\newblock {\em Phys. Rev. Lett.}, 90:227902, 2003.
\newblock \href {http://arxiv.org/abs/quant-ph/0211074}
  {\path{arXiv:quant-ph/0211074}}, \href
  {https://doi.org/10.1103/PhysRevLett.90.227902}
  {\path{doi:10.1103/PhysRevLett.90.227902}}.

\bibitem{AffleckLudwig1991}
Ian Affleck and Andreas W.~W. Ludwig.
\newblock Universal noninteger ``ground-state degeneracy'' in critical quantum
  systems.
\newblock {\em Phys. Rev. Lett.}, 67:161--164, 1991.
\newblock \href {https://doi.org/10.1103/PhysRevLett.67.161}
  {\path{doi:10.1103/PhysRevLett.67.161}}.

\bibitem{BrehmBrunnerJaudSchmidtColinet2016}
Enrico~M. Brehm, Ilka Brunner, Daniel Jaud, and Cornelius Schmidt-Colinet.
\newblock Entanglement and topological interfaces.
\newblock {\em Fortschr. Phys.}, 64:516--535, 2016.
\newblock \href {http://arxiv.org/abs/1512.05945} {\path{arXiv:1512.05945}},
  \href {https://doi.org/10.1002/prop.201600024}
  {\path{doi:10.1002/prop.201600024}}.

\bibitem{GutperleMiller2016}
Michael Gutperle and John~D. Miller.
\newblock A note on entanglement entropy for topological interfaces in {RCFT}s.
\newblock {\em J. High Energy Phys.}, 04:176, 2016.
\newblock \href {http://arxiv.org/abs/1512.07241} {\path{arXiv:1512.07241}},
  \href {https://doi.org/10.1007/JHEP04(2016)176}
  {\path{doi:10.1007/JHEP04(2016)176}}.

\bibitem{LiebSchultzMattis1961}
Elliott Lieb, Theodore Schultz, and Daniel Mattis.
\newblock Two soluble models of an antiferromagnetic chain.
\newblock {\em Ann. Phys. (N.Y.)}, 16:407--466, 1961.
\newblock \href {https://doi.org/10.1016/0003-4916(61)90115-4}
  {\path{doi:10.1016/0003-4916(61)90115-4}}.

\bibitem{Pfeuty1970}
Pierre Pfeuty.
\newblock The one-dimensional {Ising} model with a transverse field.
\newblock {\em Ann. Phys. (N.Y.)}, 57:79--90, 1970.
\newblock \href {https://doi.org/10.1016/0003-4916(70)90270-8}
  {\path{doi:10.1016/0003-4916(70)90270-8}}.

\bibitem{HauruEtAl2016}
Markus Hauru, Glen Evenbly, Wen~Wei Ho, Davide Gaiotto, and Guifre Vidal.
\newblock Topological conformal defects with tensor networks.
\newblock {\em Phys. Rev. B}, 94(11):115125, 2016.
\newblock \href {http://arxiv.org/abs/1512.03846} {\path{arXiv:1512.03846}},
  \href {https://doi.org/10.1103/PhysRevB.94.115125}
  {\path{doi:10.1103/PhysRevB.94.115125}}.

\bibitem{PollmannEtAl2010}
Frank Pollmann, Ari~M. Turner, Erez Berg, and Masaki Oshikawa.
\newblock Entanglement spectrum of a topological phase in one dimension.
\newblock {\em Physical Review B}, 81:064439, 2010.
\newblock \href {http://arxiv.org/abs/0910.1811} {\path{arXiv:0910.1811}},
  \href {https://doi.org/10.1103/PhysRevB.81.064439}
  {\path{doi:10.1103/PhysRevB.81.064439}}.

\bibitem{ZhaoNortheMeyer2020SymmetryResolved}
Suting Zhao, Christian Northe, and René Meyer.
\newblock Symmetry-resolved entanglement in {AdS$_3$/CFT$_2$} coupled to
  {$U(1)$} {Chern--Simons} theory.
\newblock {\em JHEP}, 07:030, 2021.
\newblock \href {http://arxiv.org/abs/2012.11274} {\path{arXiv:2012.11274}},
  \href {https://doi.org/10.1007/jhep07(2021)030}
  {\path{doi:10.1007/jhep07(2021)030}}.

\bibitem{WeisenbergerEtAl2021SymmetryResolved}
Konstantin Weisenberger, Suting Zhao, Christian Northe, and René Meyer.
\newblock Symmetry-resolved entanglement for excited states and two entangling
  intervals in {AdS$_3$/CFT$_2$}.
\newblock {\em JHEP}, 12:104, 2021.
\newblock \href {http://arxiv.org/abs/2108.09210} {\path{arXiv:2108.09210}},
  \href {https://doi.org/10.1007/jhep12(2021)104}
  {\path{doi:10.1007/jhep12(2021)104}}.

\bibitem{ZhaoEtAl2022ChargedMoments}
Suting Zhao, Christian Northe, Konstantin Weisenberger, and René Meyer.
\newblock Charged moments in {$W_3$} higher spin holography.
\newblock {\em JHEP}, 05:166, 2022.
\newblock \href {http://arxiv.org/abs/2202.11111} {\path{arXiv:2202.11111}},
  \href {https://doi.org/10.1007/jhep05(2022)166}
  {\path{doi:10.1007/jhep05(2022)166}}.

\bibitem{DiGiulioEtAl2022BCFTSymmetryResolved}
Giuseppe Di~Giulio, René Meyer, Christian Northe, Henri Scheppach, and Suting
  Zhao.
\newblock On the boundary conformal field theory approach to symmetry-resolved
  entanglement.
\newblock {\em SciPost Phys. Core}, 6:049, 2023.
\newblock \href {http://arxiv.org/abs/2212.09767} {\path{arXiv:2212.09767}},
  \href {https://doi.org/10.21468/scipostphyscore.6.3.049}
  {\path{doi:10.21468/scipostphyscore.6.3.049}}.

\bibitem{Northe2023ConformalResolution}
Christian Northe.
\newblock Entanglement resolution with respect to conformal symmetry.
\newblock {\em Phys. Rev. Lett.}, 131:151601, 2023.
\newblock \href {http://arxiv.org/abs/2303.07724} {\path{arXiv:2303.07724}},
  \href {https://doi.org/10.1103/physrevlett.131.151601}
  {\path{doi:10.1103/physrevlett.131.151601}}.

\bibitem{Northe2025FermionParity}
Christian Northe.
\newblock Fermion parity resolution of entanglement.
\newblock {\em JHEP}, 12:134, 2025.
\newblock \href {http://arxiv.org/abs/2509.03605} {\path{arXiv:2509.03605}},
  \href {https://doi.org/10.1007/jhep12(2025)134}
  {\path{doi:10.1007/jhep12(2025)134}}.

\bibitem{JinKorepin2004}
B.-Q. Jin and V.~E. Korepin.
\newblock Quantum spin chain, toeplitz determinants and the fisher--hartwig
  conjecture.
\newblock {\em J. Stat. Phys.}, 116:79--95, 2004.
\newblock \href {http://arxiv.org/abs/quant-ph/0304108}
  {\path{arXiv:quant-ph/0304108}}, \href
  {https://doi.org/10.1023/B:JOSS.0000037230.37166.42}
  {\path{doi:10.1023/B:JOSS.0000037230.37166.42}}.

\bibitem{ItsJinKorepin2005}
A.~R. Its, B.-Q. Jin, and V.~E. Korepin.
\newblock Entanglement in the {XY} spin chain.
\newblock {\em J. Phys. A: Math. Gen.}, 38:2975--2990, 2005.
\newblock \href {http://arxiv.org/abs/quant-ph/0409027}
  {\path{arXiv:quant-ph/0409027}}, \href
  {https://doi.org/10.1088/0305-4470/38/13/011}
  {\path{doi:10.1088/0305-4470/38/13/011}}.

\bibitem{Audenaert2007}
Koenraad M.~R. Audenaert.
\newblock A sharp continuity estimate for the von neumann entropy.
\newblock {\em Journal of Physics A: Mathematical and Theoretical},
  40(28):8127--8136, 2007.
\newblock \href {http://arxiv.org/abs/quant-ph/0610146}
  {\path{arXiv:quant-ph/0610146}}, \href
  {https://doi.org/10.1088/1751-8113/40/28/S18}
  {\path{doi:10.1088/1751-8113/40/28/S18}}.

\bibitem{ZhangRuggieroCalabrese2019SubsystemDistances}
Jiaju Zhang, Paola Ruggiero, and Pasquale Calabrese.
\newblock {Subsystem trace distance in low-lying states of $(1+1)$-dimensional
  conformal field theories}.
\newblock {\em JHEP}, 10:181, 2019.
\newblock \href {http://arxiv.org/abs/1907.04332} {\path{arXiv:1907.04332}},
  \href {https://doi.org/10.1007/jhep10(2019)181}
  {\path{doi:10.1007/jhep10(2019)181}}.

\bibitem{DeiftItsKrasovsky2011}
Percy Deift, Alexander Its, and Igor Krasovsky.
\newblock Asymptotics of toeplitz, hankel, and toeplitz+hankel determinants
  with fisher--hartwig singularities.
\newblock {\em Annals of Mathematics}, 174(2):1243--1299, 2011.
\newblock \href {http://arxiv.org/abs/0905.0443} {\path{arXiv:0905.0443}},
  \href {https://doi.org/10.4007/annals.2011.174.2.12}
  {\path{doi:10.4007/annals.2011.174.2.12}}.

\bibitem{HeEtAl2024SuperpositionTransitionRDM}
Song He, Yu-Xuan Zhang, Long Zhao, and Zi-Xuan Zhao.
\newblock {Entanglement and Pseudo Entanglement Dynamics versus Fusion in CFT}.
\newblock {\em JHEP}, 06:177, 2024.
\newblock \href {http://arxiv.org/abs/2312.02679} {\path{arXiv:2312.02679}},
  \href {https://doi.org/10.1007/jhep06(2024)177}
  {\path{doi:10.1007/jhep06(2024)177}}.

\bibitem{OhmoriTachikawa2015EntanglingSurface}
Kantaro Ohmori and Yuji Tachikawa.
\newblock {Physics at the entangling surface}.
\newblock {\em Journal of Statistical Mechanics: Theory and Experiment},
  2015(4):P04010, 2015.
\newblock \href {http://arxiv.org/abs/1406.4167} {\path{arXiv:1406.4167}},
  \href {https://doi.org/10.1088/1742-5468/2015/04/p04010}
  {\path{doi:10.1088/1742-5468/2015/04/p04010}}.

\bibitem{LiHaldane2008}
Hui Li and F.~D.~M. Haldane.
\newblock Entanglement spectrum as a generalization of entanglement entropy:
  Identification of topological order in non-abelian fractional quantum hall
  effect states.
\newblock {\em Phys. Rev. Lett.}, 101:010504, 2008.
\newblock \href {http://arxiv.org/abs/0805.0332} {\path{arXiv:0805.0332}},
  \href {https://doi.org/10.1103/PhysRevLett.101.010504}
  {\path{doi:10.1103/PhysRevLett.101.010504}}.

\end{thebibliography}
\end{document}